\documentclass[12pt, reqno]{amsart}
\usepackage{amssymb, amsmath, url,color,  amsthm, amssymb,graphicx,multirow, appendix}
\usepackage{amsmath}
\usepackage{amssymb}

\usepackage[margin=1in]{geometry}
\usepackage{bigints}
\usepackage{lscape,amssymb, amsmath,verbatim,graphicx}
\usepackage{epsfig,subfigure,float,subfigure}
\usepackage{graphicx}
\usepackage{epsfig}
\usepackage{threeparttable}
\usepackage{ragged2e}
\usepackage{multirow}
\usepackage{adjustbox}

\usepackage[cal=boondox]{mathalfa}
\DeclareFontFamily{OT1}{pzc}{}
\DeclareFontShape{OT1}{pzc}{m}{it}{<-> s * [1.10] pzcmi7t}{}
\DeclareMathAlphabet{\mathpzc}{OT1}{pzc}{m}{it}
\def\beq{\begin{equation}}
\def\eeq{\end{equation}}

\let \cite = \citep
\newtheorem{assumption}{Assumption}[section]

\usepackage{mathrsfs}
\usepackage{dsfont}

\newcommand{\lf}{\lfloor}
\newcommand{\rf}{\rfloor}

\newtheorem{theorem}{Theorem}[section]
\newtheorem{lemma}{Lemma}[section]

\newtheorem{corollary}{Corollary}[section]
\newtheorem{Remark}{Remark}[section]

\newcommand{\bDelta}{\boldsymbol \Delta}

\newcommand{\cS}{\mathcal S}

\newcommand\T{\top}
\newcommand{\bX}{{\bf X}}

\newcommand{\bB}{{\bf B}}
\newcommand{\bE}{{\bf E}}
\newcommand{\bG}{{\bf G}}
\newcommand{\bz}{{\bf z}}
\newcommand{\bY}{{\bf Y}}
\newcommand{\bZ}{{\bf Z}}

\newcommand{\bA}{{\bf A}}
\newcommand{\bC}{{\bf C}}

\newcommand{\bD}{{\bf D}}

\newcommand{\bv}{{\bf v}}
\newcommand{\bg}{{\bf g}}

\newcommand{\bU}{{\bf U}}
\newcommand{\bu}{{\bf u}}
\newcommand{\bQ}{{\bf Q}}

\newcommand{\cA}{\mathcal {A}}
\newcommand\bLamb{\mbox{\boldmath${\Lambda}$}}
\newcommand{\cz}{\mathpzc z}

\newcommand{\cD}{{\mathcal  D}}

\newcommand{\bGam}{\boldsymbol \Gamma}

\newcommand{\bW}{{\bf W}}

\newcommand{\bx}{{\bf x}}
\newcommand{\bbe}{\boldsymbol \beta}
\newcommand\bga{\mbox{\boldmath${\gamma}$}}
\newcommand{\eps}{\epsilon}

\newcommand{\bR}{\boldsymbol R}

\newcommand{\cW}{\mathcal {W}}
\newcommand{\bK}{{\bf K}}
\newcommand{\bPsi}{{\boldsymbol \Psi}}
\newcommand{\ch}{{\mathcal h}}

\newcommand{\bP}{{\bf P}}
\newcommand\bGamma{\mbox{\boldmath${\Gamma}$}}

\newcommand{\bH}{{\bf  {H}}}
\newcommand{\mca}{\mathcal {a}}
\newcommand{\bUp}{\mbox{\boldmath $ \Upsilon$}}

\numberwithin{equation}{section}
\numberwithin{Remark}{section}
\numberwithin{Assumption}{section}
\numberwithin{lemma}{section}
\numberwithin{table}{section}
\numberwithin{figure}{section}

\usepackage{setspace}
\usepackage[pdftex,colorlinks=true]{hyperref}

\usepackage[width=1\textwidth]{caption}

\begin{document}

\title[Detecting multiple changes in linear models]{Detecting multiple change points in linear models with heteroscedasticity }

\author{Lajos Horv{\'a}th$^{\ddagger}$}
\author{Gregory Rice$^{\mathsection}$}
\author{Yuqian Zhao$^{\dagger}$}
\thanks{$^{\ddagger}$ Department of Mathematics, University of Utah, Salt Lake City, UT 84112--0090 USA. Email: horvath@math.utah.edu.}
\thanks{$^{\mathsection}$ Department of Statistics and Actuarial Science, University of Waterloo, Waterloo, N2L 3G1, Canada. Email: grice@uwaterloo.ca. The research of Gregory Rice was supported by the Natural Science and Engineering Research Council of Canada’s Discovery and Accelerator grants.}
\thanks{$^{\dagger}$ University of Sussex Business School, University of Sussex, Brighton, BN1 9RH,  United Kingdom. Email: yuqian.zhao@sussex.ac.uk.}

\bigskip
\begin{abstract}
	The problem of detecting change points  in the parameters of a linear regression model with errors and covariates exhibiting heteroscedasticity is considered. Asymptotic results for weighted functionals of the cumulative sum (CUSUM) processes of model residuals are established when the model errors are weakly dependent and non-stationary, allowing for either abrupt or smooth changes in their variance. These theoretical results illuminate how to adapt standard change point test statistics for linear models to this setting. 
    We studied such adapted change-point tests in simulation experiments, along with a finite sample adjustment to the proposed testing procedures. The results suggest that these methods perform well in practice for detecting multiple change points in the linear model parameters and controlling the Type I error rate in the presence of heteroscedasticity.
    We illustrate the use of these approaches in applications to test for instability in predictive regression models and explanatory asset pricing models.
\end{abstract}

\maketitle

\newpage

\section{Introduction}\label{intro}

Linear models are widely used for causal inference and out-of-sample prediction problems with time series data, including macroeconomic forecasting, asset pricing, and portfolio optimization. For example, Stock and Watson (2002) identify predictive factors for key macroeconomic variables using linear regressions. In finance, a prominent application is the prediction of equity premia using financial and economic variables, as examined by Welch and Goyal (2008).

A critical challenge for such models in the time series setting is that their coefficients often appear to undergo structural changes due to shocks such as policy shifts, technological advances, or evolving consumer and investor behaviour. Model instability can undermine both in-sample fit and out-of-sample performance. Detecting change points in linear models is hence often a critical first step toward using them in practice. Most existing detection methods assume stationary and homoscedastic error terms; see Chapter 4 of Horv\'ath and Rice (2024), Chapter 4 of Chen and Gupta (2014), and Niu et al. (2016) for a review of change point detection methods for linear models. The assumption of homoscedasticity often appears to be implausible in practice, as model residuals frequently exhibit heteroscedasticity as well as changes in their distribution coinciding with other changes in the model parameters. This paper focuses on adapting stability tests for linear models to accommodate heteroscedastic covariates and errors.

The effect of heteroscedasticity in change point analysis has drawn increasing attention recently. Zhou (2013) and Xu (2015) advise that commonly used CUSUM-based change point procedures can become over-sized and unreliable in the presence of change points in the variance of the error process. To deal with this issue, several methods have been proposed to adapt limits for classical CUSUM-type statistics under heteroscedasticity. In the setting of changes in the mean of scalar time series, Zhou (2013) suggests a wild-bootstrap procedure to estimate the limiting distribution. Astill et al. (2023) develop a CUSUM based monitoring scheme for financial data allowing for time varying volatility. Xu (2015) builds a time transformed Wiener process, and G\'{o}recki et al. (2018) make use of Karhunen--Lo\'{e}ve expansions to characterize the limit of statistics based on heteroscedastic observations. Horv\'ath et al. (2021) derive a Wiener process-based limit for heavily-weighted CUSUM processes constructed from linear model residuals. Georgiev et al. (2018) consider the change point detection problem in predictive regression models allowing for non-stationary covariates.

In this paper, we consider a linear regression model for a scalar response $y_i$ on a $d$-dimensional covariate $\bx_i$, with $R$ possible changes:
\begin{align}\label{limu1}
	y_i= \sum_{r=1}^{R+1} \bx_i^\T\bbe_r\mathds{1}\{k_{r-1}+1 \le i \le k_{r }\} +\eps_i, \;\; k_0=0, \mbox{ and } k_{R+1}=N,
\end{align}
where $(\bx_1,y_1),...,(\bx_N,y_N)$ are the observed data, $\bx_i \in \mathbb{R}^d$ and $y_i \in \mathbb{R}$. The regression parameter changes from $\bbe_{\ell}$ to $\bbe_{\ell+1}$ at the potential change points $k_1, \ldots, k_R$. When for example the covariates contain lagged values of an exogenous series or the response, \eqref{limu1} becomes a predictive regression model with changing coefficients. We are interested in testing the null hypothesis that the regression parameter remains constant over the sample period:
\beq\label{limul1}
H_0:\;\; \bbe_1 =  \cdots  = \bbe_{R+1},
\eeq
versus the alternative hypothesis that there exists at least one change point,
\beq\label{limul2}
H_A:\; \bbe_i \ne \bbe_{i+1} \mbox{ for some } i \in \{1,...,R\}.
\eeq

Under $H_0$ we denote the common regression parameter as $\bbe_0$, which can be estimated by the least squares estimator
$$
\hat{\bbe}_{N}=\left( \bX_{N}^\T\bX_{N}\right)^{-1}\bX_{N}^\T\bY_{N},
$$
where $\bY_{N}=(y_1,  \ldots, y_N)^\T$ is the vector containing the responses, and the design matrix is given by $\bX_N = (\bx_1 \mid \cdots \mid \bx_N)^\top$. Thus, the linear model residuals are computed as
$$
\hat{\eps}_i=y_i-\bx_i^\T\hat{\bbe}_N, \quad 1\leq i \leq N.
$$
The maximally selected F--tests of ${H}_0$ may be expressed as functionals of the standard CUSUM process of the covariate weighted residuals
\begin{align}\label{zz-def}
\bZ_N(t)=N^{-1/2}\left( \sum_{i=1}^{\lf (N+1)t \rf}\bx_i\hat{\eps}_i-\frac{\lf (N+1)t \rf}{N} \sum_{i=1}^N\bx_i\hat{\eps}_i    \right), \;\;\;0< t < 1,
\end{align}
where $\underset{\emptyset}{\sum}=0$. It follows that $\bZ_N(t)=\mathbf{0}$, if $t\in [0,1/(N+1))$ and $t\in(N/(N+1),1]$. Most existing methods to test for change points in linear models make use of functionals of $\bZ_N$. The likelihood ratio based tests proposed by Bai (1995, 1997a,b, 1998) and Bai and Perron (1998, 2003) are asymptotically equivalent with functionals of $\bZ_N$. For example, the likelihood ratio based method in Bai and Perron (1998) can be written as the maximum of the standardized increments of the process $\bZ_N$. Similarly, Bai (1999) develops a maximally selected least squares test to determine whether $R$ or $R+1$ changes are present in model \eqref{limu1}. It is also asymptotically equivalent with a functional of $\bZ_N$. Hidalgo and Seo (2013) considers Lagrange multiplier and maximum likelihood statistics, respectively, for testing the constancy of parameters in parametric time series models, which are also asymptotically equivalent to functionals of $\bZ_N$ under model \eqref{limu1}.

We provide in this paper a comprehensive asymptotic analysis under $H_0$ of the weighted functionals of $\bZ_N$ allowing for quite general forms of heteroscedasticity in the covariates and the errors in \eqref{limu1}. In particular, we consider a model for non-stationary errors allowing for both {\it smooth} and {\it abrupt} changes in the error variance. An interesting consequence of the results presented is that asymptotics for the CUSUM process of the unobservable series $\{\bx_i\eps_i\}$ and for the observable series $\{\bx_i\hat{\eps}_i\}$ are the same with homoscedastic covariates/errors, although this does not remain true in heteroscedastic scenarios. If the volatility of the covariates $\bx_i$ changes during the observation period, then the asymptotic distribution of the weighted CUSUM is affected by the estimation of the regression parameter. However, the asymptotic results for suitably standardized CUSUM statistics will, interestingly, still satisfy Darling--Erd\H{o}s type limit results in this case. The behaviour of these statistics under $H_A$ is also detailed, and it is shown that the typically weighted functionals of $\bZ_N$ are consistent in detecting multiple change points. 

The finite sample performances of the proposed tests are compared and studied in a Monte Carlo simulation study, which supports that the adaptations proposed to handle heteroscedasticity of the errors and to improve finite sample performance work well in practice and outperform the existing approaches of Xu (2015), Perron et al. (2020) and Horv\'ath et al. (2021). We then illustrate the proposed methods through an application to testing model instability in macroeconomic variables and equity return prediction models.

The rest of the article is organized as follows. In Section \ref{sec-hetero}, we detail the asymptotic theory for several commonly used functionals of $\bZ_N$. Section \ref{sec-smooth} extends the results for a model with more generally non-stationary errors. Section \ref{se-simulation} details the computation of critical values and assesses the finite-sample performance of the proposed tests through Monte Carlo simulations, comparing them with existing methods. Data applications are given in Section \ref{sec-app}, and Section \ref{concl} concludes with some remarks.

\section{Abrupt changes in the variance model}\label{sec-hetero}
Let $\{\bz_i=(\bx_i^\T, \eps_i)^\T, -\infty<i<\infty\}$ denote the process describing the covariates and error terms in \eqref{limu1}. We first consider the case in which the second order properties of $\bz_i=(\bx_i, \eps_i)^\T, 1\leq i \leq N$  may change $M$ times during the observation period, where $1 < m_1<m_2<\ldots<m_M < N$ denote the times at which the second order properties of the $\bz_i$'s might change. We assume
\begin{assumption}\label{as-m}\;$m_i=\lf N\tau_i\rf$ and $0<\tau_1<\tau_2<\ldots <\tau_M<1$.
\end{assumption}
\noindent
Leybourne et al.\ (2006), Pein et al.\ (2017) and Horv\'ath et al.\ (2021) introduce heteroscedastic models, similar to Assumption \ref{as-m}, in change point analysis. We use a  decomposable Bernoulli shift model for each segment of stationarity. Let $m_0=0, m_{M+1}=N$ and correspondingly $\tau_0=0$ and $\tau_{M+1}=1$. We note that the unknown times $m_i$ may or may not coincide with the change point locations $k_j$.

{Below $\left\| \cdot \right\|$ denotes the Euclidean norm.}
\begin{assumption}\label{as-lin-ber-v} \;$\bz_i=\bg_\ell(\eta_i, \eta_{i-1}, \ldots), m_{\ell-1}<i\leq m_{\ell}, 1\leq \ell \leq M+1$, where $\bg_\ell$ are  non--random measurable functions, $\cS^\infty\to \mathbb{R}^{d+1}$,  $E\|\bz_i\|^\nu<\infty$
	with some $\nu>4$, $\{\eta_i, -\infty<i<\infty\}$ are independent and identically distributed random variables with values in a measurable space $\cS$,
	$$
	\left( E\left\|\bz_i-\bz^*_{i,j}\right\|^\nu\right)^{1/\nu}\leq cj^{-\alpha}\quad\mbox{with some}\;\;c>0\;\;\mbox{and}\;\;\alpha>2,
	$$
	$\bz_{i,j}^*=\bg_\ell(\eta_i, \ldots, \eta_{i-j+1}, \eta^*_{i-j}, \eta_{i-j-1}^*, \ldots)$, $m_{\ell-1}<i\leq m_{\ell}, 1\leq \ell \leq M+1$, $\{\eta^*_{\ell}, -\infty<\ell <\infty\}$ are independent, identically distributed copies of $\eta_0$, independent of $\{\eta_j, -\infty<j<\infty\}$.
\end{assumption}
\noindent

Under Assumption \ref{as-lin-ber-v}, the errors and covariates are not stationary over the whole observation period, but are drawn from a stationary process on the sub-segments $(m_{\ell-1}, m_\ell]$, $1 \leq \ell \leq M+1$. Conventionally, the error terms $\eps_i$ are independent of $\bx_i$, and are homoscedastic in the sense that the variance of the conditional distribution of $\eps_i$ given $\bx_i$ remains constant with respect to $i$. This condition might not hold under Assumption \ref{as-lin-ber-v}, resulting in a heteroscedastic model. In Section~\ref{sec-smooth}, we further relax this assumption to allow for non-stationarity within each sub-segment.

In this section, we aim to establish the asymptotic behaviour of $\bZ_N(t)$, as defined in \eqref{zz-def}, under the null hypothesis of no change in the regression parameter under Assumption \ref{as-lin-ber-v}.

In order to identify the regression parameters, we require
\begin{assumption}\label{weps-mu}\;$E\bx_{m_i}\eps_{m_i}=\bf0$, $1\leq i \leq M+1$.
\end{assumption}
\noindent
Assumption \ref{weps-mu} postulates that the identification of the regression parameters holds on all sub-segments of stationarity. To state the weak limit of the process $\bZ_N$, we also need to introduce $M+1$ long run covariance matrices reflecting the changing  covariances between stationary subintervals. Let
\beq\label{longd}
\bD_\ell=\lim_{N\to\infty}\frac{1}{m_\ell-m_{\ell-1}}E\left( \sum_{i=m_{\ell-1}+1}^{m_{\ell}}  \bx_i\eps_i \right)\left( \sum_{i=m_{\ell-1}+1}^{m_{\ell}}  \bx_i\eps_i \right)^\top ,\quad 1\leq \ell\leq M+1.
\eeq
We show that these matrices are well defined under Assumption \ref{as-lin-ber-v}. We define the process
$\{\bGam(t), 0\leq t \leq 1\}$ as
\beq\label{gaus}
\bGam(t)=\sum_{\ell=1}^j \bW_{\bD_\ell}(\tau_\ell-\tau_{\ell-1})+\bW_{\bD_{j+1}}(t-\tau_j), \quad \mbox{if}\;\;\;\tau_{j}<t\leq \tau_{j+1},\;\;1\leq j\leq M,
\eeq
where $\{ \bW_{\bD_j}(t), \; t \ge 0\}, 1\leq j\leq M+1$, are independent $d$ dimensional Brownian motions such that $E\bW_{\bD_j}(t)=\bf0$ and
$E\bW_{\bD_j}(t)\bW_{\bD_j}^\T(s)=\min(t,s)\bD_j$. The process $\bGam(t)\in \mathbb{R}^d$ is Gaussian process with $E\bGam(t)=\bf0$,
\begin{align}\label{gammacov}
E\bGam(t)\bGam^\T(s)=\bG(\min(t,s)),
\end{align}
and
\begin{align*}
\bG(u)=\sum_{\ell=1}^j{\bD_\ell}(\tau_\ell-\tau_{\ell-1})+\bD_{j+1}(u-\tau_j),\quad \mbox{if}\;\;\;\tau_{j}<u\leq \tau_{j+1},\;\;0\leq j\leq M.
\end{align*}
As it is common in the theory of linear regression, we require
\begin{assumption}\label{asminons}\;
$\bA_i=E\bx_{m_i}\bx_{m_i}^\T$  is a non--singular matrix for some $1\leq i \leq M+1$.
\end{assumption}
\noindent
Note that under Assumption \ref{as-lin-ber-v}, the matrices $\bA_i$ are well defined. Let
\beq\label{barga}
\bar{\bGam}(t)=\bGam(t)-t\bGam(1) -\bv(t) \left( \sum_{\ell=1}^{M+1} (\tau_\ell - \tau_{\ell-1}) \bA_\ell \right)^{-1} \sum_{\ell=1}^{M+1} \bW_{D_\ell} (\tau_\ell-\tau_{\ell-1}),\quad 0\leq t \leq 1,
\eeq
where
\beq\label{13a}
\bv(t) = \sum_{\ell=1}^{k-1}(\tau_\ell-\tau_{\ell-1}) \bA_\ell + (t-\tau_{k-1}) \bA_k - t \sum_{\ell=1}^{M+1} (\tau_\ell - \tau_{\ell-1}) \bA_\ell,
\eeq
for $\tau_{k-1}< t \leq \tau_k, 1\leq k \leq M+1$.

One often considers weighted functionals of $\bZ_N$ to improve the power of tests, particularly when changes occur near the ends of the sample. In our context, we apply the weight function $w(t)$ satisfying the following properties:
\begin{assumption}\label{as-wc-1}\;(i) $\inf_{\delta\leq t\leq 1-\delta}w(t)>0$, for all $0<\delta<1/2$,	(ii) $w(t)$ is non--decreasing in a neighbourhood of 0, and (iii) $w(t)$ is non-increasing in a neighbourhood of 1.
\end{assumption}
\noindent
Due to using the weight function $w$, we need an integral condition for the existence of a non-degenerate limit distribution of the weighted process ${\bf Z}_N$. Let
\beq\label{defiwc}
I(w,c)=\int_0^1 \frac{1}{t(1-t)}\exp\left( -\frac{cw^2(t)}{t(1-t)}   \right)dt.
\eeq
The integral $I(w,c)$ characterizes the upper and lower classes for the Brownian bridge at 0 and 1 (It\^o and McKean, 1965; O'Reilly, 1974). According to It\^{o} and McKean (1965), we know that $I(w, c)< \infty$ with some $c$ if and only if
\begin{equation}\label{eq-lw}
\underset{0<t<1}{\sup}\frac{ |B(t)|}{w(t)}<\infty \quad \mbox{ a.s.},
\end{equation}
where $\{B(t), 0\leq t \leq 1\} $ is a Brownian bridge. The most commonly used weight function satisfying $I(w,c)<\infty$ and Assumption \ref{as-wc-1} is $w(t)=[t(1-t)]^\kappa$, $0\leq \kappa <1/2$. Since $I([t(1-t)]^{1/2},c)=\infty$ for all $c>0$, \eqref{eq-lw} cannot hold with $w(t)=[t(1-t)]^{1/2}$, so in this case we have a different limit distribution. We note that applying the weight function $w(t) = [t(1 - t)]^{1/2}$ can lead to a slow convergence rate. We refer to Cs\"org\H{o} and Horv\'ath (1993) for more details and discussions on weighted empirical and Gaussian processes.

\begin{theorem}\label{lin-var-1} We assume that $H_0$, Assumptions \ref{as-m}--\ref{as-wc-1} are satisfied. If $I(w,c)<\infty$ with some $c>0$, then
	$$
	V^{HET}_N(\kappa) \equiv \sup_{0<t<1}\frac{1}{w(t)}\left\| \bZ_N(t)  \right\|\stackrel{\cD}{\to}\sup_{0<t<1}\frac{1}{w(t)}\left\| \bar{\bGam}(t)  \right\|,
	$$
	and
	$$
	Q^{HET}_N(\kappa) \equiv  \sup_{0<t<1}\frac{1}{w(t)}\left\| \bZ_N(t)  \right\|_\infty\stackrel{\cD}{\to}\sup_{0<t<1}\frac{1}{w(t)}\left\| \bar{\bGam}(t)  \right\|_\infty.
	$$
    where {  $\left\| \cdot \right\|_\infty$ denotes the maximum norm,} and $\bar{\bGam}(t) = \bGam(t)- t\bGam(1) - \bu(t)\bGam(1)$, with $\bu(t)=\bv(t)\left( \sum_{\ell=1}^{M+1}(\tau_\ell-\tau_{\ell-1})\bA_\ell \right)^{-1}$. 
\end{theorem}
\noindent
We thus have $E\bar{\bGam}(t)=\bf0$ and
\begin{equation}\label{barbgc}
\begin{split}
E\bar{\bGam}(t)\bar{\bGam}^\T(s) \equiv \bar{\bG}(t,s) = & \bG(\min(t,s))-t\bG(s)-\bu(t)\bG(s)-s\bG(t)+st\bG(1)\\
& + \bu(t)s\bG(1) -\bG(t)\bu^\T(s) + t\bG(1)\bu^\T(s) + \bu(t)\bG(1)\bu^\T(s).
\end{split}
\end{equation}
We note if $\{\bx_i, -\infty <i<\infty \}$ is stationary, then $\bA_1=\dots=\bA_{M+1}$ and therefore $\bu(t)=\mathbf{O}$. In this case
\beq\label{barbg}
E\bar{\bGam}(t)\bar{\bGam}^\T(s) \equiv \bar{\bG}(t,s) =\bG(\min(t,s))-t\bG(s)-s\bG(t)+st\bG(1).
\eeq
{ It is important to note that due to the estimation of $\bbe$, the weighted CUSUM's of $\bx_i\eps_i$ and $\bx_i\hat{\eps}_i$, $1\leq i \leq N$, have a qualitatively different asymptotic distribution in heteroscedastic models. If for instance $\bA_1 = \dots = \bA_{M+1}$, i.e., the sequence $\{{\bf x}_i, -\infty <i<\infty\}$ is homoscedastic, and the heteroscedasticity is only in the errors $\{\eps_i, 1\leq i \leq N\}$, then $\bv(t)=\mathbf{O}$, so the effect of estimating $\bbe$ does not appear in the limit distribution. In other words, interestingly the limit behaviour of the CUSUM of $\bx_i\eps_i$'s and $\bx_i\hat{\eps}_i$'s are  the same if $\{{\bf x}_i, -\infty <i<\infty\}$ is homoscedastic even if $\{\eps_i, -\infty <i<\infty\}$ is not.}

We now turn our focus to the standardized statistics using the weight function $w(t)=[t(1-t)]^{1/2}$, and we use the following assumption
\begin{assumption}\label{asvarnosd}\;$\bD_1$ and $\bD_{M+1}$ are non--singular matrices.
\end{assumption}
\noindent
Assumption \ref{asvarnosd} is a sufficient condition to obtain limit results for suitable standardized weighted supremum functionals of $\bZ_N$. The proofs in Appendix \ref{proof-2} show that under $H_0$ the maximum of such functionals is asymptotically reached on the intervals $(0, \tau_1]$  or $(\tau_M, 1]$, and therefore the maximum taken on these intervals determines the limit distribution. We can standardize the CUSUM process with the matrix valued function
$$
\tilde{\bG}(t)=E\bar{\bGam}(t)\bar{\bGam}^\T(t),\quad 0<t<1.
$$
Let us consider the following ``Darling--Erd\H{o}s" type statistics, which converge weakly to extreme value laws.
\begin{theorem}\label{coverd} If $H_0$ and Assumptions \ref{as-m}--\ref{asvarnosd} hold, then
	\begin{align*}
	\lim_{N\to \infty}P\Biggl\{a(\log N)\sup_{0< t<1} \left(\bZ_N^\T(t)\tilde{\bG}^{-1}(t)\bZ_N(t)\right)^{1/2}\leq x +b_d(\log N)
	\Biggl\}=\exp(-2e^{-x})
	\end{align*}
	and
	\begin{align*}
	\lim_{N\to \infty}P\Biggl\{a(\log N)\sup_{0< t<1} \left\|\tilde{\bG}^{-1/2}(t)\bZ_N(t)\right\|_\infty\leq x +b_1(\log N)
	\Biggl\}=\exp(-2de^{-x})
	\end{align*}
	for all $x$, where $a(x)=(2\log x)^{1/2}$, and
	$$
	b_d(x)=2\log x+\frac{d}{2}\log \log x-\log \Gamma(d/2)
	$$
	and $\Gamma(x)$ is the Gamma function.
\end{theorem}
\noindent

We note the weight function $1/w(t)$ is not explicitly used in the above statistics, as the statistics have been inherently weighted through the standardization term $\tilde{\bG}(t)$.   
\begin{Remark}\label{remark}{\rm 
The limit results in Theorem~\ref{lin-var-1} differ under the homoscedastic model, as heteroscedasticity alters the limiting distribution of functionals of $\bZ_N$. In contrast, the standardized statistics in Theorem~\ref{coverd} are invariant under homoscedasticity or heteroscedasticity. }
\end{Remark}  

The matrix valued function $\tilde{\bG}(t)$ is unknown, and in practice must be estimated from the sample. In order to do so, first we estimate $\bG(u)$
with a long run variance kernel estimator based on a fraction $u$ of the data, where $0<u\leq 1$. Here we use the standard kernel covariance estimator for the weighted residuals $\{ \bx_i\hat{\eps}_i, 1\leq i \leq N \}$. Considering the kernel function $K$ and a bandwidth parameter $h=h(N)$, we require
\begin{assumption}\label{as-k} (i) $K(0)=1$,  (ii) $K(u)=K(-u)$,  (iii) there is $c>0$ such that $K(u)=0$, if $u\not\in [-c, c]$, (iv) $\sup_{-c<u<c}|K(u)|<\infty$, (v) $K(u)$ is Lipschitz continuous on the real line, (iv) $h=h(N)\to \infty$ and $h/N\to 0$, and (v) there exists $\rho$ satisfying $\alpha -1 > \rho\ge 1$, where $\alpha$ is defined in Assumption \ref{as-lin-ber-v}, so that $ 0 < \lim_{x\to 0} [1-K(x)]/|x|^\rho < \infty$.
\end{assumption}
\noindent
The parameter $\rho$ in Assumption \ref{as-k}(v) indicates the order of the kernel near zero, which also approximates the asymptotic bias of kernel--based long run variance estimators. For example, the popular Bartlett kernel has order $\rho=1$ and the Parzen kernel has order $\rho= 2$ (see e.g. Andrews, 1991). We then have the long run covariance matrix estimator $\bG_N(u)$ that is computed from the weighted residuals $\{ \bx_i\hat{\eps}_i, 1\leq i \leq \lf Nu\rf\}$. Let
\begin{align*}
\hat{\bga}_{k,\ell}=\left\{
\begin{array}{ll}
\displaystyle \frac{1}{N}\sum_{i=1}^{k-\ell}\bx_i\hat{\eps}_i(\bx_{i+\ell}\hat{\eps}_{i+\ell})^\T,\quad\mbox{if}\;\;\;
\ell\geq 0,
\vspace{.3cm}\\
\displaystyle\frac{1}{N}\sum_{i=-(\ell-1)}^{k}\bx_i\hat{\eps}_i(\bx_{i+\ell}\hat{\eps}_{i+\ell})^\T,\quad\mbox{if}\;\;\;\ell< 0.
\end{array}
\right.
\end{align*}
Now $\{\bG_N(u), 0\leq u \leq 1\}$ is defined as
\begin{align}\label{hatgdef}
\bG_N(u)=\sum_{\ell=-(\lf Nu\rf-1)}^{\lf Nu\rf -1}K\left( \frac{\ell}{h} \right)\hat{\bga}_{\lf Nu\rf,\ell},
\end{align}
where $K$ and $h$ satisfy Assumption \ref{as-k}.

To approximate the limit in Theorem \ref{lin-var-1}, and according to \eqref{barbgc}, we need to estimate $\bu(t)$. We use
\begin{align}\label{unt}
\bu_N(t)=\left(   \sum_{i=1}^{\lf (N+1)t \rf} \bx_i\bx_i^\T - \frac{\lf (N+1)t \rf}{N} \sum_{i=1}^{N} \bx_i\bx_i^\T    \right) (\bX_N^\T \bX_N)^{-1}.
\end{align}

\begin{theorem}\label{hetcos} If $H_0$,
	Assumptions \ref{as-m}, \ref{as-lin-ber-v} (with $\nu\geq 8$), \ref{asminons} and \ref{as-k} hold, and \\$h/N^{1/3-2/(3\nu)}\rightarrow 0$, then
	$$
	\underset{0<t<1}{\sup} \|\bG_N(t)-\bG(t)\|=o_P(1), \quad
	\underset{0<t<1}{\sup} \|\bu_N(t)-\bu(t)\|=O_P(N^{-1/2}).
	$$
\end{theorem}
Then, we can estimate $\bar{\bG}(t,s)$ to approximate the limit in Theorem \ref{lin-var-1} using the plug-in estimator by replacing $\bG(t)$ with $\bG_N(t)$ and $\bu(t)$ with $\bu_N(t)$ in \eqref{barbgc}. For $\tilde{\bG}(t)$ in Theorem \ref{coverd}, we can use the plug--in estimator
\begin{equation}\label{plugin}
\begin{split}
\tilde{\bG}_N(t) =\bG_N(t)-2t\bG_N(t)+t^2\bG_N(1).
\end{split}
\end{equation} 
The consistency of $\tilde{\bG}_N(t)$ follows from Theorem \ref{hetcos}. 

We now turn to establishing the behavior of the functionals of $\bZ_n$ under the alternative hypothesis. To do so, we further assume that
\begin{assumption}\label{alt1} \;$k_\ell=\lf N\theta_\ell\rf, 1\leq \ell\leq R$, where $\theta_0=0 < \theta_1 < \cdots < \theta_R < 1=\theta_{R+1}$.
\end{assumption}
\noindent
Under Assumption \ref{alt1}, the potential change points are well separated. Also, changes in the parameters can degenerate along with the sample size increasing, i.e.,\  it is possible to have $\|\bbe_\ell-\bbe_{\ell-1}\|=o(1)$.

Let
\beq\label{betass}
\bbe^{**} = \left( \sum_{j=1}^{M+1}(\theta_j-\theta_{j-1}) \bA_j \right)^{-1} \sum_{\ell=1}^{R+1} \left( \sum_{j=1}^{M+1} \bA_j|(\tau_{\ell-1},\tau_\ell]\cap(\theta_j - \theta_{j-1})| \right) \bbe_\ell,
\eeq
where $| \cdot | $ is the Lebesgue measure. Next we define
\beq\label{bargdef}
\begin{split}
	\bar{\bg}_N(t) = & \sum_{j=1}^{k-1} \sum_{\ell=1}^{M+1} \bA_\ell |(\theta_{j-1},\theta_j]\cap (\tau_{\ell-1},\tau_\ell]|(\bbe_j - \bbe^{**}) \\
	&  -t\sum_{j=1}^{R+1}  \sum_{\ell=1}^{M+1}\bA_\ell |(\theta_{k-1},t]\cap (\tau_{\ell-1},\tau_\ell]|(\bbe_j - \bbe^{**}),
\end{split}
\eeq
if $\theta_{k-1}<t\leq \theta_k, 1\leq k \leq R+1$.

If the changes in the volatility occur at the same time as in the linear model coefficients, i.e., $M=R$, and $r_\ell=m_\ell$ for all $1\leq \ell \leq R$, then the formulas for $\bbe^{**}$ on $\bar{\bg}_N(t)$ are simpler. In this case
$$
\bbe^{**}=\left( \sum_{j=1}^{R+1} (\theta_j-\theta_{j-1}) \bA_j \right)^{-1} \sum_{\ell=1}^{R+1}(\theta_\ell - \theta_{\ell-1}) \bA_\ell \bbe_\ell,
$$
and
$$
\bar{\bg}_N(t) = \sum_{j=1}^{k-1} (\theta_j - \theta_{j-1})\bA_j(\bbe_j-\bbe_j^{**})- t\sum_{\ell=1}^{R+1}(\theta_\ell-\theta_{\ell-1})\bA_\ell(\bbe_\ell - \bbe^{**}),
$$
for $\theta_{k-1}<t\leq \theta_k$, $1\leq k \leq R+1$.

\begin{theorem}\label{thebe0a} We assume that $H_A$, Assumptions \ref{as-m}--\ref{asminons} and \ref{alt1} hold. \\
	(i) If, in addition, Assumptions \ref{as-wc-1} and
	\begin{align}\label{thebe1a}
	N^{1/2}\|\bar{\bg}_N(t)\|\to\infty\;\;\;\mbox{for some}\;\;\;0<t<1
	\end{align}
	are satisfied, then
	$$
	\sup_{0<t<1}\frac{1}{w(t)}\left\|\bZ_N(t)-N^{1/2}\bar{\bg}_N(t)  \right\|=O_P(1).
	$$
	(ii) If, in addition,
	\begin{align}\label{thebe2b}
	N^{1/2}(\log \log N)^{-1/2}\|\bar{\bg}_N(t)\|\to\infty\;\;\;\mbox{for some}\;\;\;0<t<1,
	\end{align}
	is satisfied, then
	$$
	(\log \log N)^{-1/2}\sup_{0<t<1}\frac{1}{[t(1-t)]^{1/2}}\left\|\bZ_N(t)-N^{1/2}\bar{\bg}_N(t)  \right\|=O_P(1).
	$$
\end{theorem}
The tests will stay consistent if $\tilde{\bG}(t)$ is replaced $\tilde{\bG}_N(t)$ in the testing procedures. For example, if $\underset{N\rightarrow \infty}{\lim}(N/h)^{1/2}(\log \log N)^{-1/2}\underset{1\leq \ell \leq M+1}{\max} \| \bbe^{**}-\bbe_\ell\|\rightarrow 0$, then
$$
\left(\log \log N \right)^{-1/2} \underset{0<t<1}{\sup} \frac{1}{(t(1-t))^{1/2}} \left( \bZ_N^\T(t) \tilde{\bG}_N^{-1}(t) \bZ_N(t) \right) \overset{P}{\rightarrow} \infty.
$$
The proofs of the theorems in this section are provided in the Online Supplement Section \ref{proof-2}.

\section{Smoothly changing error variance model}\label{sec-smooth}
So far, based on Assumption \ref{as-lin-ber-v}, we have considered a model where the structure of the errors in the observations might change during the observation period, but that the errors are piecewise stationary. In this section, we extend to the case when there are smooth changes in the variance of the errors during the intervals $(m_{i-1}, m_i], 1\leq i\leq M+1$.
Inspired by the mean change point model in G\'orecki et al.\ (2018), we modify the model of \eqref{limu1} as
\begin{align}\label{modelg}
y_i=\sum_{r=1}^{R+1}\bx_i^\T\bbe_r\mathds{1}\{k_{r-1}+1 \le i \le k_{r }\}+g(i/N)\eps_i,  \;\;1\leq i\leq N,
\end{align}
where the errors $\{\eps_i, -\infty<i<\infty\}$ are as described in Section \ref{sec-hetero}, and it holds that
\begin{assumption}\label{totalvar} $g$ has a finite total variation on $[0,1]$.
\end{assumption}
\noindent
Hence, we allow the variances of the errors $g(i/N)\eps_i$ to change even within the intervals $m_{\ell-1} < i \leq m_\ell$, for $1 \leq \ell \leq M+1$. The asymptotic behavior of $\bZ_N$, for $0 \leq t \leq 1$, in model \eqref{modelg} is established below.
\begin{theorem}\label{smo1} If $H_0$, Assumptions \ref{as-m} and \ref{asminons} and \ref{totalvar} hold, then 
\begin{align*}
\bZ_N(t)\stackrel{\cD^d[0,1]}{\longrightarrow}\bUp(t),
\end{align*}
where
$$
\bUp(t)=\bLamb(t)-t\bLamb(1)-\bv(t)\left(\sum_{\ell=1}^{M+1}(\tau_\ell-\tau_{\ell-1})\bA_\ell  \right)^{-1}\sum_{\ell=1}^{M+1}\bW_{\bD_\ell} \left(\int_{\tau_{\ell-1}}^{\tau_\ell}g^2(u)du \right),
$$
$$
\bLamb(t)=\int_0^t g(u)d\bGamma(u),
$$
and the Gaussian process $\{\bGamma(t), 0\leq t \leq 1\}$ is defined by \eqref{gaus}.
\end{theorem}
\noindent
We note that $\bLamb(t)$ is a $d$ dimenional time transformed Brownian motion with $E\bLamb(t)=0$, and
$$
E\bLamb(t)\bLamb^\T(s)=\bH(\min(t,s)),
$$
with
\begin{align*}
\bH(t)&=\int_0^{t}g^2(u)dE[\bGamma(u)\bGamma^\T(u)]\\
&=\sum_{j=1}^{k-1}\bD_j\int_{\tau_{j-1}}^{\tau_j}g^2(u)du+(t-\tau_{k-1})\bD_{k},\;\;\mbox{if} \;\;\;\tau_{k-1}<t \leq \tau_k,  1\leq k\leq M+1.
\end{align*}
 
Let
$$
\mca(t)=\int_0^tg^2(u)du.
$$
We note that the variance of the coordinates of $\bLamb(t)$ are proportional to $\mca(t)$, so it is natural to assume 
\begin{assumption}\label{ww1}  (i)  $0\leq \alpha_1<1/2$
$$
\lim_{t\to 0}\frac{1}{t^{\alpha_1}}\left(  \mca(t)\log\log(1/ \mca(t)) \right)^{1/2}=0
$$
(ii) $0<\alpha_2<1/2$
$$
\lim_{t\to 1}\frac{1}{(1-t)^{\alpha_2}}\left(  (\mca(1)-\mca(t))\log\log  (1/(\mca(1)-\mca(t))) \right)^{1/2}=0.
$$
\end{assumption}
We now present the weighted version of Theorem \ref{smo1}.
\begin{theorem}\label{smo2} If $H_0$, Assumptions \ref{as-m}--\ref{as-wc-1}, and \ref{totalvar} hold, then 
\begin{align*}
\frac{\bZ_N(t)}{ t^{\alpha_1}(1-t)^{\alpha_2 }}\stackrel{\cD^d[0,1]}{\longrightarrow}\frac{\bUp(t)}{ t^{\alpha_1}(1-t)^{\alpha_2}},
\end{align*}
where $\{\bUp(t),0\leq t \leq 1\}$ is defined in Theorem \ref{smo1}.
\end{theorem}
\noindent
Theorem \ref{smo2} implies immediately that 
\begin{align*}
\sup_{0<t<1}\frac{\|\bZ_N(t)\|}{ t^{\alpha_1}(1-t)^{\alpha_2}}\stackrel{\cD }{\longrightarrow}\sup_{0<t<1}\frac{\|\bUp(t)\|}{   t^{\alpha_1}(1-t)^{\alpha_2}}
\end{align*}
and
\begin{align*}
\sup_{0<t<1}\frac{\|\bZ_N(t)\|_\infty}{ t^{\alpha_1}(1-t)^{\alpha_2}}\stackrel{\cD }{\longrightarrow}\sup_{0<t<1}\frac{\|\bUp(t)\|_\infty}{   t^{\alpha_1}(1-t)^{\alpha_2}}.
\end{align*}
The proof of Theorem \ref{smo2} is given in the  Supplemental Section \ref{seprsm}. Since we modify only the error term, Theorem \ref{thebe0a} also holds under model \eqref{modelg}.

Next, we consider the standardized statistics of Theorem \ref{smo2} when
\beq\label{pow}
g(t)=ct^\varrho, \;\;\;\;\mbox{with some}\;\;\;c\neq 0\;\;\;\mbox{and}\;\;\; \varrho\geq 0.
\eeq
Let
$$
\bar{\bH}(t)=\bH(Nt)-2t\bH(tN)+t^2\bH(1).
$$
We note that $\bar{\bH}(k/N)$ is the covariance matrix of $\bLamb(k)-(k/N)\bLamb(N)$. 
\begin{theorem}\label{newed} If $H_0$, Assumptions \ref{as-m}--\ref{asminons}, \ref{totalvar} and \eqref{pow} hold, then we have 
\begin{align*}
\lim_{N\to \infty}P\left\{a(\log N)\sup_{0<t<1}\left[\bZ_N^\T(t)\bar{\bH}^{-1}(t)\bZ_N(t)\right]^{1/2}\leq x+b_d(\log N)\right\}=\exp\left(-\frac{2\varrho+2}{2\varrho+1}e^{-x}\right)
\end{align*}
for all $x$.
\end{theorem}

The main difference between Theorems \ref{newed} and \ref{coverd} arises from the substantially smaller variances of $g(i/N)\varepsilon_i$ under \eqref{pow} when $i$ is small, compared to the variances when $i$ is close to $N$. The variances of $g(i/N)\eps_i$ are converging to 0 as $i/N\to 0$, while 
to $g^2(1) E\eps^2_0$, if $i/N\to 1$. We then show that Theorem \ref{coverd} remains true if the variances of $g(i/N)\eps_{i}$ converge to positive constants if $i/N\to 0$ and $i/N\to 1$.
We replace \eqref{pow} with
\beq\label{pow1}
g(t)=c_1+c_2t^\varrho\;\;\;\mbox{with some}\;\;\;c_1\neq 0,   \;\;\;\mbox{and}\;\;\; \varrho\geq 0.
\eeq
\begin{theorem}\label{newedd} If $H_0$, Assumptions \ref{as-m}--\ref{asminons}, \ref{totalvar} and \eqref{pow1} hold, then we have 
\begin{align*}
\lim_{N\to \infty}P\left\{a(\log N)\sup_{0<t<1}\left[\bZ_N^\T(t)\bar{\bH}^{-1}(t)\bZ_N(t)\right]^{1/2}\leq x+b_d(\log N)\right\}=\exp\left(-2e^{-x}\right)
\end{align*}
for all $x$.
\end{theorem}

We also note that 
\begin{align*}
\lim_{N\to \infty}P\left\{a(\log N)\sup_{0<t<1}\left[\bZ_N^\T(t)\tilde{\bG}_N^{-1}(t)\bZ_N(t)\right]^{1/2}\leq x+b_d(\log N)\right\}=\exp\left(-2e^{-x}\right)
\end{align*}
under the assumptions of Theorem \ref{newedd}. This means that there is no difference between the covariance matrix estimated Darling--Erd\H{o}s results in Sections \ref{sec-hetero} and \ref{sec-smooth}. No information on the error structure is required to implement the testing procedure, as long as $E\varepsilon_i^2 \geq c > 0$. We can use the same method to compute the critical values used in cases of abrupt change errors or smoothly changing variance errors.

\section{Monte Carlo simulation and Finite Sample Performance}\label{se-simulation}

\subsection{The computation of critical values}\label{se-cv}\hfill\\
To assess the finite-sample performance of our tests, we focus on the standardized Darling--Erd\H{o}s-type statistics presented in Theorems~\ref{coverd} or~\ref{newedd}. As noted in Remark~\ref{remark}, the standardized statistics are more practical in data applications, as they do not require prior information on heteroscedasticity. However, Darling--Erd\H{o}s-type statistics typically suffer from slow convergence due to their exponential-type limiting distribution, often leading to tests that are under-sized and with reduced power. We put forward an improved finite sample approximation of the Darling--Erd\H{o}s  limiting distribution in this section. Let
\begin{align*}
	V^{HET}_N(1/2)=\sup_{0<t<1} \left(\bZ_N^\T(t)\tilde{\bG}_N^{-1}(t)\bZ_N(t)\right)^{1/2}, \mbox{ }
	Q^{HET}_N(1/2)=\sup_{0<t<1} \left\|\tilde{\bG}_N^{-1/2}(t)\bZ_N(t)\right\|_\infty,
\end{align*}
where $\tilde{\bG}_N(t), 0<t<1$ is defined in \eqref{plugin}. The long run variance estimators $\bG_N(u)$ in \eqref{hatgdef} were computed with the Bartlett kernel, and the bandwidth is selected through the automatic bandwidth selection method of Andrews (1991). This subsection explains the computation of critical values for the test statistics $V^{HET}_N(1/2)$ and $Q^{HET}_N(1/2)$.

In Appendix \ref{proof-2}, we show that in Theorem \ref{coverd}, a Gaussian approximation is first established for the process $\bZ_N$, and then the limiting distribution for the maximum of Gaussian processes is applied. The Gaussian approximation yields specifically that
\begin{align}\label{chor-mo}
\Biggl| P & \left\{a_N \sup_{0<t<1}\frac{1}{[t(1-t)]^{1/2}}\left({\bZ}_N^\T(t)\tilde{\bG}^{-1}(t){\bZ}_N(t)\right)^{1/2}\leq x + b_d(\log N)\right\} -\\
&P\left\{a_N \sup_{c_1(N)\leq t \leq 1-c_2(N)} \frac{1}{(t(1-t))^{1/2}}\left( \sum_{i=1}^{d} B_i^2(t)\right)^{1/2} \leq x + b_d(\log N) \right\}  \Biggl| \rightarrow 0,\notag
\end{align}
for any $c_1(N),c_2(N) \rightarrow 0$ and satisfying $c_1(N)= O\left((\log N)^{\kappa_1}/N\right)$ and $c_2(N)= O\left( (\log N)^{\kappa_2}/N \right)$ with any $-\infty < \kappa_1, \kappa_2< \infty$. Similarly,
\begin{align}\label{chor-mo1}
\Biggl| P & \left\{a_N \sup_{0<t<1}\frac{1}{[t(1-t)]^{1/2}}\left\|\tilde{\bG}^{-1/2}(t)\bZ_N(t)\right\|_\infty\leq x + b_d(\log N)\right\} -\\
&P\left\{a_N \sup_{c_1(N)\leq t \leq 1-c_2(N)} \frac{ |B_i(t)| }{(t(1-t))^{1/2}} \leq x + b_d(\log N) \right\} \Biggl| \rightarrow 0, \notag
\end{align}
where $\{B_1(t), 0\leq t \leq 1\},  \ldots, \{B_d(t), 0\leq t \leq 1\}$ are independent Brownian bridges.

To obtain the critical values, in practice one can simulate the random variable,
$$
\sup_{c_1(N)\leq t \leq 1-c_2(N)} \frac{1}{(t(1-t))^{1/2}}\left( \sum_{i=1}^{d} B_i^2(t)\right)^{1/2}
$$
for choices $c_1(N)$ and $c_2(N)$ such as $c_1(N)=c_2(N)=1/N$, and take its quantiles as critical values. We now instead introduce an approximation inspired by Vostrikova (1981), which does not require any simulation to obtain the empirical quantiles of the limits. We recall from Cs\"org\H{o} and Horv\'ath  (1997, p.\ 366),
\begin{align*}
\left\{ \frac{1}{[t(1-t)]^{1/2}}\left( \sum_{i=1}^dB_i^2(t)   \right)^{1/2}, 0<t<1 \right\}\stackrel{\cD}{=}
\left\{  U^*_d\left(\log\left(\frac{t}{1-t}\right)\right), 0<t<1\right\},
\end{align*}
where
$$
U^*_d(x)=\left(\sum_{i=1}^d U_i^2(x)\right)^{1/2},
$$
and $\{U_1(x), -\infty <x <\infty\}, \ldots, \{U_d(x), -\infty<x<\infty\}$ are independent, identically distributed Ornstein--Uhlenbeck processes, i.e.\ Gaussian processes with
$EU_i(x)=0$ and $EU_i(x)U_i(y)=\exp(-|x-y|/2)$.
As a result, we get from  \eqref{chor-mo} and \eqref{chor-mo1} that
\begin{align}\label{chor-mo2}
P\left\{\sup_{0<t<1}\frac{1}{[t(1-t)]^{1/2}}\left({\bZ}_N^\T(t)\tilde{\bG}^{-1}(t){\bZ}_N(t)\right)^{1/2}\leq x\right\}
\approx
P\left\{\sup_{0\leq t \leq \log N^2}U^*_d(t)\leq x\right\}
\end{align}
and
\begin{align}\label{chor-mo3}
P\left\{\sup_{0<t<1}\frac{1}{[t(1-t)]^{1/2}}\left\|\tilde{\bG}^{-1/2}(t)\bZ_N(t)\right\|_\infty\leq x\right\}\approx
P\left\{ \sup_{0\leq t \leq \log N^2}U^*_1(t)\leq x \right\}^d,
\end{align}
The critical values of the statistics $V^{HET}_N(1/2)$, and $Q^{HET}_N(1/2)$ are be approximated by using \eqref{chor-mo2} and \eqref{chor-mo3}. In Vostrikova (1981), it is shown that for $T>0$ and $r\geq 1$,
\begin{align}\label{eq-vost}
P\left\{ \sup_{0\leq t \leq T}U^*_r(t)> x   \right\}=\frac{x^r\exp(-x^2/2)}{2^{r/2}\Gamma(r/2)}
\left\{  T-\frac{r}{x^2} T+\frac{4}{x^2}+O\left(\frac{1}{x^4}\right)\right\},
\end{align}
where $\Gamma(\cdot)$ denotes the Gamma function. We ignore the $O(1/x^4)$ term, and then compute critical values directly from \eqref{eq-vost}. 

In an unreported simulation, we found that using critical values based on the Vostrikova approximation yields higher power than those from the Darling--Erd\H{o}s limit. Therefore, we use the Vostrikova-based critical values in the analysis below. For further details on the implementation of the non-standardised CUSUM statistics in Theorems~\ref{lin-var-1}, \ref{smo1}, and~\ref{smo2}, we refer the reader to Section~\ref{secriticalvalue} of the Supplementary Material.

\subsection{Monte Carlo Simulations}\label{se-mo}  \hfill\\
We now introduce the model and settings for the Monte Carlo simulation study, with results discussed in the next subsection. We consider a data generating process (DGP) taking the form \eqref{limu1} and simply allow one change point in the regression parameters:
\begin{equation}\label{sim-mod}
y_i = \bx_i^\T
(  1 + \delta \mathds{1}\{ i \ge k_1\})
+ \eps_i, \qquad 1 \leq i \leq N.
\end{equation}
The covariates are ${\bx}_i = (1, x_{2,i})^\T$ for $1 \leq i \leq N$, where $x_{2,i}$ is a heteroscedastic covariate following a segmented independent and identically distributed (i.i.d.) normal distribution with mean 0, standard deviation 3 before, and 0.5 after the break at $\lfloor 0.45 N \rfloor$. Prior to the change point $k_1$, we set the regression parameter as $\boldsymbol{\beta}_1 = (1 , 1 )^\top$, and the coefficients become to $\boldsymbol{\beta}_2 = (1+\delta, 1+\delta)^\top$ after the change. The change point size $\delta$ varies in the range $\delta \in \{-1.5, -1.2, -0.9, -0.6, -0.3, 0, 0.3, 0.6, 0.9, 1.2, 1.5\}$. The null hypothesis $H_0$ holds when $\delta=0$. 
The rejection rates at the nominal size $\omega = 0.05$ are reported as power curves in terms of the change size $\delta$, for each DGP and for middle change point at $k^ = \lfloor 0.5N \rfloor$ and early change point $k^ = \lfloor 0.2N \rfloor$.

For the error term $\epsilon_i$, we consider four heteroscedastic processes, encompassing both abrupt variance changes (cases i--iii) and a smooth variance change (case iv).

(i) (\textbf{Normal}) the error terms $\eps_i$ are i.i.d normal random variables:
\begin{equation*}
\epsilon_i = \left\{\begin{matrix}
N(0,3),\quad 1\leq i\leq \lfloor m^* N  \rfloor,
\\
N(0,0.5), \quad \lfloor m^* N  \rfloor<i\leq N.
\end{matrix}\right.
\end{equation*}

(ii) (\textbf{AR}) the error terms $\epsilon_i$ follow autoregressive (AR-1) process:
\begin{equation*}
\epsilon_i = \left\{\begin{matrix}
0.3 \epsilon_{i-1} + \varepsilon^{(1)}_i,\;\;  \mbox{with } \varepsilon^{(1)}_i \sim N(0,3), \;\;  1\leq i\leq \lfloor m^* N  \rfloor,
\\
0.3 \epsilon_{i-1} + \varepsilon^{(2)}_i,\;\;  \mbox{with } \varepsilon^{(2)}_i \sim N(0,0.5), \;\; \lfloor m^* N  \rfloor<i\leq N.
\end{matrix}\right.
\end{equation*}

(iii) (\textbf{GARCH}) the error terms $\epsilon_i$ follow a stationary GARCH(1,1) process defined by
\begin{equation*}
\epsilon_i = h_i^{1/2} \varepsilon_i, \mbox{ with }  h_i =
\left\{\begin{matrix}
3 + 0.01 \epsilon_{i-1}^2 + 0.8 h_{i-1},\;\;1\leq i\leq \lfloor m^* N  \rfloor,
\\
0.5 + 0.01 \epsilon_{i-1}^2 + 0.8 h_{i-1}, \;\; \lfloor m^* N  \rfloor<i\leq N.
\end{matrix}\right.
\end{equation*}
and the $\varepsilon_i$'s are i.i.d. standard normal random variables. In this exercise, we always set $m^*=\lf 0.45N \rf$, i.e., there is a variance change around the middle of the sample.

(iv) (\textbf{HeteSmooth}) Lastly, the error term follows a heteroscedastic smooth variance change process,
\begin{equation*}
\epsilon_i = g\left(\frac{i}{N} \right) \varepsilon_i, \mbox{  for  }
g\left(\frac{i}{N} \right) = 1+\frac{2}{1+ \exp(-10 (i/N - 0.5) )}.
\end{equation*}
We set the sample size $N= \{125,250\}$. All reported results are based on 2000 independent replications in each setting.  

We compare the proposed methods to three existing tests in the literature that considered change point detection in linear models with heteroscedastic errors, including Xu (2015), Perron et al. (2020) and Horv\'ath et al., (2021). Xu (2015) propose CUSUM--type tests to detect changes in linear models with nonstationary variance. We adopt the robust CUSUM test, with critical values derived from the limit approximated by a T-discrete steps Wiener process, denoted as the $\rm XUC$ test.
Perron et al. (2020) introduce a likelihood ratio--type test to detect changes in the coefficients, accommodating scenarios where the coefficient change up to $p$ times and the variance of the errors changes up to $q$ times. Specifically, we adopt their simulation setup, setting $p\in[0, 1, 2]$ and $q\in [1,2,3]$. The null hypothesis of no change is rejected when the test statistics exceed the critical values in at least one model specification. We employ the statistic $LR_{3,N}$ and refer to it as the $\rm{PYZ}$ test. Horv\'ath et al. (2021) introduce a R\'enyi--type statistic to compare the least squares estimators from sub-segments split by potential change points. Following their suggestion, we choose the tuning parameters $a_N = b_N = N^{1/2}$ and denote the statistic as $\rm{HMR}$ test. In Online Supplement Section \ref{sec-app-sim}, we examine the effect of the choice of $\kappa$ and the methods used to calculate the critical values. The results generally recommend using Darling--Erd\H{o}s type statistics with the distributional approximation obtained from \eqref{eq-vost} in terms of testing power. Therefore, we use the tests $V_N^{HET}(1/2)$ and $Q_N^{HET}(1/2)$ in this experiment.

Figure \ref{figure-peer125} and \ref{figure-peer250} display the power curves of the test candidates for $N = 125$ and $N = 250$, respectively. The overall patterns are consistent across both early and mid-sample changes, with greater power observed as the sample size increases. Tests $\mathrm{XUC}$ and $V^{\mathrm{HET}}_N(1/2)$ maintain approximately nominal size, while $Q^{\mathrm{HET}}_N(1/2)$ is slightly oversized. The $\mathrm{HMR}$ test performs well under \textbf{GARCH} errors but exhibits substantial size distortion when applied to other error structures, particularly under \textbf{HeteSmooth} errors. The test $\rm{PYZ}$ is oversized in all DGPs, and it deteriorates in models generated with \textbf{GARCH} errors. We note that the test $\rm{PYZ}$ is designed to mitigate power reduction, but the simulation shows that it can be oversized when dealing with changes on the tails and heteroscedastic covariates that fall outside the scope of its intended design. 

Under the alternative hypothesis, all tests start to gain reasonable power in large samples. The test $\rm{XUC}$ exhibits the lowest overall power, particularly with small sample sizes and early change points. The test $\rm{HMR}$ is relatively competitive in models with \textbf{NORMAL} and \textbf{GARCH} errors. The test $\rm{PYZ}$ achieves significantly enhanced power, but it is not reliable due to the size distortions. The test $Q^{HET}_N(1/2)$ consistently outperforms $V^{HET}_N(1/2)$, which is primarily used in Section \ref{sec-app} for empirical applications. Overall, the simulation results imply distortions of oversize or reduced power of the existing tests by encountering severe heteroscedasticity and various locations of changes in the linear models.
 
Additional simulation results for models with homoscedastic errors, as well as the effects of change-point locations and signal-to-noise ratios, are reported in Online Supplement Section~\ref{sec-app-sim}.

 \begin{figure}[H]
	\centering
	\caption{Rejection rates of $V^{HET}_N(1/2)$, $Q^{HET}_N(1/2)$, $\rm{XUC}$, $\rm{HMR}$, and $\rm{PYZ}$ at the 95\% significance level with heteroscedastic errors for an early change point $k^* = 0.2$ (first row) and a middle change point $k^* = 0.5$ (second row), with a sample size of $N = 125$.}
	\label{figure-peer125}
	\includegraphics[width=15cm]{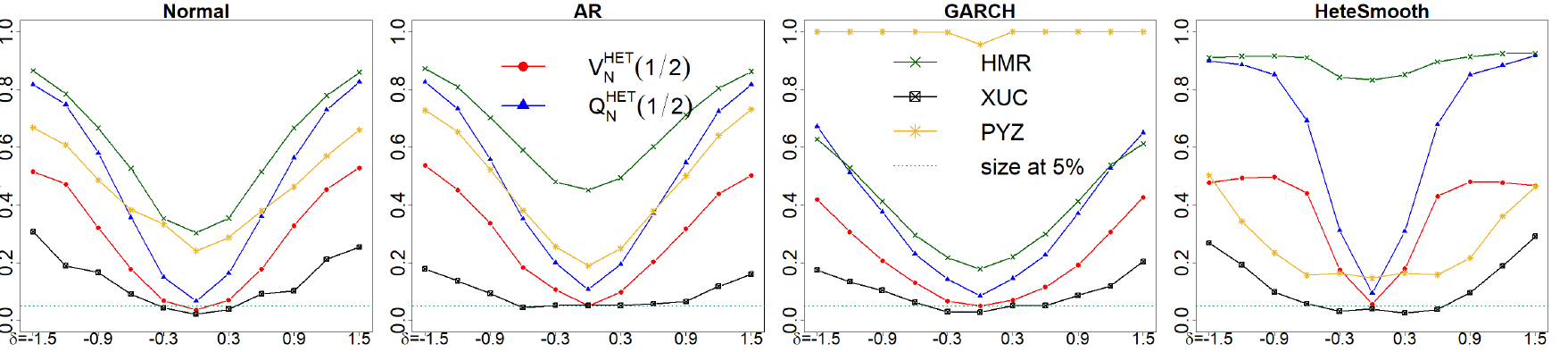}
	\includegraphics[width=15cm]{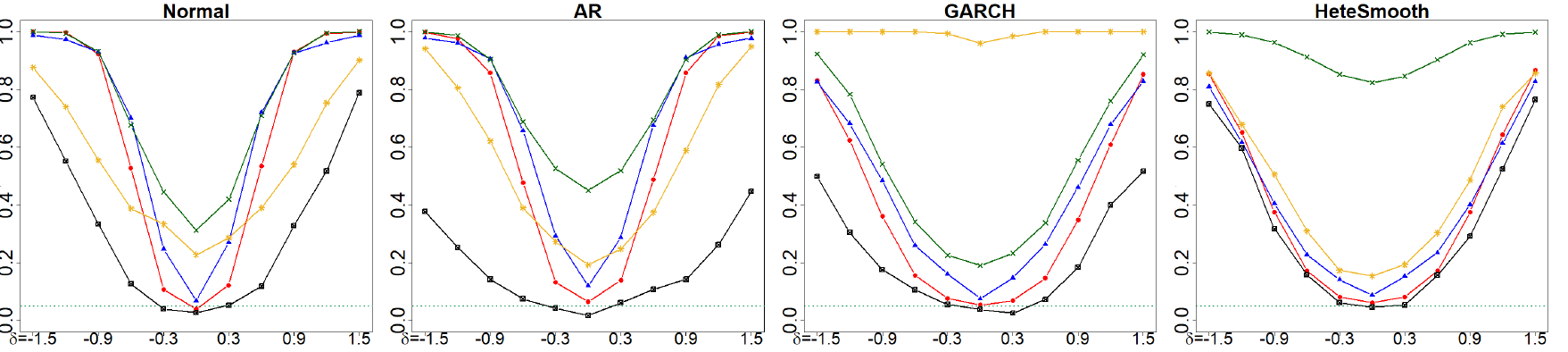}
\end{figure}

\begin{figure}[h]
	\centering
	\caption{Rejection rates of $V^{HET}_N(1/2)$, $Q^{HET}_N(1/2)$, $\rm{XUC}$, $\rm{HMR}$, and $\rm{PYZ}$ at the 95\% significance level with heteroscedastic errors for an early change point $k^* = 0.2$ (first row) and a middle change point $k^* = 0.5$ (second row), with a sample size of $N = 250$.}
	\label{figure-peer250}
	\includegraphics[width=15cm]{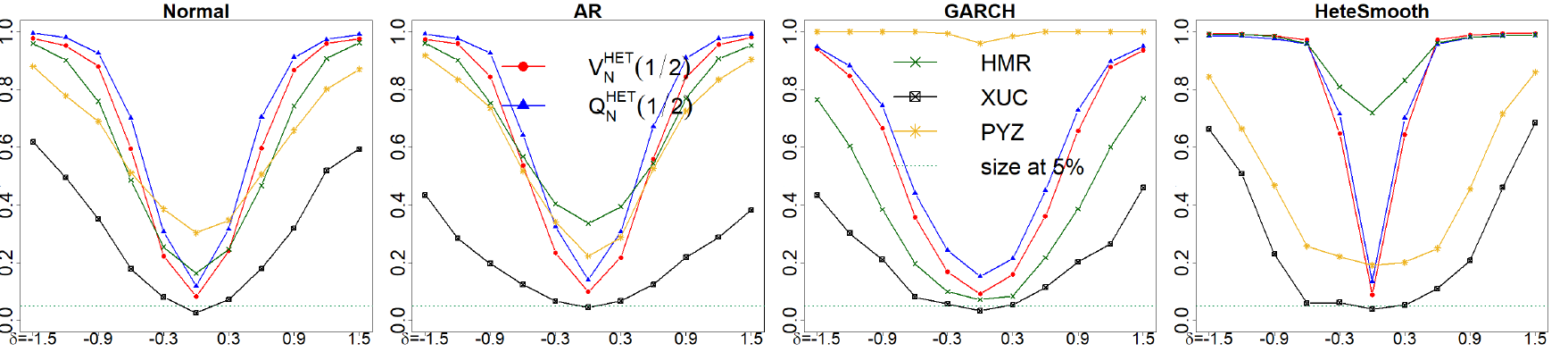}
	\includegraphics[width=15cm]{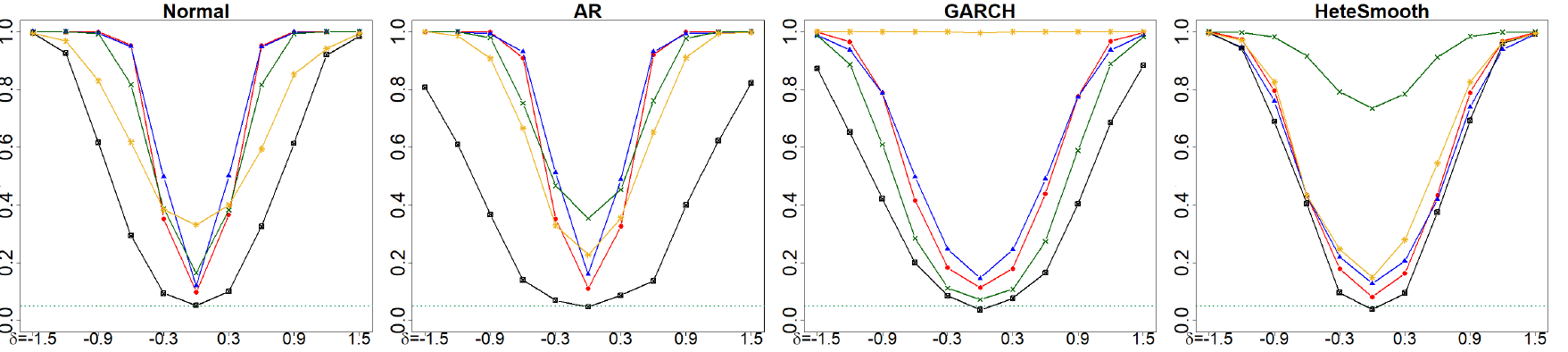}
\end{figure}

\section{Empirical data examples}\label{sec-app}

\subsection{Testing instability in macroeconomic prediction}\label{sec-app1}\hfill\\
We first illustrate a data application of the proposed methods to test the instability of predictive regression models. A typical univariate predictive regression model takes the form
\begin{equation}\label{eq-predict}
y_{i,j} = \beta_1 + \beta_2 x_{2,i} + \epsilon_i, \mbox{ } 1 \leq i\leq N, \mbox{ for variable } j,
\end{equation}
where $x_i$ denotes the predictor variable, typically representing observation from lagged values. These simple models are widely used in macroeconomic studies, which we also analyze here.

We follow McCracken and Ng (2016) and aim to forecast monthly U.S. GDP growth, industrial production, nonfarm employment, and total CPI inflation, indexed by $j = \{1, 2, 3, 4\}$. The covariate that we use to forecast each of these series is the real activity/employment predictive factor derived from a panel of 134 U.S. macroeconomic indicators in McCracken and Ng (2016). We consider a univariate predictive model because, as stated in McCracken and Ng (2016), two factor models only provide slight improvements marginally in terms of forecasting error. The sample ranges from January 1993 to December 2022 with 357 observations in total\footnote{The monthly macroeconomic data is collected from the webpage ``https://research.stlouisfed.org/econ\\/mccracken/fred-databases/".}. {Preliminary data analysis showed that the covariates appeared to exhibit abrupt heteroscedasticity at the onset of several periods of market turbulence, such as the Great Recession of 2008-2009 and the beginning of the COVID-19 pandemic. The standard deviation of the model residuals also shows changes when using a rolling window estimation, see Figure \ref{figure-employ}.}

The proposed tests were applied to each model for detecting change points. A change point is estimated based on each test statistic using the argument at which the corresponding normalized CUSUM processes achieved their maxima. We apply standard binary segmentation based on each change point test statistic with a threshold taken to be the 95\% null significance level in an attempt to detect additional change points. Table \ref{table-macro} in the appendix shows the change points detected by each approach and the corresponding estimated model coefficients. Consistent with the simulation results, the statistics $Q^{HET}_N(1/2)$ detects more changes, while the test $V^{HET}_N(1/2)$ is relatively conservative.

Table \ref{table-macro} shows the coefficient estimations in subsamples split by the first four estimated change points. In general, we found changes occurring in the GDP growth, industrial production and nonfarm employment models. Indicated by the $V^{HET}_N(1/2)$ and $Q^{HET}_N(1/2)$ tests, these three models experience changes around the period of recovery from the 2008 great recession. The tests also suggest changes around the COVID-19 pandemic in GDP growth and nonfarm employment models, with $Q^{HET}_N(1/2)$ detecting additional changes during the 1997 Asian financial crisis in the GDP growth model and during the early 2000s recession in all three models.

Figure \ref{figure-employ} illustrates an example of segmentation produced by $Q_N^{HET}(1/2)$ method for predictive regression model for nonfarm employment in terms real activity/employment predictive factor. The black line shows the model residuals from the model \eqref{eq-predict} applied to the entire sample, and the blue line shows rolling window estimates of the standard deviation of the model residuals, with window size 24 months. The shaded bars indicate the U.S. business cycle contractions according to NBER\footnote{NBER U.S. business cycle dating: ``https://www.nber.org/research/business-cycle-dating".}.

\begin{table}[h]
	\caption{ Change point detection results of the tests $V^{HET}_N(1/2)$ and $Q^{HET}_N(1/2)$ for macroeconomic predictive regression models. The dependent variables are the monthly U.S. GDP growth, Industrial production, Employment, CPI inflation from January 1993 to December 2020. The independent variable is the first predictive factor derived from McCracken and Ng (2016). The subsamples split by detected change points, and only one subsample indicates no change point detected. The values in the parentheses are the estimators of coefficient of the first predictive factor, with $^{***}$, $^{**}$ and $^{*}$ indicating significance at 1\%, 5\% and 10\% significance levels, respectively. }\label{table-macro}
	\begin{adjustbox}{max width=\linewidth}
		\begin{tabular}{cccccc }
			\hline\hline
			& \multicolumn{5}{c}{ $\boldsymbol{V^{HET}_N(1/2)}$ }            \\\hline
			& Subsample1                                                                & Subsample2                                                                & Subsample3                                                                & Subsample4                                                                & Subsample5                            \\\hline
			GDP Growth            & \begin{tabular}[c]{@{}c@{}}93Jan--09Jul\\ ($1.65^{***}$)\end{tabular}  & \begin{tabular}[c]{@{}c@{}}09Aug--20Feb\\ ($-2.65^{***}$)\end{tabular} & \begin{tabular}[c]{@{}c@{}}20Mar--22Dec\\ ($-0.83^{**}$)\end{tabular}  &      &    \\
			Industrial Production & \begin{tabular}[c]{@{}c@{}}93Jan--08Dec\\ ($-1.62^{***}$)\end{tabular} & \begin{tabular}[c]{@{}c@{}}09Jan--22Dec\\ ($-1.54^{***}$)\end{tabular} &    &    &                \\
			Employment            & \begin{tabular}[c]{@{}c@{}}93Jan--00Jun\\ ($-0.37^{***}$)\end{tabular} & \begin{tabular}[c]{@{}c@{}}00Jul--20Jul\\ ($-0.17^{***}$)\end{tabular} & \begin{tabular}[c]{@{}c@{}}20Aug--22Dec\\ ($-0.52^{***}$)\end{tabular} &      &     \\
			CPI inflation         & \begin{tabular}[c]{@{}c@{}}93Jan--22Dec\\ ($-0.15^{***}$)\end{tabular} &     &       &       &     \\\hline
            & \multicolumn{5}{c}{$\boldsymbol{Q^{HET}_N(1/2)}$} \\
            GDP Growth             & \begin{tabular}[c]{@{}c@{}}93Jan--97Jan\\ ($1.26^{**}$)\end{tabular}   & \begin{tabular}[c]{@{}c@{}}97Feb--03Oct\\ ($0.88^{***}$)\end{tabular}  & \begin{tabular}[c]{@{}c@{}}03Nov--09Jan\\ ($2.12^{***}$)\end{tabular} & \begin{tabular}[c]{@{}c@{}}09Feb--21Apr\\ ($-0.98^{***}$)\end{tabular} & \begin{tabular}[c]{@{}c@{}}21May--22Dec\\ $0.96$\end{tabular} \\
			Industrial Production  & \begin{tabular}[c]{@{}c@{}}93Jan--03Mar\\ ($-1.81^{***}$)\end{tabular} & \begin{tabular}[c]{@{}c@{}}03Apr--08Aug\\ ($-1.44^{***}$)\end{tabular}   & \begin{tabular}[c]{@{}c@{}}08Sep--14Jul\\ ($-1.74^{***}$)\end{tabular}  & \begin{tabular}[c]{@{}c@{}}14Aug--22Dec\\ ($-1.60^{***}$)\end{tabular}  &       \\
			Employment        & \begin{tabular}[c]{@{}c@{}}93Jan--01Apr\\ $(-0.39^{***})$\end{tabular} & \begin{tabular}[c]{@{}c@{}}01May--11Feb\\ $(-0.45^{***})$\end{tabular} & \begin{tabular}[c]{@{}c@{}}11Mar--20Dec\\ $(-0.10^{**})$\end{tabular} & \begin{tabular}[c]{@{}c@{}}21Jan--22Dec\\ $(-0.22^{*})$\end{tabular}  &                    \\
			CPI inflation       & \begin{tabular}[c]{@{}c@{}}93Jan--22Dec\\ ($-0.15^{***}$)\end{tabular} &      &    &        &        \\\hline\hline
		\end{tabular}
	\end{adjustbox}
\end{table}

\begin{figure}[h]
	\centering
	\caption{ Segmentations produced by $Q_N^{HET}(1/2)$ method for the predictive regression model for nonfarm employment in terms real activity/employment predictive factor. The black line shows the model residuals from the model \eqref{eq-predict} applied to the entire sample, and the blue line shows rolling window estimates of the standard deviation of the model residuals, with window size 24 months. The shared bars indicate the U.S. business cycle contractions. }\label{figure-employ}
	\includegraphics[width=13cm]{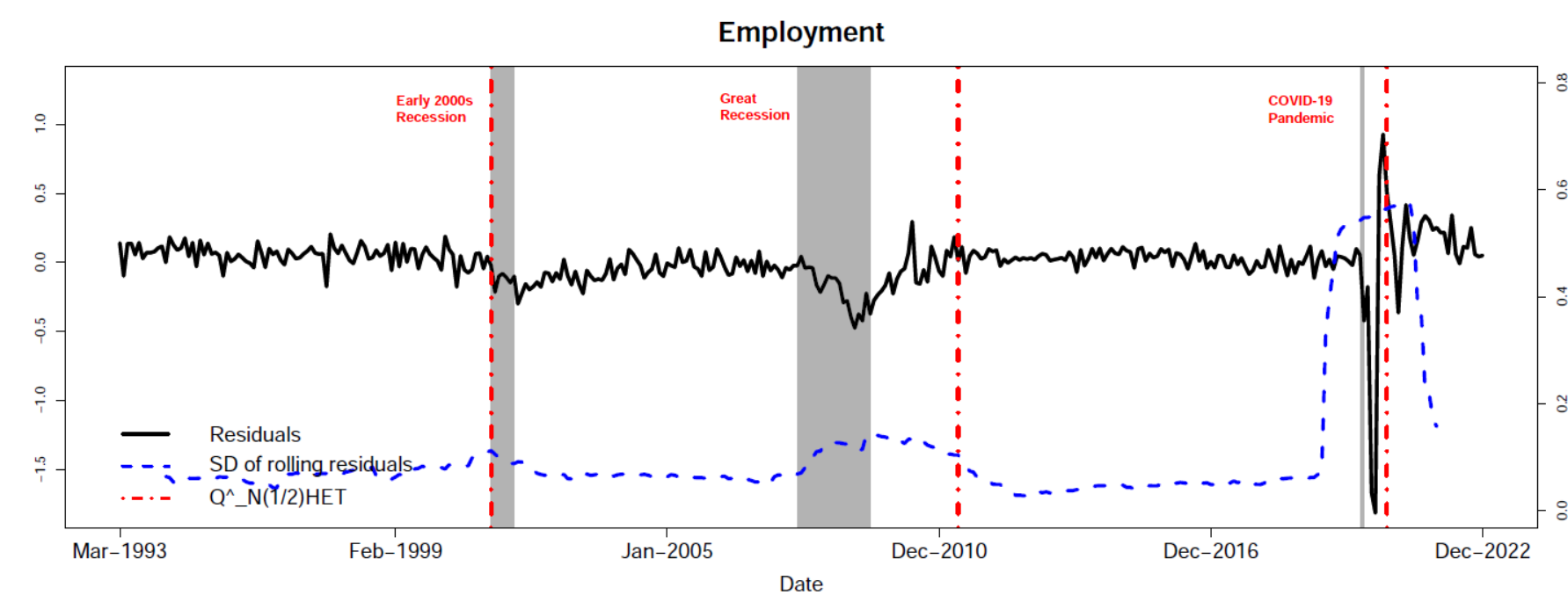}
\end{figure}

\subsection{Changes in investor sentiment effect on the U.S. stock market}\label{sec-app2}\hfill\\
In this section, we demonstrate a second application to detect changes in explaining the sentiment anomaly on cross--sectional U.S. stock returns. This provides an example of the proposed tests for an explanatory regression model with multiple covariates. The aims of modeling market sentiment anomalies are to try and model two phenomena; how the demand for speculative investments drives stock prices away from their fundamental values, and relative arbitrage, i.e.\ the existence of a collection of stocks that are too risky and costly for arbitrage. The topic-influential work of Baker and Wurgler (2006) reviews the anecdotal history of investment sentiment in the U.S.\  between 1961 and 2002, and constructs sentiment factors to predict cross--sectional stock returns.


We procured a dataset covering the period January 1960--June 2022\footnote{The data is collected from the webpage of Jeffrey Wurgler ``http://people.stern.nyu.edu/jwurgler/".}, containing two sentiment factors (SENTI) constructed from six underlying sentiment proxies, including the closed--end fund discount, NYSE share turnover, the number and average first--day returns on IPOs, the equity share in new issues, and the dividend premium. Our study uses the second sentiment factor because it accounts for business cycle variations. The sample consists of 684 time series observations, and extends beyond the data considered in Baker and Wurgler (2006); it includes some recent major economic events of note, such as the US housing bubble, the great recession, and the Covid-19 pandemic.

Following Baker and Wurgler (2006), we consider the premium of the size factor (small--minus--big, SMB) factor as a dependent variable to verify the distinct sentiment effect among small firms. The linear regression model specifies four dependent variables, including the sentiment index, market excess return (RMRF), premium of the book--to--market (high--minus-low, HML), and premium on winners minus losers (momentum, MOM) factors\footnote{The Fama--French--Carhart factors are obtained from the data library ``https://mba.tuck.dartmouth.edu\\/pages/faculty/ken.french/data\_library.html".}. The model can be stated explicitly as: 	
\begin{equation}\label{eq-sentiment}
\mbox{SMB}_i = \beta_1 + \beta_2 \mbox{SENTI}_i + \beta_3 \mbox{RMRF}_i + \beta_4 \mbox{HML}_i + \beta_5 \mbox{MOM}_i + \epsilon_i,\quad 1\leq i \leq 684.
\end{equation}
{ While all other covariates appear to be reasonably stationary, the covariate SENTI exhibited apparent changes in its variance. The standard deviation of the model residuals again exhibits heteroscedasticity when using a rolling window estimation. }

We test for change points in the regression parameters in \eqref{eq-sentiment} using the statistics $V^{HET}_N(1/2)$ and $Q^{HET}_N(1/2)$. The changes around April 2002, April 2009, and June 2020 are detected for both tests in Table \ref{table-beta}, indicating changes occurring during the burst of the early 2000s recession, the recovery from the great recession, and the outbreak of the Covid-19 pandemic. The $Q^{HET}_N(1/2)$ test suggests the presence of one more change in July 1973, which can be a consequence of the 1973–-1974 stock market crash.

Table \ref{table-beta} also shows the results of the estimations when the sample is segmented with the detected changes. Focusing on $V^{HET}_N(1/2)$ and $Q^{HET}_N(1/2)$, we estimate $\beta_2=-0.22$ with p--value 0.07 based on the sampling period between 1965 and 2002, while $\beta_2$ is estimated as $-0.23$ and $-0.18$ in subsamples respectively by split from July 1973. Both coefficients appear to be insignificant, but their p--values close to 0.10. Our findings are roughly consistent with Baker and Wurgler (2006), who estimated $\beta_2=-0.30$ with p--value $0.15$ using the period 1961 to 2002. The negative sign of the coefficient indicates that there is a negative relationship with the sentiment premium, i.e., the small firms turn to gain less returns with intense market sentiment. This effect becomes more manifest during the formation and collapse of the U.S. housing bubble between 2002 and 2009, given the coefficient $\beta_2$ enlarges to $-1.29$. The sentiment effect then becomes insignificant in subsamples after 2009, but the uncovered negative effect is consistent throughout each subsamples.

\begin{table}[H]
	\caption{ The estimated parameters after the segmentation of the model \eqref{eq-sentiment} using the estimator derived from $V^{HET}_N(1/2)$ and $Q^{HET}_N(1/2)$, with $^{***}$, $^{**}$ and $^{*}$ indicating significance at 1\%, 5\% and 10\% significance levels, respectively. }\label{table-beta}
	\begin{adjustbox}{max width=\linewidth}
	\begin{tabular}{cccccc}
		\hline\hline
		& \multicolumn{5}{c}{\textbf{$V^{HET}_N(1/2)$}}        \\\hline
		& $\beta_1$ & $\beta_2$    & $\beta_3$    & $\beta_4$    & $\beta_5$   \\\hline
		65Jul--02Apr & 0.30      & $-0.22^{*}$  & $0.14^{***}$ & $-0.23^{**}$ & $-0.01$      \\
		02May--09May & 0.14      & $-1.29^{**}$ & $0.18^{***}$ & 0.08         & 0.05      \\
		09Jun--20May & $-0.40$     & $-1.11$        & $0.21^{***}$ & 0.03         & $-0.00$      \\
		20Jun--22Jun & 1.04      & $-1.07$        & 0.06         & $-0.06$        & $-0.12$     \\ \hline 
        &   \multicolumn{5}{c}{\textbf{$Q^{HET}_N(1/2)$}}                      \\\hline 
		  65Jul--73Jul & 0.37      & $-0.23$        & $0.39^{***}$ & $-0.29^{**}$  & $-0.10$    \\
		  73Aug--02May & 0.33      & $-0.18$        & $0.08^{**}$  & $-0.25^{***}$ & 0.00      \\
		  02Jun--09Apr & 0.26      & $-1.25^{**}$ & $0.17^{**}$  & 0.13          & 0.05      \\
		  09May--20Aug & -0.38     & $-0.86$        & $0.20^{***}$ & 0.06          & 0.01      \\
		  20Sep--22Jun & 0.94      & $-1.02$        & 0.12         & $-0.07$         & $-0.05$    \\\hline\hline
	\end{tabular}
	\end{adjustbox}
\end{table}

\section{Conclusion}\label{concl}
We propose quadratic forms and maxima of weighted CUSUM residual processes to test for multiple changes in linear model parameters under potential heteroscedasticity in both covariates and errors. The error variance is allowed to change either abruptly or smoothly. The asymptotic distributions of the proposed test statistics are established under general conditions that accommodate both homoscedastic and heteroscedastic cases. We examine the finite sample performance of the standardized statistics in detail. Monte Carlo simulations demonstrate that the tests exhibit good size and power in finite samples, and that the adjustments for heteroscedasticity in model errors perform well in practice. We applied our method to find changes in popular macroeconomic and return prediction models, and to detect changes in the sentiment asset pricing models in the U.S.\ stock market.


\pagebreak

\begin{center}
	\textbf{\large Online Supplement: Detecting Multiple Change Points in Linear Models with Heteroscedastic errors}
\end{center}\label{sec-appen}

\setcounter{section}{0}
\appendix
\section{Proofs of Theorems \ref{lin-var-1}--\ref{thebe0a}}\label{proof-2}
	Throughout this section we assume that $M\geq 1$ since $M=0$, i.e.\ the second order properties of the observations stay stationary is already covered in Section \ref{intro}.
	
	\begin{lemma}\label{stromu} If Assumptions \ref{as-m}--\ref{weps-mu} are satisfied, then for each $N$ we can define $M+1$  independent Gaussian processes $\{\bGam_{N,j}(x), 0\leq x \leq m_j-m_{j-1}\}$
		such that $E\bGam_{N,j}(x)=\bf0$, $E\bGam_{N,j}(x)\bGam_{N,j}^\T(y)=\min(x,y)\bD_j$, $1\leq j \leq M+1,$
		$$
		\sup_{1\leq x\leq m_1}\frac{1}{x^\zeta}\left\|\sum_{i=1}^{\lf x\rf}\bx_i\eps_i-\bGam_{N,1}(x)\right\|=O_P(1)
		$$
		
		$$
		\sup_{1\leq x \leq m_{j}-m_{j-1}}\left\|\sum_{i=m_{j-1}+1}^{m_{j-1}+\lf x\rf}\bx_i\eps_i-\bGam_{N,j}(x)\right\|=O_P(N^{\zeta}), \quad 2\leq j\leq M,
		$$
		
		and
		$$
		\sup_{m_M\leq x< N }\frac{1}{(N-x)^\zeta}\left\|\sum_{i=\lf x\rf +1}^N\bx_i\eps_i-\bGam_{N,M+1}(N-x)\right\|=O_P(1)
		$$
		with some $\zeta<1/2$.
	\end{lemma}
	\begin{proof} It follows from Assumption \ref{as-lin-ber-v} that $\{\bx_i\eps_i, m_{\ell-1}<i\leq m_{\ell}\}$ is a Bernoulli decomposable sequence for any $1\leq \ell\leq M+1$ and
		\beq\label{beva1}
		\left(E\left\|\bx_{i}\eps_i-\bx_{i,j}^*\eps_{i,j}^*   \right\|^{\nu/2}\right)^{2/\nu}\leq c_1j^{-\alpha},
		\eeq
		where $\nu$ and $\alpha$ are given in Assumption \ref{as-lin-ber-v}. Now the approximations in Lemma \ref{stromu} follows from Aue et al.\ (2014). They also prove that the infinite series defining $\bD_\ell, 1\leq \ell\leq M+1$ is absolutely convergent.
	\end{proof}

    Let 
    \begin{align*}
        \bP_N(t)=\frac{1}{N^{1/2}}\left(\sum_{i=1}^{\lf Nt\rf}\bx_i\eps_i-\frac{\lf Nt  \rf}{N}\sum_{i=1}^N\bx_i\eps_i\right),
    \end{align*}
	and 
    \begin{align*}
        \bB_N(t)=N^{-1/2}\left(\hat{\bGam}_N(t) - t \hat{\bGam}_N(1) \right).
    \end{align*}
    where
		\begin{align}\label{Gamhat}
		\hat{\bGam}_N(x)=\sum_{j=1}^{\ell-1}\bGam_{N,j}(m_j-m_{j-1})+\bGam_{N,\ell}(x-m_{\ell}), \quad m_{\ell}<x\leq m_{\ell+1}, 1\leq \ell\leq M.
		\end{align}
    We note that for any N, $\{\bB_N(t), 0\leq t \leq 1\}$ is a Gaussian process with $E\bB_N(t)=0$ and $E\bB_N(t)\bB^\top_N(s)=\bG(\min(t,s))-t\bG(s)-s\bG(t)+ts\bG(1)$.
	\begin{lemma}\label{stro2v}  If Assumptions \ref{as-m}--\ref{weps-mu} are satisfied, then
		\beq\label{stro21v}
		N^{-1/2+\zeta}\sup_{0< t < 1}\frac{1}{[t(1-t)]^\zeta}\left\| \bP_N(t)- \bB_N(t) \right\|=O_P(1)
		\eeq
		with  some $\zeta<1/2$.\\
		Also,
		\beq\label{stro22v}
		\sup_{1/(N+1)\leq  t \leq 1-1/(N+1)}\frac{1}{(t(1-t))^{1/2}}\left\| \bP_N(t)-\bB_N(t) \right\|=O_P(1).
		\eeq
	\end{lemma}
	\begin{proof} We write
		\begin{align*}
		\bR(k)=\sum_{i=1}^k\bx_i\eps_i=\sum_{j=1}^\ell\sum_{i=m_{j-1}}^{m_j}\bx_i\eps_i+\sum_{i=m_{\ell}+1}^k\bx_i\eps_i,
		\end{align*}
		if $m_\ell<k\leq m_{\ell+1}, 1\leq \ell\leq M+1$. By Lemma \ref{stromu} we have
		\begin{align*}
		\sup_{1\leq x \leq m_1}   \frac{1}{x^\zeta}\left\|  \bR(\lf x \rf)-\bGam_{N,1}(x) \right\|=O_P(1),
		\end{align*}
		
		\begin{align*}
		\sup_{m_{j-1}<x\leq m_j}\left\| \bR(\lf x \rf)-\left(\sum_{i=1}^{j-1}\bGam_{N,i}(m_i-m_{i-1}) +\bGam_{N,j}(x-m_{j-1})\right)    \right\|  =O_P(N^{\zeta}), \quad 2\leq j \leq M,
		\end{align*}
		
		and
		\begin{align*}
		\sup_{m_{M}<x<N}  \frac{1}{(N-x)^\zeta}\left\|\bR(N)-\bR(\lf x \rf)-\bGam_{N,M+1}(N-x)\right\|=O_P(1).
		\end{align*}
		
		Thus, we get
		\begin{align}\label{muap1}
		\sup_{1\leq x \leq N/2}     \frac{1}{x^\zeta}\left\| \bR(x)-\hat{\bGam}_N(x)   \right\|=O_P(1),
		\end{align}
		
		\begin{align}\label{muap2}
		\sup_{N/2\leq x <N}    \frac{1}{(N-x)^\zeta}\left\| \left(\bR(N)-\bR(\lf x\rf )\right)-\left(\hat{\bGam}_N(N)-\hat{\bGam}_N(x) \right)  \right\|=O_P(1).
		\end{align}
		By the definition of $\bB_N(t)$, \eqref{stro21v} follows immediately from   \eqref{muap1} and \eqref{muap2}.  \eqref{stro22v} follows similarly.
	\end{proof}

	\begin{lemma}\label{stro3v} We assume that  Assumptions \ref{as-m}--\ref{weps-mu}, and \ref{as-wc-1} hold.\\
		(i) If $I(w,c)<\infty$ with some $c>0$, then
		\beq\label{stro31mu}
		\sup_{0<t<1}\frac{1}{w(t)}\left\|\bP_N(t)\right\|\;\stackrel{\cD}{\to}
		\sup_{0<t<1}\frac{1}{w(t)}\left\| \bB(t)   \right\|,
		\eeq
		where $\bB(t)$ satisfies $\{\bB(t), 0\leq t \leq 1\}\overset{\cD}{=} \{\bB_N(t), 0\leq t \leq 1\}$.\\
		(ii) If in addition Assumption \ref{asvarnosd} also holds, then
		\begin{align}\label{stro32mu}
		\lim_{N\to\infty}P\Biggl\{ a(\log N)&\max_{1\leq k<N}\left[
		\left(\sum_{i=1}^k\bx_i{\eps}_i-\frac{k}{N}\sum_{i=1}^N\bx_i{\eps}_i\right)^\T(N \bG (k/N))^{-1} \right.\\
		&\vspace{3cm}\left.\times
		\left(\sum_{i=1}^k\bx_i{\eps}_i-\frac{k}{N}\sum_{i=1}^N\bx_i{\eps}_i\right)\right]^{1/2}
		\leq x+ b_d(\log N) \Biggl\}=\exp(-2e^{-x})\notag
		\end{align}
		for all $x$, where $a(x)$ and $b_d(x)$ are defined in Theorem \ref{coverd}.
	\end{lemma}
	
	\begin{proof} It follows from Assumption \ref{as-m} and Lemma \ref{stro2v} that
		\begin{align*}
		\sup_{\tau_1\leq t \leq \tau_M}\frac{1}{w(t)}\left\| \bP_N(t)-\bB_N(t)   \right\|=o_P(1).
		\end{align*}
		As in Lemma \ref{stro2v} for all $0<\delta<\tau_1$
		\begin{align*}
		\sup_{\delta\leq t \leq \tau_1}\frac{1}{w(t)}\left\| \bP_N(t)-\bB_N(t)   \right\|=o_P(1),
		\end{align*}
		and
		\begin{align*}
		\sup_{1/(N+1)\leq t \leq \delta}\frac{1}{w(t)}\left\| \bP_N(t)-\bB_N(t)   \right\|
		&=\sup_{1/(N+1)\leq t \leq \delta}\frac{1}{[t(1-t)]^{1/2}}\left\| \bP_N(t)-\bB_N(t)   \right\|\sup_{0<t\leq \delta}\frac{t^{1/2}}{w(t)}\\
		&=O_P(1)\sup_{0<t\leq \delta}\frac{t^{1/2}}{w(t)}.
		\end{align*}
		
		Further,
		\begin{align*}
		\lim_{\delta\to 0}\limsup_{N\to \infty}P\left\{\sup_{1/(N+1)\leq t \leq \delta}\frac{1}{w(t)}\left\| \bP_N(t)-\bB_N(t)   \right\|
		>x\right\}=0
		\end{align*}
		for all $x>0$. It is easy to see that
		$$
		\sup_{0<t\leq 1/(N+1)}\frac{1}{w(t)}\|\bP_N(t)\|=o_P(1)
		$$
		and
		$$
		\sup_{1/(N+1)<t\leq \delta}\frac{1}{w(t)}\|\bB_N(t)\|\stackrel{\cD}{\to}      \sup_{0<t\leq \delta}\frac{1}{w(t)}\|\bB(t)\|
		$$
		for all $0<\delta<\tau_1$, since the coordinates of $\bB$ are linear combinations of independent Brownian bridges. By symmetry, for any $0<\delta<1-\tau_M$
		\begin{align*}
		\sup_{\tau_M\leq t \leq 1-\delta}\frac{1}{w(t)}\left\| \bP_N(t)-\bB_N(t)   \right\|=o_P(1),
		\end{align*}
		and
		$$
		\sup_{1-\delta\leq t\leq 1-1/(N+1)}\frac{1}{w(t)}\|\bar{\bGam}_N(t)\|\stackrel{\cD}{\to}      \sup_{1-\delta\leq t< 1}\frac{1}{w(t)}\|\bB(t)\|.
		$$
		Also,
		$$
		\sup_{1-1/(N+1)\leq t<1}\frac{1}{w(t)}\|\bP_N(t)\|=o_P(1)
		$$
		and
		\begin{align*}
		\lim_{\delta\to 0}\limsup_{N\to \infty}P\left\{\sup_{1-\delta \leq t \leq 1-1/(N+1)}\frac{1}{w(t)}\left\| \bP_N(t)-\bB_N(t)   \right\|
		>x\right\}=0
		\end{align*}
		for all $x>0$, completing the proof of \eqref{stro31mu}. \\
		
		First we note that  by Assumption \ref{as-m} and Lemma \ref{stro2v} we have
		$$
		\sup_{\tau_1\leq t\leq \tau_M}\frac{1}{[t(1-t)]^{1/2}}\|\bP_N(t)\|=O_P(1).
		$$
		Using Lemma \ref{stromu} we get
		\begin{align*}
		\max_{1\leq k\leq m_1}\left(\frac{N}{k(N-k)}\right)^{1/2}\left\|\sum_{i=1}^k\bx_i\eps_i-\frac{k}{N}\sum_{i=1}^N\bx_i\eps_i   \right\|=
		\max_{1\leq k\leq m_1}\left(\frac{N}{k(N-k)}\right)^{1/2}\left\|\sum_{i=1}^k\bx_i\eps_i\right\|+O_P(1)
		\end{align*}
		and
		\begin{align*}
		\max_{1\leq k\leq m_1} \left(\frac{N^2}{k(N-k)}\right)^{1/2} \left\|\bB_N(k/N)\right\|\stackrel{\cD}{=}
		\max_{1\leq k\leq m_1} \left(\frac{N^2}{k(N-k)}\right)^{1/2} \left\|\bW_{\bD_1}(k/N)\right\|+O_P(1),
		\end{align*}
		where  $\{\bW_{\bD_1}, 0\leq t \leq 1\}$ is a Gaussian process with $E\bW_{\bD_1}(t)=\bf0$ and $E\bW_{\bD_1}(t)\bW_{\bD_1}^\T(s)=\min(t,s)\bD_1$. Also,
		$\bG(k/N)=(k/N)\bD_1$, and therefore
		\begin{align*}
		\biggl\{\bW_{\bD_1}^\T(k/N)\bG^{-1}(k/N)\bW_{\bD_1}(k/N), 1\leq k \leq m_1\biggl\}\stackrel{\cD}{=}\biggl\{
		\frac{1}{(k/N)}\|\bW(k/N)\|^2, 1\leq k \leq m_1\biggl\},
		\end{align*}
		$$
		\bW(t)=(W_1(t),W_2(t), \ldots, W_d(t))^\T,
		$$
		where $\{W_1(t), 0\leq t \leq 1\}, \{W_2(t), 0\leq t \leq 1\}, \ldots, \{W_d(t), 0\leq t \leq 1\}$ are independent Wiener processes. Theorem A.3.1  of  Cs\"org\H{o} and Horv\'ath (1997) yields
		\begin{align}\label{lil1}
		\frac{1}{2\log \log N}\max_{1\leq k\leq m_1}\frac{N}{k}\|\bW(k/N)\|^2\stackrel{P}{\to} 1
		\end{align}
		and
		\begin{align}\label{lil2}
		\max_{1\leq k\leq \log N}\frac{N}{k}\|\bW(k/N)\|^2=O_P(\log \log \log N).
		\end{align}
		Thus we conclude
		\begin{align*}
		\lim_{N\to\infty}P\Biggl\{\max_{1\leq k\leq m_1}\frac{N}{k}\|\bW(k/N)\|^2=\max_{\log N\leq k\leq m_1}\frac{N}{k}\|\bW(k/N)\|^2
		\Biggl\}=1.
		\end{align*}
		Putting together \eqref{stro22v}, \eqref{lil1} and \eqref{lil2} we get
		\begin{align*}
		\frac{1}{2\log \log N}\max_{1\leq k\leq m_1}\frac{1}{k}\left(\sum_{i=1}^k\bx_i\eps_i   \right)^\T\bD_1^{-1}\left(\sum_{i=1}^k\bx_i\eps_i   \right)    \stackrel{P}{\to} 1
		\end{align*}
		and
		\begin{align*}
		\max_{1\leq k\leq \log N}\frac{1}{k}\left(\sum_{i=1}^k\bx_i\eps_i   \right)^\T\bD_1^{-1}\left(\sum_{i=1}^k\bx_i\eps_i   \right) =O_P(\log \log \log N).
		\end{align*}
		Hence
		\begin{align*}
		\lim_{N\to\infty}P\Biggl\{\max_{1\leq k\leq m_1}\frac{1}{k} \left(\sum_{i=1}^k\bx_i\eps_i   \right)^\T\bD_1^{-1}\left(\sum_{i=1}^k\bx_i\eps_i   \right) =\max_{\log N\leq k\leq m_1}\frac{1}{k}\left(\sum_{i=1}^k\bx_i\eps_i   \right)^\T\bD_1^{-1}\left(\sum_{i=1}^k\bx_i\eps_i   \right)
		\Biggl\}=1.
		\end{align*}
		Lemma \ref{stromu} yields
		\begin{align*}
		\max_{\log N\leq k\leq m_1}\frac{1}{k}\Biggl| \left(\sum_{i=1}^k\bx_i\eps_i   \right)^\T\bD_1^{-1}\left(\sum_{i=1}^k\bx_i\eps_i   \right)
		-\bGam_{N, 1}^\T(k)\bD_1^{-1}\bGam_{N, 1}(k)  \Biggl|=o_P\left(1/\left(\log \log N\right)\right)
		\end{align*}
		and therefore
		\begin{align*}
		\Biggl|\max_{1\leq k\leq m_1}\frac{1}{k}\left(\sum_{i=1}^k\bx_i\eps_i   \right)^\T\bD_1^{-1}\left(\sum_{i=1}^k\bx_i\eps_i   \right)
		- \max_{1\leq k\leq m_1}\frac{1}{k}\bGam_{N, 1}^\T(k)\bD_1^{-1}\bGam_{N, 1}(k) \Biggl|=o_P\left(1/\left(\log \log N\right)\right).
		\end{align*}
		Observing that
		\begin{align*}
		\biggl\{ \bGam_{N, 1}^\T(k)\bD_1^{-1}\bGam_{N, 1}(k), 1\leq k \leq m_1   \biggl\}\stackrel{\cD}{=}
		\biggl\{  \|\bW(k)\|^2, 1\leq k\leq m_1   \biggl\},
		\end{align*}
		Lemma A.3.1  of Cs\"org\H{o} and Horv\'ath (1997)  implies
		\begin{align*}
		\lim_{N\to\infty}P\left\{a(\log N)\max_{1\leq k\leq m_1}\frac{1}{k}\bGam_{N, 1}^\T(k)\bD_1^{-1}\bGam_{N, 1}(k)\leq x+b_d(\log N)
		\right\}=\exp(-e^{-x})
		\end{align*}
		for all $x$.\\
		
		By symmetry,
		\begin{align*}
		\Biggl|\max_{m_M<k<N}&\left(\sum_{i=1}^k\bx_i\eps_i-\frac{k}{N}\sum_{i=1}^N\bx_i\eps_i\right)^\T(N\bar{\bG}(k/N,k/N))^{-1}\left(\sum_{i=1}^k\bx_i\eps_i
		-\frac{k}{N}\sum_{i=1}^N\bx_i\eps_i\right)\\
		&-\max_{m_M<k<N}\frac{1}{N-k}\bGam_{N,M+1}^\T(N-k)\bD_{M+1}^{-1}\bGam_{N,M+1}(N-k)
		\Biggl|=o_P\left(\left(\log \log N\right)^{-1}\right).
		\end{align*}
		Applying again Lemma A.3.1 of Cs\"org\H{o} and Horv\'ath (1997) we get
		\begin{equation*}
		\begin{split}
		& \lim_{N\to\infty}P\left\{a(\log N)\max_{m_M< k<N}\frac{1}{N-k}\bGam_{N, M+1}^\T(N-k)\bD_{M+1}^{-1}\bGam_{N, M+1}(N-k)\leq x+b_d(\log N)
		\right\}\\
		& =\exp(-e^{-x})
		\end{split}
		\end{equation*}
		for all $x$. Since $\{\bGam_{N,1}(x), 1\leq x\leq m_1\}$ and $\{\bGam_{N,M+1}(N-x), m_M< x\leq N\}$ are independent, the proof of \eqref{stro32mu} is complete.
	\end{proof}

	\begin{lemma}\label{stro4v} Assuming that \ref{as-m}--\ref{weps-mu}, and \ref{as-wc-1} hold, we have\\
		(i)
		\begin{align}\label{a17}
		\underset{\tau_1\leq t <\tau_M}{\sup}\frac{N^{-1/2}}{w(t)}\left\| \sum_{i=1}^{\lf (N+1)t\rf}[\bx_i\bx_i^\T-E\bx_i\bx_i^\T] \right\|=O_P(1),
		\end{align}
		(ii)
		\begin{align}\label{a18}
		\underset{0< t \leq \tau_1}{\sup}\frac{N^{-1/2}}{w(t)}\left\| \sum_{i=1}^{\lf (N+1)t\rf}[\bx_i\bx_i^\T-E\bx_i\bx_i^\T] \right\|=O_P(1),
		\end{align}
		and
		\begin{align}\label{a19}
		\underset{\tau_M< t < 1}{\sup}\frac{N^{-1/2}}{w(t)}\left\| \sum_{i=1}^{ N }[\bx_i\bx_i^\T-E\bx_i\bx_i^\T] \right\|=O_P(1),
		\end{align}
		(iii)
		\begin{align}\label{a20}
		\underset{0< t \leq \tau_1}{\sup}\frac{N^{-1/2}}{(t(1-t))^{1/2}}\left\| \sum_{i=1}^{\lf (N+1)t\rf}[\bx_i\bx_i^\T-E\bx_i\bx_i^\T] \right\|=O_P((\log\log N)^{1/2}),
		\end{align}
		\begin{align}\label{a21}
		\underset{\tau_1< t \leq \tau_M}{\sup}\frac{N^{-1/2}}{(t(1-t))^{1/2}}\left\| \sum_{i=1}^{\lf (N+1)t\rf}[\bx_i\bx_i^\T-E\bx_i\bx_i^\T] \right\|=O_P(1),
		\end{align}
		and
		\begin{align}\label{a22}
		\underset{\tau_M< t <1}{\sup}\frac{N^{-1/2}}{(t(1-t))^{1/2}}\left\| \sum_{i=\lf (N+1)t\rf+1}^{N}[\bx_i\bx_i^\T-E\bx_i\bx_i^\T] \right\|=O_P( (\log\log N)^{1/2} ).
		\end{align}
	\end{lemma}
	\begin{proof} Similarly to \eqref{beva1}, Assumption \ref{as-lin-ber-v} yields
		\begin{equation}\label{a23}
		\left(E\left\|\bx_i\bx_i^\T-\bx_{i,j}^*\bx_{i,j}^{*\T}\right\|^{\nu/2}\right)^{\nu/2}\leq c_1(i-m_{\ell-1})^{-\alpha},
		\end{equation}
		if $m_{\ell-1}<i\leq m_\ell$, $1\leq \ell \leq M+1$. Using the approximation in Aue et al. (2009), we can define independent Gaussian process $\{\bDelta_{N,1}(k), 0<k\leq m_1\}$, $\{\bDelta_{N,2}(k), m_1<k\leq m_2\}$,\dots, $\{\bDelta_{N,M+1}(k), m_M<k\leq N\}$ such that
		\begin{align}\label{a23-2}
		\underset{1 \leq k \leq M+1}{\max}\underset{1 \leq \ell \leq m_k-m_{k-1}}{\max} \frac{1}{\ell^\zeta}
		\left\|  \sum_{i=m_{k-1}}^{m_{k-1}+\ell}[\bx_i\bx_i^\T-E\bx_i\bx_i^\T] - \bDelta_{N,k}(\ell)  \right\| = O_P(1),
		\end{align}
		and
		\begin{align}\label{a24}
		\underset{1 \leq \ell \leq N-m_N}{\max} \frac{1}{\ell^\zeta}
		\left\|  \sum_{i=N-\ell}^{N}[\bx_i\bx_i^\T-E\bx_i \bx_i^\T] - \bDelta_{N,M+1}(\ell)  \right\| = O_P(1),
		\end{align}
		with some $\zeta<1/2$. We note that $E\bDelta_{N,k}(x)=\mathbf{O}$, $0 \leq x \leq m_k-m_{k-1}$, $1\leq k \leq M+1$, where $\mathbf{O}$ is the zero matrix. Also, the covariance of $\bDelta_{N,k}(x)$ is
		$$
		E\bDelta_{N,k}(x)\otimes \bDelta_{N,k}(y) = \min(x,y)\underset{N\rightarrow\infty}{\lim} \sum_{|\ell|\leq m_{k}-m_{k-1}} E \bx_{m_{k-1}+1} \bx_{m_{k-1}+1}^\T \otimes \bx_{m_{k-1}+1+\ell} \bx_{m_{k-1}+1+\ell}^\T,
		$$
		where $\otimes$ denotes the Kronecker product. The statement in \eqref{a17} is an immediate consequence of \eqref{a23-2} and \eqref{a24}, since $0<\tau_1,\; \tau_M<1$. Due to the approximations in \eqref{a23} with $k=1$ and \eqref{a24}, the proof of Lemma \ref{stro3v} can be repeated to establish \eqref{a18} and \eqref{a19}. We observe that \eqref{a21} follows from \eqref{a17}. The coordinates of $\bDelta_{N,k}(x)$ are Brownian motions and therefore by the law of the iterated logarithm, we get
		$$
		\underset{1\leq x\leq m_k-m_{k-1}}{\max} (x\log \log(x+3) )^{-1/2}\| \bDelta_{N,k}(x)\| = O_P(1).
		$$
		Hence, we obtain \eqref{a20} and \eqref{a22}.
	\end{proof}
	
	\begin{lemma}\label{stro5v}
		If Assumptions \ref{as-m}--\ref{weps-mu} hold, then
		$$
		\hat{\bbe}_N-\bbe_0=\left( \sum_{\ell=1}^{M+1}(\tau_\ell-\tau_{\ell-1}) \bA_\ell  \right)^{-1} \frac{1}{N} \sum_{i=1}^{N} \bx_i \eps_i + O_P(N^{-1/2}).
		$$
	\end{lemma}
	\begin{proof}
		We note
		$$
		\hat{\bbe}_N-\bbe_0 = \left(\bX_N^\T\bX_N \right)^{-1} \bX_N^\T\bE_N = \left(\bX_N^\T\bX_N \right)^{-1} \sum_{i=1}^{N} \bx_i\eps_i,
		$$
		$\bE_N=\{\eps_1, \eps_2, \dots, \eps_N\}^\T$. It follows from Lemma \ref{stromu} that $\| \sum_{i=1}^{N} \bx_i \eps_i \|=O_P(N^{1/2})$ and from \eqref{a23-2} and \eqref{a24} that
		\begin{align}
		\left\| \frac{1}{N} \bX_N^\T\bX_N - \sum_{\ell=1}^{M+1}(\tau_\ell-\tau_{\ell-1}) \bA_\ell  \right\|=O_P(N^{-1/2}).
		\end{align}
		This completes the proof of the lemma.
	\end{proof}

	\noindent
	{\bf Proof of Theorem \ref{lin-var-1}.} We write, as in \eqref{a20},
	\beq\label{decomp}
	\begin{split}
		& \sum_{i=1}^k\bx_i\hat{\eps}_i-\frac{k}{N}\sum_{i=1}^{N}x_i\hat{\eps}_i=\sum_{i=1}^{k}\bx_i\eps_i-\frac{k}{N}\sum_{i=1}^{N}\bx_i\eps_i-\sum_{i=1}^{k}(\bx_i\bx_i^\T-E\bx_i\bx_i^T)(\hat{\bbe}_N-\bbe_0)\\
		&+\frac{k}{N}\sum_{i=1}^{N}(\bx_i\bx_i^\T - E\bx_i \bx_i^\T)(\hat{\bbe}_N-\bbe_0) - \left(\sum_{i=1}^k E\bx_i\bx_i^\T - \frac{k}{N}\sum_{i=1}^{N} E\bx_i\bx^\T  \right) (\hat{\bbe}_N-\bbe_0),
	\end{split}
	\eeq
	if $k\leq N/2$, and
	\begin{align*}
	& \sum_{i=1}^k\bx_i\hat{\eps}_i-\frac{k}{N}\sum_{i=1}^{N}x_i\hat{\eps}_i=\sum_{i=1}^{k}\bx_i\eps_i-\frac{k}{N}\sum_{i=1}^{N}\bx_i\eps_i-\left(1-\frac{k}{N}\right) \sum_{i=1}^{N}(\bx_i\bx_i^\T-E\bx_i\bx_i^T)(\hat{\bbe}_N-\bbe_0)\\
	&+\sum_{i=k+1}^{N}(\bx_i\bx_i^\T - E\bx_i \bx_i^\T)(\hat{\bbe}_N-\bbe_0) - \left( \left(1-\frac{k}{N}\right) \sum_{i=1}^N E\bx_i\bx_i^\T - \sum_{i=k+1}^{N} E\bx_i\bx^\T  \right) (\hat{\bbe}_N-\bbe_0),
	\end{align*}
	if $k> N/2$. It follows from Lemmas \ref{stro4v} and \ref{stro5v} that
	\begin{align*}
	\underset{0<t\leq1/2}{\sup} \frac{N^{-1/2}}{w(t)} \left\| \sum_{i=1}^{\lf (N+1)t\rf} (\bx_i\bx_i^\T - E\bx_i\bx_i^\T)(\hat{\bbe}_N - \bbe_0) \right\| = O_P(N^{-1/2}),
	\end{align*}
	\begin{align*}
	\underset{0<t\leq1/2}{\sup} \frac{N^{-1/2}}{w(t)} \left\| \frac{\lf (N+1)t\rf}{N}\sum_{i=1}^{N} (\bx_i\bx_i^\T - E\bx_i\bx_i^\T)(\hat{\bbe}_N - \bbe_0) \right\| = O_P(N^{-1/2}),
	\end{align*}
	and
	\begin{align*}
	\underset{1/2<t<1}{\sup} \frac{N^{-1/2}}{w(t)} \left\| \left(1-\frac{\lf (N+1)t\rf}{N}\right) \sum_{i=1}^{N} (\bx_i\bx_i^\T - E\bx_i\bx_i^\T)(\hat{\bbe}_N - \bbe_0) \right\| = O_P(N^{-1/2}),
	\end{align*}

	\begin{align*}
	\underset{1/2<t<1}{\sup} \frac{N^{-1/2}}{w(t)} \left\|  \sum_{i=\lf (N+1)t\rf+1}^{N} (\bx_i\bx_i^\T - E\bx_i\bx_i^\T)(\hat{\bbe}_N - \bbe_0) \right\| = O_P(N^{-1/2}).
	\end{align*}
	Let
	\begin{align*}
	\bv_N(t) = \left\{\begin{matrix}
	\displaystyle \frac{1}{N}\sum_{i=1}^{\lf (N+1)t\rf} E \bx_i\bx_i^\T - \frac{\lf (N+1)t\rf}{N}\frac{1}{N} \sum_{i=1}^{N} E \bx_i\bx_i^\T, & 0<t\leq 1/2\\
	\displaystyle \frac{1}{N}\left(1-\frac{\lf (N+1)t\rf}{N} \right) \sum_{i=1}^{N} E \bx_i\bx_i^\T - \frac{1}{N}\sum_{i=\lf (N+1)t\rf+1}^{N} E \bx_i\bx_i^\T,& 1/2<t<1.
	\end{matrix}\right.
	\end{align*}
	It is easy to see that 
    \begin{align}\label{a22-2}
      \underset{0<t<1}{\sup} \| \bv_N(t)-\bv(t)\|\rightarrow 0  
    \end{align}
    where $\bv(t)$ is defined in \eqref{13a}. Also,
	\begin{align*}
	\underset{\delta\rightarrow 0}{\lim} \underset{N\rightarrow \infty}{\lim} \underset{0<t \leq \delta}{\sup} \frac{1}{w(t)} \| \bv(t)\| = 0, \quad \underset{\delta\rightarrow 0}{\lim} \underset{0<t \leq \delta}{\sup} \frac{1}{w(t)} \| \bv(t)\| = 0,
	\end{align*}
	and
	\begin{align*}
	\underset{\delta\rightarrow 0}{\lim} \underset{N\rightarrow \infty}{\lim} \underset{1-\delta \leq t <1}{\sup} \frac{1}{w(t)} \| \bv(t)\| = 0.
	\end{align*}
	On the interval $1\leq k \leq m_1$, $\bx_i\eps_i$ is stationary, so by Lemma \ref{stromu}
	\begin{align*}
	\underset{\delta\rightarrow 0}{\lim} \underset{N\rightarrow \infty}{\lim \sup} P \left\{ \underset{0<t\leq \delta}{\sup} \frac{N^{-1/2}}{w(t)} \left\| \sum_{i=1}^{\lf (N+1)t\rf}\bx_i\eps_i - \frac{\lf (N+1)t\rf}{N} \sum_{i=1}^{N} \bx_i\eps_i   \right\| >u \right\} = 0,
	\end{align*}
	\begin{align*}
	\underset{\delta\rightarrow 0}{\lim} \underset{N\rightarrow \infty}{\lim \sup} P \left\{ \underset{1-\delta\leq t<1}{\sup} \frac{N^{-1/2}}{w(t)} \left\| \sum_{i=1}^{\lf (N+1)t\rf}\bx_i\eps_i - \frac{\lf (N+1)t\rf}{N} \sum_{i=1}^{N} \bx_i\eps_i   \right\| >u \right\} = 0,
	\end{align*}
	for all $u>0$. By the law of the iterated logarithm, we get
	\begin{align*}
	\underset{\delta \rightarrow 0}{\lim} \underset{0<t \leq \delta}{\sup}
	\frac{1}{w(t)} \left\| \bB(t) \right\| = 0, \quad \mbox{a.s. and} \quad
	\underset{\delta \rightarrow 0}{\lim} \underset{1-\delta\leq t < t}{\sup}
	\frac{1}{w(t)} \left\| \bB(t) \right\| = 0, \quad \mbox{a.s.}
	\end{align*}
	It follows from Lemma \ref{stromu} that
	$$
	\bU_N(t) \overset{\mathcal{D}[\delta,1-\delta]}{\rightarrow} \bar{\bGam}(t), \quad \mbox{for all } 0 <\delta<1/2,
	$$
	where
	$$
	\bU_N(t) = N^{-1/2}\left( \sum_{i=1}^{\lf (N+1)t\rf} \bx_i\eps_i - \frac{\lf (N+1)t\rf}{N} \sum_{i=1}^{N} \bx_i\eps_i \right) - \bv_N(t) N^{1/2} (\hat{\bbe}_N - \bbe_0).
	$$
	and $\bar{\bGam}(t)$ is defined in \eqref{gaus}. This completes the proof of Theorem \ref{lin-var-1}.
	\noindent

	{\bf Proof of Theorem \ref{coverd}.}
	We use again the decomposition in \eqref{decomp}. We show that the maximum is reached on the interval $[1/\log N, 1-1/\log N]$. The Gaussian approximation in \eqref{a23-2} and \eqref{a24} with the law of the iterated logarithm yield,
	\begin{align*}
	\underset{1\leq k \leq N/2}{\max}
	\frac{1}{k^{1/2}} \left\| \sum_{i=1}^{k} (\bx_i\bx_i^\T - E\bx_i\bx_i^\T) (\hat{\bbe}_N-\bbe_0) \right\| = O_P(N^{-1/2}(\log\log N)^{1/2})
	\end{align*}
	and
	\begin{align*}
	\underset{1\leq k \leq N/2}{\max}
	\frac{1}{k^{1/2}} \left\| \frac{k}{N}\sum_{i=1}^{N} (\bx_i\bx_i^\T - E\bx_i\bx_i^\T) (\hat{\bbe}_N-\bbe_0) \right\| = O_P( N^{-1/2}).
	\end{align*}
	Using again \eqref{a23-2} and \eqref{a24}
	\begin{align*}
	\underset{1\leq k \leq \log N}{\max}
	\frac{1}{k^{1/2}} \left\| \sum_{i=1}^{k} \bx_i\eps_i - \frac{k}{N} \sum_{i=1}^{N} \bx_i\eps_i \right\| = O_P( (\log \log \log N)^{1/2}).
	\end{align*}
	and by elementary arguments
	\begin{align*}
	\underset{1\leq k \leq \log N}{\max}
	\frac{1}{k^{1/2}} \left\| \left(\sum_{i=1}^{k} E\bx_i\bx_i^\T - \frac{k}{N} \sum_{i=1}^{N} E\bx_i\bx_i^\T\right) (\hat{\bbe}_N-\bbe_0) \right\| = O_P( N^{-1/2} (\log \log \log N)^{1/2}  ),
	\end{align*}
	\begin{align*}
	\underset{1\leq k \leq N/2}{\max}
	\frac{1}{k^{1/2}} \left\| \left(\sum_{i=1}^{k} E\bx_i\bx_i^\T - \frac{k}{N} \sum_{i=1}^{N} E\bx_i\bx_i^\T\right) (\hat{\bbe}_N-\bbe_0) \right\| = O_P( 1 ).
	\end{align*}
	Using again the law of the iterated logarithm for Brownian motions with \eqref{a23-2} and \eqref{a24}, we get
	\begin{align*}
	\left( \frac{1}{\log \log N}\right)^{1/2}
	\underset{1\leq k \leq N/2}{\sup}
	\frac{1}{k^{1/2}} \left\| \sum_{i=1}^{k} \bx_i\eps_i - \frac{k}{N} \sum_{i=1}^{N} \bx_i\eps_i \right\| \overset{P}{\rightarrow} c_7.
	\end{align*}
	with some $c_7>0$. Thus we get
    \begin{align*}
    \underset{N\rightarrow \infty}{\lim} P & \bigg\{ \underset{\log N \leq k \leq N/2}{\max}   \frac{1}{k^{1/2}} 
    \left\|  \sum_{i=1}^{k} \bx_i\hat{\eps}_i - \frac{k}{N} \sum_{i=1}^{N} \bx_i \hat{\eps}_i \right\|  \\
    &  = \underset{1\leq k \leq N/2}{\max} \frac{1}{k^{1/2}} 
    \left\|  \sum_{i=1}^{k} \bx_i\hat{\eps}_i - \frac{k}{N} \sum_{i=1}^{N} \bx_i \hat{\eps}_i \right\|  \bigg\} = 1.
\end{align*}

	and
    \begin{align*}
    &\left\| \underset{\log N \leq k \leq N/2}{\max}  \frac{1}{k^{1/2}} \left(\sum_{i=1}^{k} \bx_i\hat{\eps}_i - \frac{k}{N} \sum_{i=1}^{N} \bx_i \hat{\eps}_i \right) \right.  \left. - \frac{1}{k^{1/2}} \left[ \left(  \sum_{i=1}^{k} \bx_i \eps_i - \frac{k}{N} \sum_{i=1}^{N} \bx_i\eps_i  \right) \right. \right. \\
    &\quad \left. \left. -
    \left(  \sum_{i=1}^{k} E\bx_i \bx_i^\T - \frac{k}{N} \sum_{i=1}^{N} \bx_i\bx_i^\T  \right) (\hat{\bbe}_N - \bbe_0) \right]   \right\|  = O_P(N^{-1/2} (\log \log N)^{1/2} ).
\end{align*}

	The Darling--Erd\H{o}s law (c.f. Cs\"org\H{o} and Horv\'ath, 1997) with the approximation implies
	\begin{align*}
	\underset{N/\log N \leq k \leq N/2}{\max} \frac{1}{k^{1/2}} \left\| \sum_{i=1}^{k} \bx_i\eps_i - \frac{k}{N}\sum_{i=1}^{N}\bx_i\eps_i \right\|  = O_P( (\log \log \log N)^{1/2} ).
	\end{align*}
	Thus, we conclude
	\begin{align*}
	\underset{N\rightarrow \infty}{\lim} P & \left \{ \underset{\log N \leq k \leq N/\log N}{\max} \frac{1}{k^{1/2}} \left\|  \sum_{i=1}^{k} \bx_i\hat{\eps}_i - \frac{k}{N} \sum_{i=1}^{N} \bx_i \hat{\eps}_i \right\| \right . \\
	& \left .= \underset{\log N \leq k \leq N/\log N}{\max} \frac{1}{k^{1/2}} \left\|  \sum_{i=1}^{k} \bx_i\hat{\eps}_i - \frac{k}{N} \sum_{i=1}^{N} \bx_i \hat{\eps}_i \right\| \right \} =1.
	\end{align*}
	Since
	\begin{align*}
	\underset{\log N \leq k \leq N/\log N}{\max} \frac{1}{k^{1/2}} \left\|  \left( \sum_{i=1}^{k} E\bx_i\bx_i^\T - \frac{k}{N} \sum_{i=1}^{N} E\bx_i \bx_i\right) (\hat{\bbe}_N-\bbe_0) \right\| = O_P( (\log N)^{-1/2} ),
	\end{align*}
	we conclude that
	\begin{align*}
	\underset{\log N \leq k \leq N/\log N}{\max} \frac{1}{k^{1/2}} \left\|  \sum_{i=1}^{k} \bx_i\hat{\eps}_i - \frac{k}{N} \sum_{i=1}^{N} \bx_i\hat{\eps}_i \right\| = O_P( (\log N)^{-1/2} ),
	\end{align*}
    and
	\begin{align*}
	\underset{\log N \leq k \leq N/\log N}{\max} \frac{1}{k^{1/2}} \left\|  \sum_{i=1}^{k} \bx_i\hat{\eps}_i  \right\| = O_P( (\log N)^{-1/2} ).
	\end{align*}
	According to the Darling--Erd\H{o}s law (c.f. Cs\"org\H{o} and Horv\'ath, 1997) and \eqref{a23-2} on $[1, m_1]$, we get for all $x$
	\begin{align*}
	\underset{N\rightarrow \infty}{\lim} P\{  a(\log N) \underset{\log N \leq k \leq N/\log N}{\max} \frac{1}{k} \left(\sum_{i=1}^{k} \bx_i \eps_i\right)^\T\bD_1^{-1} \left(\sum_{i=1}^{k} \bx_i \eps_i\right)\leq x + b_d (\log N)   \} = \exp(-e^{-x}).
	\end{align*}
	We can repeat our arguments on $[N/2, N]$ and we get the limit result for
	\begin{align*}
	\max \left\{ k^{-1/2} \left\| \sum_{i=N-k}^{N} \bx_i \eps_i \right\|, \quad \log N\leq k \leq N/\log N   \right\}.
	\end{align*}
	Since the approximating processes on $[1,m_1]$ and $[m_{M}+1,N]$ are independent, we get for all $x \in \mathbb{R}$
    \begin{align*}
    &\underset{N\rightarrow \infty}{\lim} P\Bigg\{ a(\log N) \max \Bigg[
    \underset{\log N \leq k \leq N/\log N}{\max} \left(\frac{1}{k} \sum_{i=1}^{k}\bx_i \eps_i \right)^\T \bD_1^{-1} \left( \sum_{i=1}^{k} \bx_i \eps_i  \right),  \\
	&\quad \underset{N-N/\log N\leq k \leq N-\log N}{\max} \frac{1}{N-k} \left(\sum_{i=k+1}^{N} \bx_i \eps_i\right)^\T \bD_{M+1}^{-1} \left(\sum_{i=k+1}^{N} \bx_i \eps_i\right) 
    \Bigg] \leq x + b_d(\log N)   \Bigg\} = \exp(-2 e ^{-x}).
\end{align*}

	Elementary calculations yield
	\begin{align*}
	E\bar{\bGam}_N(t)\bar{\bGam}^\top_N(t)=\tilde{\bG}(t) = & \bG(t) - 2t\bG(t) + t^2\bG(t) - \bu(t)\bG(t) - \bG(t) \bu^\T(u)\\ &+ t\bu(t)\bG(t)+t\bG(t)\bu(t)^\T + \bu(t)\bG(1)\bu(t)^\T,
 	\end{align*}
	and
	\begin{align*}
	\underset{(\log N)/N \leq t \leq 1/\log N}{\sup} \left\|  \tilde{\bG}(t) - t \bD_1 \right \| = O\left( \frac{1}{ (\log N)^2}\right),
	\end{align*}
	\begin{align*}
	\underset{1-1/(\log N)  \leq t \leq 1-(\log N)/N}{\sup} \left\|  \tilde{\bG}(t) - (1-t) \bD_{M+1} \right \| = O\left( \frac{1}{ (\log N)^2}\right),
	\end{align*}
	since
	\begin{align*}
	\underset{(\log N)/N \leq t \leq 1/\log N}{\sup} \left\|  \frac{\bv(t)}{t} \right \| = O(1),
	\underset{1-1/\log N \leq t \leq 1-(\log N)/N}{\sup} \left\|  \frac{\bv(t)}{1-t} \right \| = O(1).
	\end{align*}
	Here we can replace $(N/k)^{-1} \bD_1$ with $\tilde{\bG}^{-1}(k/N)$ and $(N/(N-k) )^{-1} \bD_1$ with $\tilde{\bG}^{-1}(k/N)$. The first part of proof is completed, and the similar arguments apply to the second part.
	\qed
	\noindent

	{\bf Proof of Theorem \ref{hetcos}.} We  consider
	\begin{align*}
	\sum_{\ell=0}^{k-1} K\left(\frac{\ell}{h}\right) \frac{1}{N-\ell} \sum_{i=1}^{k-\ell} \eps_i \eps_{i+\ell } \bx_i \bx_{i+\ell}^\top = & \sum_{\ell=0}^{k-1} K\left(\frac{\ell}{h}\right) \frac{1}{N-\ell} \sum_{i=1}^{k} \eps_i \eps_{i+\ell } \bx_i \bx_{i+\ell}^\T\\
	&-\sum_{\ell=0}^{k-1} K\left(\frac{\ell}{h}\right) \frac{1}{N-\ell} \sum_{i=k-\ell+1}^{k} \eps_i \eps_{i+\ell } \bx_i \bx_{i+\ell}^\T.
	\end{align*}
	We can assume without loss of generality that $c=1$ in Assumption \ref{as-k}. We can also assume that $h<k$. Thus, we have for all $1\leq j \leq d$,
	\begin{align*}
	\sum_{\ell=0}^{h} K\left(\frac{\ell}{h}\right) \frac{1}{N-\ell} \sum_{i=1}^{k} (\eps_i \eps_{i+\ell } \bx_i\bx_{i+\ell}^\T - E\eps_i \eps_{i+\ell } \bx_i\bx_{i+\ell}^\T ) = \sum_{i=1}^{k} \boldsymbol{\xi}_{i,j},
	\end{align*}
	where
	$$
	\boldsymbol{\xi}_{i,j} = \sum_{\ell=0}^{h} \frac{1}{N-\ell} K\left(\frac{\ell}{h}\right) [\eps_i \eps_{i+\ell } \bx_i \bx_{i+\ell}^\T - E \eps_i \eps_{i+\ell } \bx_i \bx_{i+\ell}^\T].
	$$ By definition, $E\boldsymbol{\xi}_i = \mathbf{O}$, where $\mathbf{O}$ denotes the zero matrix. Using Assumption \ref{as-lin-ber-v},
	\begin{align*}
	\left(E \|  \boldsymbol{\xi}_{i,\ell} \|^\nu  \right)^{1/\nu} = \left(E \|  \boldsymbol{\xi}_{0,\ell} \|^\nu  \right)^{1/\nu} \leq c_8 \frac{h^{1/2}}{N}
	\end{align*}
	with some constant $c_8$. For any $a<b$ we get
	\begin{align*}
	\left(E \left \|  \sum_{i=a}^b \boldsymbol{\xi}_{i,\ell} \right \|^\nu  \right)^{1/\nu} \leq  \sum_{i=a}^{b} \left(E \|  \boldsymbol{\xi}_{i,\ell} \|^\nu  \right)^{1/\nu} \leq c_8 (b-a) \frac{h^{1/2}}{N}
	\end{align*}
	The maximal inequality of M\'oricz et al (1982) gives
	\begin{align*}
	E\left( \underset{1\leq k \leq N-1}{\max} \left \|  \sum_{i=1}^k \boldsymbol{\xi}_{i,\ell} \right \| \right)^\nu\leq c_9 h^{\nu/2}.
	\end{align*}
	Hence,
	\beq\label{332a}
	\begin{split}
		&\underset{1\leq k \leq N-1}{\max} \left\|   \left\{  \sum_{\ell=0}^{h} K \left(\frac{\ell}{h}\right) \frac{1}{N-\ell} \sum_{i=\ell}^{k} (\bx_i \bx_{i+\ell}^\T \eps_i \eps_{i+\ell} -E\bx_i \bx_{i+\ell}^\T \eps_i \eps_{i+\ell}  ) (\hat{\bbe}_N - \bbe)     \right\}  \right\| \\
		& = O_P\left(\frac{h}{N}\right) = o_P(1).
	\end{split}
	\eeq
	Assumption \ref{as-lin-ber-v} implies
	\begin{align*}
	\underset{1\leq k \leq N-1}{\max} \left \|  \left \{   \sum_{\ell=0}^{h} K\left( \frac{\ell}{h} \right) \frac{1}{N-\ell} \sum_{i=1}^{k} E \bx_i \bx_{i+\ell}^\T \eps_i \eps_{i+\ell}   \right \}	(\hat{\bbe}_N - \bbe_0) \right \| = o_P(1).
	\end{align*}
	By the triangle inequality,
	\begin{align*}
	& \left \| \sum_{\ell=0}^{h} K\left( \frac{\ell}{h}\right) \frac{1}{N-\ell} \sum_{i=k-\ell+1}^{k} \bx_i \bx_{i+\ell}^\T \eps_i \eps_{i+\ell} (\hat{\bbe}-\bbe_0) \right \| \\
	& \leq \frac{c_9}{N} \sum_{\ell=0}^{h} \sum_{i=1}^{d} \underset{1\leq k \leq N-1}{\max} \sum_{i=k-\ell+1}^{k} \| \bx_i \bx_i^\T \eps_i \eps_{i+\ell,j}\| \sum_{m=1}^{d} | \hat{\bbe}_{N,m} - \bbe_{0, m}|.
	\end{align*}
	It follows from Assumption \ref{as-lin-ber-v} that
	\begin{align*}
	E \left \|
	\sum_{i=k-\ell+1}^{k} \eps_i \eps_{i+\ell,j} \bx_i \bx_{i+\ell}^\T \right \|^\nu \leq c_{10} \ell^{\nu/2},
	\end{align*}
	and therefore by Markov's inequality,
	\begin{align*}
	P\left \{ \underset{1\leq k \leq N-1}{\max} \left\| \sum_{i=k+1-\ell}^{k}\eps_i \eps_{i+\ell,j} \bx_i\bx_{i+\ell}^\T  \right\| > x N^{1/\nu} \ell^{1/2}  \right \} \leq \frac{c_{10}}{x^\nu},
	\end{align*}
	resulting in
	\begin{align}\label{334a}
	E \underset{1\leq k \leq N-1}{\max} \left\| \sum_{i=k-\ell+1}^{k} \eps_i \eps_{i+\ell,j} \bx_i\bx_{i+\ell}^\T \right \| \leq c_{11} N^{1/\nu} \ell^{1/2}.
	\end{align}
	Thus, we conclude
	\begin{align*}
	& \underset{1\leq k \leq N-1}{\max} \left\| \sum_{\ell=0}^{h}K \left(\frac{\ell}{h}\right) \frac{1}{N-\ell} \sum_{i=k-\ell+1}^{k} \eps_i \bx_i \bx_{i+\ell}\eps_{i+\ell} (\hat{\bbe}_N - \bbe_0) \right\| \\
	& = O_P\left( \frac{h}{N} N^{1/\nu} h^{1/2} N^{-1/2}   \right) = O_P\left(  \left( \frac{h}{N^{1-2/(3\nu)}}  \right)^{3/2} \right) = o_P(1).
	\end{align*}
	We can repeat our arguments above and we get
	\begin{align*}
    &\underset{1\leq k \leq N-1}{\max} \left\| \sum_{\ell=0}^{N-1} K \left(\frac{\ell}{h}\right) \frac{1}{N-\ell} \sum_{i=1}^{k-\ell} \Bigg[ \bx_i\bx_{i+\ell}^\T \eps_i \eps_{i+\ell} (\hat{\bbe}_N - \bbe_0) \right. \\
    & \quad + \bx_i\bx_{i+\ell}^\T \eps_{i+\ell} \eps_{i} (\hat{\bbe}_N - \bbe_0) + \eps_i^2\eps_{i+\ell}^2\bx_i\bx_{i+\ell}^\T (\hat{\bbe}_N - \bbe_0)(\hat{\bbe}_N - \bbe_0)^\top \bx_{i+\ell}\bx_i^\T 
    \Bigg] \Bigg\| = o_P(1).
        \end{align*}
	Here we need to consider only
	\begin{align*}
	\hat{\bga}_{k,\ell} = \left\{\begin{matrix}
	\displaystyle \frac{1}{N-\ell}\sum_{i=1}^{k-\ell} \eps_i\eps_{i+\ell} \bx_i\bx_{i+\ell}^\T,  & 0 \leq \ell \leq k\\
	\displaystyle  \frac{1}{N-\ell}\sum_{i=-(\ell-1)}^{k} \eps_i\eps_{i+\ell} \bx_i\bx_{i+\ell}^\T, & -k < \ell < 0,
	\end{matrix}\right.
	\end{align*}
	$1\leq k \leq N-1$, and the corresponding estimators
	$$
	\hat{\bD}_N(k) = \sum_{\ell=0}^{k-1} K \left( \frac{\ell}{h}\right) \hat{\bga}_{k,\ell}.
	$$
	
	We write again
	\begin{align*}
	\sum_{\ell=0}^{h}K\left( \frac{\ell}{h}\right) \frac{1}{N-\ell} \sum_{i=1}^{k-\ell} \eps_i \eps_{i+\ell} \bx_i\bx_{i+\ell}^\T = & \sum_{\ell=0}^{h} K\left( \frac{\ell}{h}\right) \frac{1}{N-\ell} \sum_{i=1}^{k-1} \boldsymbol{\eta}_{i,\ell}\\
	& + \sum_{\ell=0}^{h} \frac{1}{N-\ell} \sum_{i=1}^{k-\ell} E \eps_i \eps_{i+\ell} \bx_i\bx_{i+\ell}^\T,
	\end{align*}
	and
	\begin{align*}
	\sum_{\ell=0}^{h}K\left( \frac{\ell}{h}\right) \frac{1}{N-\ell} \sum_{i=1}^{k-\ell} \boldsymbol{\eta}_{i,\ell} = & \sum_{\ell=0}^{h} K\left( \frac{\ell}{h}\right) \frac{1}{N-\ell} \sum_{i=1}^{k-1} \boldsymbol{\eta}_{i,\ell}\\
	& - \sum_{\ell=0}^{h} K\left( \frac{\ell}{h}\right)  \frac{1}{N-\ell} \sum_{i=k-\ell+1}^{k} \boldsymbol{\eta}_{i,\ell}.
	\end{align*}
	with $\boldsymbol{\eta}_{i,\ell} = \eps_i \eps_{i+\ell}\bx_i\bx_{i+\ell}^\T-E\eps_i \eps_{i+\ell}\bx_i\bx_{i+\ell}^\T$.
	Arguing as in the proof of
	\begin{align*}
	\underset{1\leq k \leq N-1}{\max}\left\| \sum_{\ell=0}^{h}K\left( \frac{\ell}{N}\right) \frac{1}{N-\ell} \sum_{i=k-\ell+1}^{k} \boldsymbol{\eta}_{i,\ell} \right \|= & O_P\left( h^{3/2} N^{1/\nu-1/2}\right)\\
	= & O_P\left(  \left( \frac{h}{N^{1/3-2/(3\nu)}}\right)^{3/2}\right).
	\end{align*}
	Following the proof of \eqref{332a},
	\begin{align*}
	\underset{1\leq k \leq N-1}{\max}\left\| \sum_{\ell=0}^{h}K\left( \frac{\ell}{N}\right) \frac{1}{N-\ell} \sum_{i=k-\ell+1}^{k} \boldsymbol{\eta}_{i,\ell} \right \|=  O_P\left( \frac{h}{N^{1/2}}\right)
	\end{align*}
	We only need to consider
	\begin{align*}
	\sum_{\ell=0}^{h} & K\left( \frac{\ell}{N}\right) \frac{1}{N-\ell} \sum_{i=1}^{k-\ell} E\eps_i\eps_{i+\ell}\bx_i \bx_{i+\ell}^\T \\  =  & \sum_{\ell=0}^{h}K\left( \frac{\ell}{N}\right) \frac{1}{N-\ell} \sum_{r=1}^{j-1}\sum_{i=m_{r-1}+1}^{m_r} E \eps_i \eps_{i+\ell} \bx_i \bx_{i+\ell}^\T \mathds{1}\{ m_{j-1}\leq k-\ell\} \\& +\sum_{\ell=0}^{h}K\left( \frac{\ell}{N}\right) \frac{1}{N-\ell} \sum_{i=m_{j-1}+1}^{k-\ell} E \eps_i \eps_{i+\ell} \bx_i \bx_{i+\ell}^\T \mathds{1}\{ m_{j-1}> k-\ell\},
	\end{align*}
	if $m_{j-1}<h \leq m_j$. We write
	\begin{align}\label{337a}
	\sum_{i=m_{r-1}+1}^{m_r} E\eps_i\eps_{i+\ell}\bx_i\bx_{i+\ell}^\T = \sum_{i = m_{r-1}+1}^{m_r-\ell} E\eps_i\eps_{i+\ell} \bx_i \bx_{i+\ell}^\T+\sum_{i = m_{r-1}-\ell+1}^{m_r} E\eps_i\eps_{i+\ell} \bx_i \bx_{i+\ell}^\T.
	\end{align}
	Since $i$ and $i+\ell$ share the same volatility in the first term of \eqref{337a}, we get
	\begin{align*}
	\sum_{\ell=0}^{h}K\left( \frac{\ell}{N}\right) & \frac{1}{N-\ell} \sum_{i=m_{r-1}+1}^{m_r} E\eps_i\eps_{i+\ell}\bx_i \bx_{i+\ell}^\T \rightarrow \\& (\tau_r - \tau_{r-1})\underset{N\rightarrow \infty}{\lim} \sum_{\ell=0}^{m_r - m_{r-1}} E \eps_{m_{r-1}+1}\eps_{m_{r-1}+1+\ell} \bx_{m_{r-1}+1}\bx_{m_{r-j+1}+\ell}^\T.
	\end{align*}
	For the second term of \eqref{337a}, we have that $i$ and $i+\ell$ are in different volatility regimes and
	\begin{align*}
	\underset{0\leq \ell \leq h}{\max} \left\| \sum_{i=m_{r-1}-\ell+1}^{m_r} E \eps_i \eps_{i+\ell} \bx_i \bx_{i+\ell}^\T   \right\| \leq c_{12}h
	\end{align*}
	and therefore
	\begin{align*}
	\left\| \sum_{\ell=0}^{h}K\left( \frac{\ell}{N}\right) \frac{1}{N-\ell} \sum_{i=m_{r-1}-\ell+1}^{m_r} E \eps_i \eps_{i+\ell} \bx_i \bx_{i+\ell}^\T \right \|=  O\left( \frac{h^2}{N}\right).
	\end{align*}
	By the arguments above
	\begin{align*}
	\sum_{\ell=0}^{h}K \left( \frac{\ell}{N}\right) & \frac{1}{N-\ell} \sum_{i=m_{j-1}+1}^{k-\ell} E \eps_i \eps_{i+\ell} \bx_i \bx_{i+\ell}^\T \rightarrow \\&
	(u-\tau_{j-1}) \underset{N\rightarrow \infty}{\lim} \sum_{\ell=0}^{m_j-m_{j-1}} \eps_{m_{j-1}+1}\eps_{m_{j-1}+\ell+1} \bx_{m_{j-1}+1}\bx_{m_{j-1}+\ell+1}^\T.
	\end{align*}
	We then have the same calculations for $\sum_{\ell=-h}^{0}$, and the proof is complete.
	\qed

	\noindent
	{\bf Proof of Theorem \ref{thebe0a}.}
	According to Assumption \ref{as-lin-ber-v}, $\{(\eps_i,\bx_i), m_{\ell-1}<i \leq m_\ell \}$ is stationary for each $\ell=1,\dots,M+1$. Hence, we get
	\begin{align*}
	\frac{1}{N}\bX_N^\T\bX_N = \sum_{\ell=1}^{M+1} \frac{1}{N} \sum_{i=m_{\ell-1}+1}^{m_\ell} \bx_i\bx_i^\T \overset{P}{\rightarrow} \sum_{\ell=1}^{M+1} (\tau_\ell - \tau_{\ell-1})\bA_\ell
	\end{align*}
	and for $1\leq k\leq d$,
	\begin{align*}
	\left| \sum_{i=1}^{N}x_{i,k}\eps_i\right | = \left| \sum_{\ell=1}^{M+1}\sum_{i=m_{\ell-1}}^{m_\ell}x_{i,k}\eps_i\right |  = O_P(N^{1/2}).
	\end{align*}
	Also,
	\begin{align*}
	\bX_N^\T\bY_N = \sum_{\ell=1}^{R+1} \left( \sum_{i=r_{\ell-1}+1}^{r_\ell} \bx_i \bx_i^\T  \right)\bbe_\ell + \bX_N^\T\bY_N = \sum_{\ell=1}^{R+1} \left( \sum_{i=r_{\ell-1}+1}^{r_\ell} \bx_i \bx_i ^\T \right) \bbe_\ell + O_P(N^{1/2}),
	\end{align*}
	and
	\begin{align*}
	\sum_{i=r_{\ell-1}+1}^{r_\ell} \bx_i \bx_i^\T = \sum_{j=1}^{M+1} \sum_{i=r_{\ell-1}+1}^{r_\ell} \bx_i\bx_i^\T \mathds{1}\{ m_{j-1}< i \leq m_j \},
	\end{align*}
	\begin{align*}
	\frac{1}{N}\sum_{j=1}^{M+1} \sum_{i=r_{\ell-1}+1}^{r_\ell} \bx_i\bx_i^\T \mathds{1}\{ m_{j-1}< i \leq m_j \} \overset{P}{\rightarrow} \sum_{j=1}^{M+1} \bA_j \left| (\tau_{\ell-1},\tau_\ell] \cap (\theta_{j-1},\theta_j]  \right|.
	\end{align*}
	Thus, we get
	\begin{align}\label{342a}
	\hat{\bbe}_N \overset{P}{\rightarrow} \bbe^{**}.
	\end{align}
	Also, if $r_{k-1}<\lf N u\rf \leq r_k$, then we have
	\begin{align*}
	\sum_{i=1}^{\lf Nt\rf} \hat{\eps}_i\bx_i = \sum_{i=1}^{\lf Nt \rf} \eps_i \bx_i + \sum_{j=1}^{k-1} \sum_{i=r_{j-1}+1}^{r_j} \bx_i\bx_i^\T(\bbe_j-\hat{\bbe}_N)+ \sum_{i=r_{k-1}+1}^{\lf Nt \rf} \bx_i \bx_i^\T (\bbe_k -\hat{\bbe}_N).
	\end{align*}
	We showed in the proof of Theorem \ref{lin-var-1}
	\begin{align*}
	\underset{0<t<1}{\sup} \left\| \sum_{i=1}^{\lf Nt\rf } \bx_i \eps_i \right\|= O_P(N^{1/2}).
	\end{align*}
	It is easy to see that \eqref{342a} implies
	\begin{align*}
	\frac{1}{N} \sum_{j=1}^{k-1}\sum_{i=r_{j-1}+1}^{r_j} \bx_i \bx_i^\T\left( \bbe_j-\hat{\bbe}_N \right) \overset{P}{\rightarrow} \sum_{j=1}^{k-1} \sum_{\ell=1}^{M+1}\bA_\ell | (\theta_{j-1},\theta_j] \cap (\nu_{\ell-1},\nu_\ell] | (\bbe_j-\bbe^{**}).
	\end{align*}
	This completes the proof.
	\qed

\section{Proofs of Theorems \ref{smo1} -- \ref{newedd}}\label{seprsm}
Similarly to $\bR(k)$ we introduce 
$$
{\bQ}_N(k)=\sum_{i=1}^k\bx_ig(i/N)\eps_i, \;\;\;\bQ_N(0)={\bf 0}.
$$

\begin{lemma}\label{sesm2}  If $H_0$, Assumptions \ref{as-m}--\ref{weps-mu} and \ref{totalvar} hold, then 
for each $N$ we can define  Gaussian processes $\{\bLamb_N(t), 0\leq t\leq 1\}$ such that
$$
\sup_{0\leq t \leq 1}\left\| N^{-1/2}\bQ_N(Nt)-\bLamb_N(t) \right\|=o_P(1),\;\;\;\mbox{where}\;\;\;\bLamb_N(t)=\int_0^tg(u)d\left(N^{-1/2}\hat{\bGamma}_N(Nu)\right),
$$
where $\{\hat{\bGamma}_N(x), 0\leq x \leq N\}$   is defined  in \eqref{Gamhat}.  Also, $E\bLamb_N(t)={\bf 0}$ and $E\bLamb_N(t)\bLamb^\T_N(s)=\bH(\min(t,s))$.
\end{lemma}
\begin{proof}  By Abel's summation formula we have 
\begin{align}\label{rep}
\bQ_N(k)=g(k/N)\bR(k)-\sum_{\ell=1}^{k-1}\bR(\ell)\left[g((\ell+1)/N)-g(\ell/N)\right], \;\;1\leq k \leq N.
\end{align}
Using the proof of Lemma \ref{stro2v} we get 
\begin{align*}
N^{-1/2}\max_{1\leq k \leq N}&\left\|\bQ_N(k)-\left( g(k/N)\hat{\bGamma}_N(k)-\sum_{\ell=1}^{k-1}\hat{\bGamma}_N(\ell)[g((\ell+1)/N)-g(\ell/N) ]  \right)
\right\|\\
&\leq N^{-1/2}\max_{1\leq k\leq N}|g(k/N)|\max_{1\leq k \leq N}\left\|N^{-1/2}\bQ(k)-\hat{\bGamma}_N(k)\right\|\\
&\hspace{1cm}  +N^{-1/2}\max_{1\leq k \leq N}\left\|\sum_{\ell=1}^{k-1} \left(  \bQ(\ell)-\hat{\bGamma}_N(\ell)\right)\left[g((\ell+1)/N)-g(\ell/N) \right]  \right\|\\
&=o_P(1)\sup_{0\leq t \leq 1}|g(t)|+o_P(1)\sum_{\ell=1}^{N-1}\left|g((\ell+1)/N)-g(\ell/N) \right| \\
&=o_P(1)
\end{align*}
on account of Assumption \ref{totalvar}.  Due to Assumption \ref{totalvar}, the Jordan decomposition theorem (cf.\ Hewett and Stromberg, 1969, p.\ 266)  yields that there are two  non decreasing functions $g_1(t)$ and $g_2(t)$ such that $g(t)=g_1(t)-g_2(t)$. Let
$$
\bGamma_N(t)=N^{-1/2}\hat{\bGamma}_N(Nt),\;\;0\leq t \leq 1.
$$
 Next we write
\begin{align}\label{jo}
\sum_{\ell=1}^{k-1}&{\bGamma}_N(\ell/N)[g_1((\ell+1)/N)-g_1(\ell/N) ]\\
& =\sum_{\ell=1}^{k-1}{\bGamma}_N(\ell/N)d\int_{\ell/N}^{(\ell+1)/N} dg_1(t)  \notag\\
&=\int_0^{k/N}{\bGamma}_N(t)dg_1(t)+\sum_{\ell=1}^{k-1}\int_{\ell/N}^{(\ell+1)/N}[{\bGamma}_N(\ell/N)-\bGamma_N(x)]dg_1(x).\notag
\end{align}
Let 
$
\{\Gamma_{N,k}(t), 0\leq t \leq 1\}, 1\leq k \leq d
$
denote the $k$th coordinate of $\bGamma_N(t)$. The covariance function of $\{\bGamma_N(t), 0\leq t \leq 1\}$ is $\bG(\min(t,s))$ of \eqref{gammacov}. This   yields that 
\begin{align*}
\{\Gamma_{N,k}(t), 0\leq t \leq 1\}\stackrel{\cD}{=}\left\{  \sum_{\ell=1}^j\sigma_{k,\ell}W_\ell(\tau_\ell-\tau_{\ell-1})+\sigma_{k,j+1}W_{j+1}(t-\tau_j) , \tau_{j-1}<t\leq \tau_{j}, 1\leq j \leq M +1\right\}
\end{align*}
where $W_1, W_2, \ldots , W_{M+1}$ are independent Wiener processes and $\sigma_{k,1}, \sigma_{k,2}, \ldots, \sigma_{k,M+1}$ are positive constants. Using the rate of continuity of the Wiener processes (cf.\ Horv\'ath and Rice, 2024, p.\ 514)  and the monotonicity of $g_1(u)$ we get 
\begin{align*}
\max_{1< k\leq M+1}&\left|\sum_{\ell=1}^{k-1}\int_\ell^{\ell+1}[\Gamma_{N,k}(\ell/N)-\Gamma_{N,k}(x/N)]dg_1(x/N)\right\|\\
&\leq (g_1(1)-g_1(0)) \max_{1\leq k \leq M+1}\sup_{0\leq x \leq 1}\sup_{|u|\leq 1/N}|\Gamma_{N,k}(x)-\Gamma_{N,k}(x+u)|\\
&\leq (g_1(1)-g_1(0))\left( \sum_{k=1}^d\sum_{\ell=1}^{M+1}\sigma_{k,\ell}\right)\sum_{\ell=1}^{M+1}\sup_{0\leq x \leq 1}\sup_{|u|\leq 1/N}|W_\ell(x+u)-W_\ell(x)|\\
&=o_P(1).
\end{align*}
Thus we conclude 
\begin{align*}
\max_{1< k\leq M+1}&\left\|\sum_{\ell=1}^{k-1}\int_{\ell/N}^{(\ell+1)/N}[\bGamma_N(\ell/N)-\bGamma_N(u)]dg_1(u)\right\|
=o_P(1).
\end{align*}
Integration by parts yields 
\begin{align*}
\bGamma_N(k/N)g_1(k/N)-\int_0^{k/N}\bGamma_N(x)dg_1(x)=\int_0^{k/N}g_1(x)d\bGamma_N(x).
\end{align*}
For every $N$ the process 
$$
\int_0^tg_1(x)d\bGamma_N(x)\;\;\;\mbox{has continous sample paths with probability 1}
$$
and therefore 
$$
\sup_{0\leq t \leq 1}\left\|  \int_{\lf Nt\rf/N}^tg_1(u)d\bGamma_N(u) \right\|=o_P(1).
$$
We obtain  the same estimates when $g_1$ is replaced with $g_2$ in the computations above.
Defining 
$$
\bLamb_N(t)=\int_0^tg(x)d\bGamma_N(x),
$$
the result is proven.
\end{proof}

Using the definition of $\hat{\eps}_i$ we have 
\begin{align}\label{hata}
\sum_{i=1}^k\bx_i \hat{\eps}_i-\frac{k}{N}\sum_{i=1}^N\bx_i\hat{\eps}_i
&=\sum_{i=1}^k\bx_ig(i/N){\eps}_i-\frac{k}{N}\sum_{i=1}^N\bx_ig(i/N){\eps}_i   \\
&\hspace{1cm}-\sum_{i=1}^k\left(\bx_i\bx_i^\T-E\bx_i\bx_i^\T\right)(\hat{\bbe}_N-\bbe_0)    \notag\\
&\hspace{1cm}+\frac{k}{N}\sum_{i=1}^N\left(\bx_i\bx_i^\T-E\bx_i\bx_i^\T\right)(\hat{\bbe}_N-\bbe_0)  \notag  \\
&\hspace{1cm}-\left(\sum_{i=1}^kE\bx_i\bx_i^\T-\frac{k}{N}\sum_{i=1}^NE\bx_i\bx_i^\T\right)(\hat{\bbe}_N-\bbe_0).  \notag 
\end{align}
Also,
\begin{align}\label{x0}
\hat{\bbe}_N-\bbe_0=\left(\bX_N^\T\bX_N\right)^{-1}\sum_{i=1}^N\bx_i g(i/N)\eps_i.
\end{align}

\begin{lemma}\label{lesm3} If $H_0$, Assumptions \ref{as-m}--\ref{asvarnosd} and \ref{totalvar} hold, then 

\beq\label{x3}
\max_{1\leq k\leq N}\left\|\sum_{i=1}^k\left(\bx_i\bx_i^\T-E\bx_i\bx_i^\T\right)\right\|=O_P\left(N^{1/2}\right),
\eeq

\beq\label{x4}
\hat{\bbe}_N-\bbe_0=\left(\left(  \sum_{\ell=1}^{M+1}(\tau_\ell-\tau_{\ell-1})\bA_\ell   \right)^{-1}+o_P(1)\right)\frac{1}{N}\sum_{i=1}^N\bx_ig(i/N) \eps_i,
\eeq
and
\beq\label{x44}
\left\|\hat{\bbe}_N-\bbe_0\right\|=O_P\left(N^{-1/2}\right).
\eeq
\end{lemma}
\begin{proof} The result in \eqref{x3} are proven in Lemma \ref{stro4v}. According to Lemma \ref{sesm2}
$$
\left\|\sum_{i=1}^N\bx_ig(i/N) \eps_i\right\|=O_P\left( N^{1/2} \right),
$$
\eqref{x4} and \eqref{x44} follow from \eqref{sesm2} and \eqref{x0}.
\end{proof}

\noindent
{\bf Proof of Theorem \ref{smo1}}.

Lemma \ref{lesm3} states that 
\begin{align*}
\max_{0\leq t\leq 1}\left\| \bQ_N(t)-\bLamb_N(t)   \right\|=o_P(1).
\end{align*}
Lemma \ref{stro4v} and \eqref{x44} yield 
$$
\max_{1\leq k\leq N}\left\|\sum_{i=1}^k\left(\bx_i\bx_i^\T-E\bx_i\bx_i^\T\right)(\hat{\bbe}_N-\bbe_0)\right\|    =O_P(1)
$$
and
$$
\max_{1\leq k \leq N}\left\|\frac{k}{N}\sum_{i=1}^N\left(\bx_i\bx_i^\T-E\bx_i\bx_i^\T\right)(\hat{\bbe}_N-\bbe_0) \right\|=O_P(1).
$$
Lemma \ref{lesm3} and \eqref{a22-2} imply
\begin{align*}
\sup_{0\leq t \leq 1}\left\|\frac{1}{N}\left(\sum_{i=1}^{\lf Nt\rf }E\bx_i\bx_i^\T-\frac{\lf Nt\rf}{N}\sum_{i=1}^NE\bx_i\bx_i^\T\right)N^{1/2}(\hat{\bbe}_N-\bbe_0)-\bv(t)N^{1/2}(\hat{\bbe}_N-\bbe_0)\right\|=o_P(1)
\end{align*}
and
\begin{align*}
\sup_{0\leq t\leq 1}\left\|\bv(t)N^{1/2}(\hat{\bbe}_N-\bbe_0)-\bv(t)\left(\sum_{\ell=1}^{M+1}(\tau_\ell-\tau_{\ell-1})\bA_\ell\right)^{-1}\bLamb_N(1)\right\|=o_P(1).
\end{align*}
Thuus we conclude
\begin{align*}
\sup_{0\leq t \leq 1}\left\|\bZ_N(t)-\left(\bLamb_N(t)-t\bLamb_N(1)-\bv(t) \left(\sum_{\ell=1}^{M+1}(\tau_\ell-\tau_{\ell-1})\bA_\ell\right)^{-1}\bLamb_N(1)\right)\right\|=o_P(1).
\end{align*}
Observing that
$$
\left\{  \bLamb_N(t)-t\bLamb_N(1)-\bv(t) \left(\sum_{\ell=1}^{M+1}(\tau_\ell-\tau_{\ell-1})\bA_\ell\right)^{-1}\bLamb_N(1) \right\}\stackrel{\cD}{=}\{\bUp(t), 0\leq t \leq 1\},
$$
the result is proven. \qed\\
\noindent

{\bf Proof of Theorem \ref{smo2}.}

We recall that in Lemma \ref{stromu} we defined  a sequence of Gaussian processes $\bGam_{N,1}(k), 1\leq k \leq m_1$
such that 
\beq\label{hus1}
\max_{1\leq k \leq m_1}\frac{1}{k^{\zeta}}\|\bR(k)-\bGam_{N,1}(k)\|=O_P(1).
\eeq
with some $\zeta<1/2$ and $\bGam_{N,1}(t)=0, E\bGam_{N,1}(t)\bGam^\T_{N,t}(s)=\bD_1\min(t,s)$. We use again \eqref{hata}. We observe that 
\begin{align*}
\max_{1\leq k \leq m_1}&\frac{1}{k^{\zeta}}\left\|\bQ_N(k)-\left( g(k/N)\bGam_{N,1}(k)-\sum_{\ell=1}^{k-1}\bGam_{N,1}(\ell)[g((\ell+1)/N)-g(\ell/N) ]  \right)
\right\|\\
&\leq \max_{1\leq k\leq N}|g(k/N)|\max_{1\leq k \leq m_1}\frac{1}{k^{\zeta}}\left\|\bR(k)-\bGam_{N,1}(k)\right\|\\
&\hspace{1cm}  +\max_{1\leq k \leq N}\frac{1}{k^{\zeta}}\left\|\sum_{\ell=1}^{k-1} \left(  \bR(\ell)-\bGam_{N,1}(\ell)\right)\left[g((\ell+1)/N)-g(\ell/N) \right]  \right\|\\
&=O_P(1)\sup_{0\leq t \leq 1}|g(t)|+\max_{1\leq k \leq m_1}\frac{1}{k^{\zeta}}\max_{1\leq \ell\leq k}\|\bR(\ell)-\bGam_{N,1}(\ell)\|\sum_{\ell=1}^{N-1}\left|g((\ell+1)/N)-g(\ell/N) \right| \\
&=O_P(1).
\end{align*}
The modulus of continuity of Wiener processes (cf.\ Horv\'ath and Rice, 2024, p.\ 514) yields in \eqref{jo}
\begin{align*}
\max_{1\leq k\leq m_1}&\frac{1}{k^{\zeta}}\left\|\sum_{\ell=1}^{k-1}\bGam_{N,1}(\ell)[g_1((\ell+1)/N)-g_1(\ell/N)]-\int_0^k\bGam_{N,1}(u)dg_1(u)\right\|\\
&\leq \max_{1\leq k\leq m_1} \frac{1}{k^{\zeta}} \left\|\sum_{\ell=1}^{k-1}\int_{\ell}^{\ell+1}[\bGam_{N,1}(\ell)-\bGam_{1,N}(x)]dg_1(x/N)\right\|\\
&=O_P((\log N)^{1/2}),
\end{align*}
since
\begin{align*}
\max_{1\leq k\leq m_1}\max_{|u|\leq 1}\left\| \bGam_{N,1}(k) -\bGam_{N,1}(k+u)\right\|=O_P((\log N)^{1/2}).
\end{align*}
Thus we conclude 
\begin{align*}
\max_{1\leq k\leq m_1}\frac{1}{k^{\zeta}}\left\| \bQ_N(k)-\int_0^kg(u/N)d\bGam_{N,1}(u)  \right\|=O_P(1).
\end{align*}
Now  we get
\begin{align*}
\sup_{1/N\leq t \leq \tau_1}&\frac{1}{t^{1/2}}\left\|  N^{-1/2}\bQ_N(Nt)-N^{-1/2}\int_0^{t} g(s)d\bGam_{N,1}(Ns)\right\|\\
&=\sup_{1/N\leq t \leq \tau_1}(Nt)^{-1/2}\left\|  N^{-1/2}\bQ_N(Nt)-N^{-1/2}\int_0^{t} g(s)d\bGam_{N,1}(Ns)\right\|\\
&\leq \max_{1\leq k \leq m_1}\frac{1}{k^\zeta}\left| \bQ_N(k)-\int_0^kg(u/N)d\bGam_{N,1}(u)   \right\| \sup_{1/N\leq t \leq \tau_1}(Nt)^{-1/2+\zeta}\\
&=O_P(1).
\end{align*}
Checking the covariances on can verify 
$$
\left\{ N^{-1/2}\int_0^{t} g(s)d\bGam_{N,1}(Ns), 0\leq t\leq \tau_1  \right\} \stackrel{\cD}{=}\left\{\bLamb(t),\;\;\;0\leq t\leq \tau_1\right\}.
$$
We use again that the coordinates of $\bLamb(t)$ are  distributed as constants times $W(\mca(t))$ on $0\leq t \leq \tau_1$. Hence by the law of the iterated logarithm 
\begin{align*}
\sup_{0<t\leq \tau_1}\left(\frac{1}{\mca(t)\log\log(1/\mca(t))}\right)^{1/2}\|  \bLamb(t)\|=O_P(1).
\end{align*}
Using Assumption \ref{ww1} we get for all $x>0$
\begin{align}\label{della1}
\lim_{\delta\to 0}P\left\{ \sup_{0<t\leq \delta}\frac{1}{t^{\alpha_1}}\|\bLamb(t)\|>x  \right\}=0
\end{align}
and therefore the approximation yields 
\begin{align}\label{della2}
\lim_{\delta\to 0}\limsup_{N\to\infty}P\left\{ \sup_{0<t\leq \delta}\frac{1}{t^{\alpha_1}}\|N^{-1/2}\bQ_N(t)\|>x  \right\}=0.
\end{align}
We get from \eqref{della1} and \eqref{della2}
\begin{align}\label{della3}
\lim_{\delta\to 0}P\left\{ \sup_{0<t\leq \delta}\frac{1}{t^{\alpha_1}}\|\bLamb(t)-t\bLamb(1)\|>x  \right\}=0
\end{align}
and
\begin{align}
\lim_{\delta\to 0}\limsup_{N\to\infty}P\left\{ \sup_{1/N\leq t\leq \delta}\frac{1}{t^{\alpha_1}}\left\|N^{-1/2}\bQ_N(Nt)-tN^{-1/2}\bQ_N(N)\right\|>x  \right\}=0.
\end{align}
By symmetry we also have

\begin{align}\label{della5}
\lim_{\delta\to 0}P\left\{ \sup_{1-\delta\leq t <1}\frac{1}{(1-t)^{\alpha_2}}\|\bLamb(t)-t\bLamb(1)\|>x  \right\}=0
\end{align}
and
\begin{align}\label{della4}
\lim_{\delta\to 0}\limsup_{N\to\infty}P\left\{ \sup_{1-\delta\leq t <1-1/N}\frac{1}{(1-t)^{\alpha_2}}\left\|N^{-1/2}\bQ_N(Nt)-tN^{-1/2}\bQ_N(N)\right\|>x  \right\}=0.
\end{align}
  Along the proof of Theorem \ref{smo1} one can verify 
\begin{align*}
\max_{1/N\leq t \leq 1-1/N}\frac{1}{t^{\alpha_1}(1-t)^{\alpha_2}}\left\|\left(\sum_{i=1}^{\lf Nt \rf}[\bx_i\bx_i^\T-E\bx_i\bx_i^\T]-t    \sum_{i=1}^{N}[\bx_i\bx_i^\T-E\bx_i\bx_i^\T]\right)      (\hat{\bbe}_N-\bbe_0)   \right\|
=o_P\left(N^{1/2}\right),
\end{align*}

\begin{align*}
\lim_{\delta\to 0}\limsup_{N\to \infty} P\left\{\sup_{1/N\leq t \leq \delta}\frac{N^{-1/2}}{t^{\alpha_1}}\left\|\left(\sum_{i=1}^kE\bx_i\bx_i^\T-\frac{k}{N}\sum_{i=1}^NE\bx_i\bx_i^\T\right)(\hat{\bbe}_N-\bbe_0)\right\|>x\right\}=0
\end{align*}
and
\begin{align*}
\lim_{\delta\to 0}\limsup_{N\to \infty} P\left\{\sup_{1-\delta\leq t \leq 1-1/N}\frac{N^{-1/2}}{(1-t)^{\alpha_2}}\left\|\left(\sum_{i=1}^kE\bx_i\bx_i^\T-\frac{k}{N}\sum_{i=1}^NE\bx_i\bx_i^\T\right)(\hat{\bbe}_N-\bbe_0)\right\|>x\right\}=0
\end{align*}
for all $x>0$. Since $\sup_{0<t<1/N}\|\bZ_N(t)\|=\sup_{1-1/N\leq t <1}\|\bZ_N(t)\|=0$, the result follows from Theorem \ref{smo1}.

\medskip
{\bf Proof of Theorem \ref{newed}.}
\begin{lemma}\label{lemgregus} If $H_0$, Assumptions \ref{as-m}--\ref{asminons},  \ref{totalvar} and \eqref{pow} hold, then we have 
\begin{align*}
\lim_{N\to \infty}P&\Biggl\{ a(\log N) \max_{1\leq k \leq N-1}\left[ \left(\bQ_N(k)-\frac{k}{N}\bQ_N(N)\right)^\T  \bar{\bH}^{-1}(k/N)  \left(\bQ_N(k)-\frac{k}{N}\bQ_N(N)\right)    \right]^{1/2}\\
&\hspace{1cm}\leq  x+b_d(\log N) \Biggl\}=\exp\left( -\frac{2\varrho+2}{2\varrho+1}e^{-x}  \right)
\end{align*}
for all $x$.
\end{lemma}

\begin{proof} We can and will assume without loss of generality that $c=1$.
We use Lemma \ref{stromu} and define
$$
\bLamb_{N,1}(t)=\int_0^{t} g(x/N)d\bGam_{N,1}(x). 
$$
It is easy to see that 
$$
E\left(\bLamb_{N,1}(t)\right)^2=\bD_1\int_0^tg^2(x/N)dx=\bD_1\cA_{N,1}(t),\;\;\;\;1\leq t \leq m_1, 
$$
with
\begin{align*}
\cA_{N}(t)=\frac{N^{-2\varrho}}{2\varrho+1}t^{2\varrho+1}.
\end{align*}
Let
$$
\cz(x)=\sum_{i=1}^{\lf x\rf}\bx_i\eps_i,\;\;\; 1\leq x \leq N.
$$
Integration by parts yields for all $0\leq t \leq m_1$
\begin{align*}
&\int_0^tg(x/N)d\left(\cz(x)-\bGam_{N,1}(x)\right)\\
&=g(t/N)\left(\cz(t)-\bGam_{N,1}(t)\right)-c_2N^{-\varrho}\int_0^t \left(\cz(x)-\bGam_{N,1}(x)\right)\varrho x^{\varrho-1}dx,
\end{align*}
and Lemma \ref{stromu} implies
\begin{align*}
N^{\varrho}\max_{1\leq t \leq m_1}\frac{1}{t^\beta}\left\| \int_0^tg(x/N)d\left(\cz(x)-\bGam_{N,1}(x)\right)  \right\|=O_P(1)
\end{align*}
with $\beta=\varrho+\zeta$. Thus we have 
\beq\label{laj}
N^{\varrho}\max_{1\leq t \leq m_1}\frac{1}{t^\beta}\left\|\bQ_N(\lf t\rf]-\bLamb_{N,1}(t)\right\|=O_P(1).
\eeq
Checking the covariances one can verify that 
\begin{align}\label{edrep100}
& \left\{  \frac{1}{\cA_{N}(t)}\bLamb^\T_{N,1}(t)  \bD_1^{-1}  \bLamb_{N,1}(t), \;1\leq t\leq m_1\right\}  \notag  \\
&\stackrel{\cD}{=} \left\{ \frac{1}{\cA_{N}(t)}\sum_{i=1}^d W^2_i\left(\cA_{N}(t)\right), \;1\leq t \leq m_1   \right\},
\end{align}
where $\{W_1(t), 0\leq t \leq m_1\}, \ldots, \{W_d(t), 0\leq t \leq m_1\}$ are independent Wiener processes.  We note that 
$$
\lim_{t\to \infty}\frac{N^{-\varrho}t^{\varrho+\zeta}}{\cA^{1/2}_{N}(t)}=0
$$
for all $N$. The approximation in \eqref{laj} implies that
\beq\label{lajo}
\max_{1\leq x \leq m_1}\frac{1}{\cA^{1/2}_N(x)}\left\| \bQ_N(x)-\bLamb_{N,1}(x)    \right\|=O_P(1)
\eeq
and therefore arguing as in the proof of Lemma \ref{stromu}, we get  
\begin{align}\label{newed1}
&\left|\max_{1\leq k\leq m_1}\left[\frac{1}{\cA_N(k)}  \bQ_N^\T(k)\bD_1^{-1} \bQ_N(k)   \right] ^{1/2}-\sup_{1\leq t\leq m_1}
 \left[\frac{1}{\cA_N(t)} \bLamb^\T_{N,1}(t)\bD_1^{-1} \bLamb_{N,1}(t)  \right] ^{1/2}\right|    \notag   \\
&=o_P\left((\log\log N)^{-1/2}\right)
\end{align}
and
\begin{align}\label{newed2}
&\Biggl|\max_{1\leq k\leq m_1}\frac{1}{\cA_N(k)} \left[ \left(   \bQ_N(k)-\frac{k}{N}\bQ_N(N)\right)^\T\bD_1^{-1} \left(\bQ_N(k)-\frac{k}{N}\bQ_N(N) \right)   \right] ^{1/2}  \notag   \\
&\hspace{2cm} -\sup_{1\leq t\leq m_1}
\frac{1}{\cA_N(t)} \left[ \bLamb^\T_{N,1}(t)\bD_1^{-1} \bLamb_{N,1}(t)  \right] ^{1/2}\Biggl|     \notag \\
&=o_P\left((\log\log N)^{-1/2}\right).
\end{align}
The representation in \eqref{edrep100} implies
\begin{align*}
\sup_{1\leq t \leq m_1}&\left[\frac{1}{\cA_N(t)} \bLamb^\T_{N,1}(t)\bD_1^{-1} \bLamb_{N,1}(t)  \right] ^{1/2}\\
&\stackrel{\cD}{=} \sup_{1\leq t\leq m_1}\left( \frac{1}{\cA_{N}(t)}\sum_{i=1}^d W^2_i\left(\cA_{N}(t)\right)\right)^{1/2}
=\sup_{\cA_N(1)\leq x\leq \cA_N(m_1)}\left( \frac{1}{x}\sum_{i=1}^d W_i^2(x)   \right)^{1/2}.
\end{align*}
Using again Lemma A.3.1 of Cs\"org\H{o} and Horv\'ath (1997) (cf.\ also Theorem A.2.7  in Horv\'ath and Rice, 2024) we get 
\begin{align}\label{ched1}
\lim_{N\to \infty}P &  \left\{ a(\log u_N)\sup_{\cA_N(1)\leq t \leq \cA_N(m_1)}\left[\frac{1}{t}\sum_{i=1}^dW_i^2(t)\right]^{1/2} \leq x+b_d(\log u_N)\right\}  \notag \\
&=\exp(-e^{-x}), 
\end{align}
where
$$
u_N=\frac{\cA_N(m_1)}{\cA_N(1)}=\tau_1^{1+2\varrho}N^{1+2\varrho}.
$$
Elementary arguments yield
\begin{align*}
&\left( (2\log\log N)^{1/2}-(2\log\log u_N))^{1/2}    \right)(2\log\log u_N)^{1/2}\\
&\hspace{.6cm}=\left(\frac{\log\log u_N}{\log\log \iota_N}\right)^{1/2}\left(  \log\log N -\log\log u_N  \right)\\
&\hspace{.6cm}=\left(\frac{\log\log u_N}{\log\log \iota_N}\right)^{1/2}   \log[\log N /\log u_N]
\to -\log(2\varrho +1),
\end{align*}
where $\iota_N$ is between $N$ and $u_N$. Also,
\begin{align*}
\left| b_d(\log N)-b_d(u_N)  \right|\to 0.
\end{align*}
We observe that 
$$
\left( (2\log \log N)^{1/2}-(2\log \log u_N)^{1/2}   \right)\sup_{\cA_N(1)\leq t \leq \cA_N(m_1)}\left[\frac{1}{t}\sum_{i=1}^dW^2_i(t)\right]^{1/2}\stackrel{P}{\to} -\log(1+2\varrho).
$$
Hence we can rewrite \eqref{ched1} as 
\begin{align}\label{ched11}
\lim_{N\to \infty}P &   \Biggl\{ a(\log N)\sup_{\cA_N(1)\leq t \leq \cA_N(m_1)}\left[\frac{1}{t}\sum_{i=1}^dW_i^2(t)\right]^{1/2}\leq x+\log(2\varrho+1)+b_d(\log N)\Biggl\} \notag \\
& =\exp(-e^{-x}).
\end{align}
Thus, we conclude
\begin{align}\label{ched12}
\lim_{N\to \infty}&P\Biggl\{ a(\log N)\sup_{1\leq k \leq m_1}\left[\frac{1}{\cA_N(k)}\left( \bQ_N(k)-\frac{k}{N}\bQ_N(N)   \right)^\T\bD_1^{-1}\left( \bQ_N(k)-\frac{k}{N}\bQ_N(N)\right)   \right]^{1/2}\\
&\hspace{1cm}\leq x+\log(2\varrho+1)+b_d(\log N)\Biggl\}=\exp(-e^{-x}).\notag
\end{align}

Next we note
$$
\bQ_N(N)-\bQ_N(x)=\sum_{i=x+1}^Ng(i/N)\bx_i\eps_i=\int_x^Ng(t/N)\cz(t)=-\int_x^Ng(t/N)d(\cz_N(t)-\cz_N(N)).
$$
We also introduce 
$$
\bar{\cA}_N(x)=\int_x^Ng^2(t/N)dt=\frac{N^{-2\varrho}}{{2\varrho+1}}\left(N^{2\varrho+1}-x^{2\varrho+1}\right).
$$
Applying again Lemma \ref{stromu} we get along the lines of \eqref{newed1} and \eqref{newed2}  that
\begin{align*}
\Biggl|\max_{m_{M}\leq k \leq N-1}&\left[\frac{1}{\bar{\cA}_N(k)}\left(\bQ_N(N)-\bQ_N(k)\right)^\T\bD_{M+1}^{-1}\left(\bQ_N(N)-\bQ_N(k)\right)\right]^{1/2}  \\
& \hspace{1cm}-\max_{m_M\leq x \leq N-1}\left[\frac{1}{\bar{\cA}_N(k)}\bLamb^\T_{N,M+1}(x)\bD_{M+1}^{-1}\bLamb_{N,M+1}(x)\right]^{1/2}\Biggl|\\
&=o_P((\log\log N)^{-1/2})
\end{align*}
and
\begin{align*}
\Biggl|\max_{m_{M}\leq k \leq N-1}&\left[\frac{1}{\bar{\cA}_N(k)}\left(\bQ_N(k)-\frac{k}{N}\bQ_N(N)\right)^\T\bD_{M+1}^{-1}\left(\bQ_N(k)-\frac{k}{N}\bQ_N(N)\right) \right]^{1/2} \\
& \hspace{1cm}-\max_{m_M\leq x \leq N-1}\left[\frac{1}{\bar{\cA}_N(x)}\bLamb^\T_{N,M+1}(x)\bD_{M+1}^{-1}\bLamb_{N,M+1}(x)\right]^{1/2}\Biggl|\\
&=o_P((\log\log N)^{-1/2}),
\end{align*}
where
$$
\bLamb_{N,M+1}(x)=\int_x^N g(t/N)d\bGam_{N,M+1}(N-t),\;\;\;m_M\leq x \leq N.
$$
Observing that 
\begin{align*}
\max_{m_M\leq x \leq N-1}\left[\frac{1}{\bar{\cA}_N(x)}\bLamb^\T_{N,M+1}(x)\bD_{M+1}^{-1}\bLamb_{N,M+1}(x)\right]^{1/2}\stackrel{\cD}{=}\max_{\bar{\cA}_N(N-1)) \leq x \leq \bar{\cA}_N(m_M)}
\left[\frac{1}{x}\sum_{i=1}^d W^2_i(x)\right]^{1/2},
\end{align*}
we obtain by Lemma A.3.1 of Cs\"org\H{o} and Horv\'ath (1997)  (cf.\ also Horv\'ath and Rice, 2024, pp.\  520)  that
\begin{align} \label{greg1}
\lim_{N\to \infty}P & \left\{  a(\log v_N) \max_{\bar{\cA}_N(N-1)) \leq x \leq \bar{\cA}_N(m_M)}\left[\frac{1}{x}\sum_{i=1}^d W^2_i(x)\right]^{1/2}\leq x +b_d(\log v_N) \right\} \notag \\
& =\exp(-e^{-x})  
\end{align}
with
$$
v_N=\frac{\bar{\cA}(m_M)}{\bar{\cA}_N(N-1)}=N\frac{1-\tau_M^{2\varrho+1}}{2\varrho+1}+O(1),\;\;\;\mbox{as}\;\;\;N\to \infty.
$$
Elementary arguments yield
$$
|a(\log v_N)-a(\log N)|(\log\log N)\to 0,
$$

$$
|b_d(\log v_N)-b_d(\log N)|\to 0
$$
and
$$
|a(\log v_N)-a(\log N)|\max_{\bar{\cA}_N(N-1)) \leq x \leq \bar{\cA}_N(m_M)}\left[\frac{1}{x}\sum_{i=1}^d W^2_i(x)\right]^{1/2}\stackrel{P}{\to}0.
$$
Thus we can rewrite \eqref{greg1} as 
\begin{align*} 
\lim_{N\to \infty}P\left\{  a(\log N) \max_{\bar{\cA}_N(N-1)) \leq x \leq \bar{\cA}_N(m_M)}\left[\frac{1}{x}\sum_{i=1}^d W^2_i(x)\right]^{1/2}\leq x +b_d(\log N) \right\}=\exp(-e^{-x}). 
\end{align*}
It follows from Theorem \ref{smo2} that 
$$
\max_{m_1\leq t\leq m_M} \left(\bQ_N(k)-\frac{k}{N}\bQ_N(N)\right)^\T  \bar{\bH}_N^{-1}(k) \left(\bQ_N(k)-\frac{k}{N}\bQ_N(N)\right)=O_P(1).
$$
Since $\{\bLamb_{N,1}(t), 1\leq t\leq m_1\}$ and $\{\bLamb_{N,M+1}(t), m_M\leq t \leq N\}$ are independent, the proof is complete.
\end{proof}

\medskip
\noindent
{\it Proof of Theorem \ref{newed}.}  We follow the proof of Theorem \ref{coverd}. We use the representation in \eqref{hata} instead of \eqref{decomp} to show that 
\begin{align*}
\Biggl| &\max_{0<t<1}\left[ \bZ_N^\T(t)\bar{H}^{-1}(t)\bZ_N(t)  \right]^{1/2}  \\
&\hspace{1cm}-\max_{1\leq k \leq N-1}\left[\left(\bQ_N(k)-\frac{k}{N}\bQ_N(N)\right)^\T\bar{\bH}^{-1}(k/N)  \left(\bQ_N(k)-\frac{k}{N}\bQ_N(N)\right) \right]^{1/2}
\Biggl|\\
&=o_P\left((\log\log N)^{-1/2}  \right). 
\end{align*}
The result now follows from Lemma \ref{lemgregus}.
\qed

\medskip
\begin{lemma}\label{lemgreguska} If $H_0$, Assumptions \ref{as-m}--\ref{asminons}, \ref{totalvar} and \eqref{pow1} hold, then we have 
\begin{align*}
\lim_{N\to \infty}P&\Biggl\{ a(\log N) \max_{1\leq k \leq N-1}\left[ \left(\bQ_N(k)-\frac{k}{N}\bQ_N(N)\right)^\T  \bar{\bH}^{-1}(k/N)  \left(\bQ_N(k)-\frac{k}{N}\bQ_N(N)\right)    \right]^{1/2}\\
&\hspace{1cm}\leq  x+b_d(\log N) \Biggl\}=\exp\left( -2e^{-x}  \right)
\end{align*}
for all $x$.
\end{lemma}

\begin{proof} Under the present conditions 
\begin{align}\label{upp}
\cA_{N}(t)&=t\left[  c_1^2+\frac{2c_1c_2}{\varrho+1}N^{-\varrho} t^\varrho +\frac{c_2^2}{2\varrho+1}N^{-2\varrho}t^{2\varrho} \right]\\
&=t\left[ \left(c_1+\frac{c_2}{\varrho+1}\left(  \frac{t}{N} \right)^\varrho\right)^2   +\frac{c_2^2\varrho^2}{(2\varrho+1)(\varrho+1)^2}\left(  \frac{t}{N}   \right)^{2\varrho} \right],\;\;\;1\leq t\leq m_1,  \notag
\end{align}
and therefore if $c_1\neq 0$, 
$$
c_1^2t\leq \cA_N(t)\leq t\left[  \left(c_1+\frac{c_2}{\varrho+1}\tau_1^\varrho\right)^2+\frac{c_2^2\varrho^2}{(2\varrho+1)(\varrho+1)^2}\tau_1^{2\varrho}      \right].
$$
Hence \eqref{lajo} holds which implies \eqref{newed1}. We get from \eqref{upp} that
$$
u_N=\frac{\cA_N(m_1)}{\cA_N(1)}=N\frac{\tau_1}{c_1^2}\left[  \left(c_1+\frac{c_2}{\varrho+1}\tau_1^\varrho\right)^2+\frac{c_2^2\varrho^2}{(2\varrho+1)(\varrho+1)^2}\tau_1^{2\varrho}      \right]\left(1+O\left(N^{-\varrho}\right)\right).
$$
Now 
$$
|a(\log N)-a(\log u_N)|\sup_{1\leq t\leq m_1}\left[  \frac{1}{t}\sum_{i=1}^dW_i^2(t)  \right]^{1/2}\stackrel{P}{\to}0
$$
and
$$
|b_d(\log N)-b_d(\log u_N)|\to 0.
$$
Hence \eqref{ched12} can be  replaced with 
\begin{align*}
\lim_{N\to \infty}&P\Biggl\{ a(\log N)\sup_{1\leq k \leq m_1}\left[\frac{1}{\cA_N(k)}\left( \bQ_N(k)-\frac{k}{N}\bQ_N(N)   \right)^\T\bD_1^{-1}\left( \bQ_N(k)-\frac{k}{N}\bQ_N(N)\right)   \right]^{1/2}\\
&\hspace{1cm}\leq x+b_d(\log N)\Biggl\}=\exp(-e^{-x}).\notag
\end{align*}
We note that the proof,  when we maximize on $[m_M, N-1]$,  is the same as in the proof of Lemma \ref{lemgregus} and therefore
\begin{align*}
\lim_{N\to \infty}&P\Biggl\{ a(\log N)\sup_{m_M\leq k \leq N-1}\left[\frac{1}{\bar{\cA}_N(k)}\left( \bQ_N(k)-\frac{k}{N}\bQ_N(N)   \right)^\T\bD_1^{-1}\left( \bQ_N(k)-\frac{k}{N}\bQ_N(N)\right)   \right]^{1/2}\\
&\hspace{1cm}\leq x+b_d(\log N)\Biggl\}=\exp(-e^{-x}).\notag
\end{align*}
The rest of the arguments are the same as in the proof of Lemma \ref{lemgreguska}.
\end{proof}
  
\medskip
\noindent
{\it Proof of Theorem \ref{newedd}.}  It is the same as of Theorem \ref{newed}, we just need to replace Lemma \ref{lemgregus} with Lemma \ref{lemgreguska}.
\qed

\medskip

\section{Computation of Critical Values in Theorems~\ref{lin-var-1} and~\ref{secmoca1} }\label{secriticalvalue}
\subsection{Implementation of Theorem \ref{lin-var-1}}~\\
In order to implement the proposed tests in finite samples, in this section we discuss computational methods to obtain critical values of the test statistics in Theorem \ref{lin-var-1}. The critical values of the tests in Theorem \ref{lin-var-1} need to be simulated based on the covariance structure estimated from the data. According to Theorem \ref{lin-var-1}, when $0 \leq \kappa <1/2$, the critical values may be simulated as supremum functionals of a Gaussian process $\bar{\bGam}(t)$ with covariance structure $\bar{\bG}(t,s)$ shown in \eqref{barbgc}. Here we approximate the Gaussian process $\bar{\bGam}(t)$ using Karhunen--Lo\'{e}ve expansions. Given $\bar{\bG}(t,s)$ a non--negative definite function, we define eigenvalues $\lambda_1\geq \lambda_2 \geq \dots$ and the corresponding eigenfunctions $\boldsymbol{\phi}_1(t)$, $\boldsymbol{\phi}_2(t)$, $\dots$ valued in $\mathbb{R}^d$ such that
\begin{align*}
\int_0^1 \boldsymbol{\phi}_i^\top(t) \boldsymbol{\phi}_j(t) dt = \left\{\begin{matrix}
1, \quad i = j,
\\
0, \quad i\neq j,
\end{matrix}\right.
\end{align*}
and
\begin{align*}
\lambda_i \boldsymbol{\phi}_i(t) = \int_{0}^{1} \bar{\bG}(t,s) \boldsymbol{\phi}_i(s) ds,
\end{align*}
where the intergral is carried out coordinatewise. We then have by Mercer's theorem that
\begin{align*}
\bar{\bG}(t,s) = \sum_{\ell=1}^{\infty} \lambda_i \boldsymbol{\phi}_\ell(t)\boldsymbol{\phi}^\top_\ell(s).
\end{align*}
Since $\bar{\bG}$ is continuous, we have convergence of the above representation both in $L^2$ and also in supremum norm. The Gaussian process $\bar{\bGam}(t)$ then admits the following Karhunen--Lo\'{e}ve expansion
\begin{align*}
\bar{\bGam}(t) = \sum_{\ell=1}^{\infty} \lambda_\ell^{1/2} \mathcal{N}_\ell \boldsymbol{\phi}_\ell(t),
\end{align*}
where $\mathcal{N}_1$, $\mathcal{N}_2$, $\dots$ are independent standard normal random variables. In finite samples, we use the plug--in estimator $\tilde{\bG}_N(t,s)$ by estimating $\bar{\bG}(t,s)$ in \eqref{barbgc} with $\hat{\bG}_N(t)$ in \eqref{hatgdef} and $\bu_N(t)$ in \eqref{unt}. According to Theorem \ref{hetcos}, we thus have
\begin{align*}
\tilde{\lambda}_{i,N} \tilde{ \boldsymbol{\phi}}_{i,N}(t) = \int_{0}^{1} \tilde{\bG}_N(t,s) \tilde{\boldsymbol{\phi} }_{i,N}(s) ds, \quad 1\leq i \leq N,
\end{align*}
for $\tilde{\lambda}_{1,N}\geq \tilde{\lambda}_{2,N} \geq \dots$ and
\begin{align*}
\int_0^1 \tilde{\boldsymbol{\phi}}_{i,N}^\top(t) \tilde{\boldsymbol{\phi}}_{j,N}(t) dt = \left\{\begin{matrix}
1, \quad i = j,
\\
0, \quad i\neq j.
\end{matrix}\right.
\end{align*}
Hence, we have the following approximation,
\begin{align*}
\tilde{\bGam}(t) = \sum_{\ell=1}^{L} \lambda_{\ell,N}^{1/2} \mathcal{N}_\ell \boldsymbol{\phi}_{\ell,N}(t),
\end{align*}
where the process $\tilde{\bGam}(t)$ is approximately distributed as $\bar{\bGam}(t)$ when the truncation parameter $L$ is suitably large. The limits and corresponding critical values of the tests $V^{HET}_N(\kappa)$ and $Q^{HET}_N(\kappa)$ in Theorem \ref{lin-var-1} thereby can be approximated via $ \left\| \tilde{\bGam}(t)  \right\|/w(t)$ and $  \left\| \tilde{\bGam}(t)  \right\|_\infty/w(t)$, respectively, using simulation. In Online supplement, we also provide a finite-sample approximation of the standardized Darling--Erd\H{o}s type statistics in Theorem \ref{coverd}.
 
\subsection{Implementation of Theorem \ref{smo1} and \ref{smo2} }\label{secmoca1}~\\
In this section we suggest some possible ways simulate the critical values of tests based on Theorems \ref{smo1} and \ref{smo2}. 
It follows from the definition of the Gaussian process $\bLamb(t)$ that  $E(\bLamb(t)-t\bLamb(1))={\bf 0}$ and 
\begin{align*}
\bK(t,s)=E\left[(\bLamb(t)-t\bLamb(1))(\bLamb(s)-s\bLamb(1))^\T\right]=\bH(\min(t,s))-s\bH(t)-t\bH(s)+ts\bH(1).
\end{align*}
Statistical inference on $\bH(u)$ will imply immediately estimators for $\bK(u)$.  We suggest a long run kernel estimator for $\bH(u)$. Let
\begin{align*}
\hat{\bC}_{N,t}(\ell)=
\left\{
\begin{array}{ll}
\displaystyle \frac{1}{N}\sum_{i=1}^{\lf Nt\rf-\ell}(\bx_i\hat{\eps}_i)(\bx_{i+\ell}\hat{\eps}_{i+\ell})^\T, \;\;\mbox{if}\;\;0\leq \ell\leq \lf Nt\rf-1,
\vspace{.3cm}\\
\displaystyle \frac{1}{N} \sum_{i=-\ell+1}^{\lf Nt\rf}(\bx_i\hat{\eps}_i)(\bx_{i+\ell}\hat{\eps}_{i+\ell})^\T, \;\;\mbox{if}\;\;-(\lf Nt\rf  -1)\leq \ell <0
\end{array}
\right.
\end{align*}
and define
\begin{align*}
\hat{\bH}_N(t)=\sum_{h=-(\lf Nt \rf-1)}^{\lf Nt\rf -1}\cW\left( \frac{\ell}{h} \right)\hat{\bC}_{N,t}(\ell),
\end{align*}
where $\cW$ is a kernel.  We note if $t$ is close to 0, then $\hat{\bH}_N(t)$ is computed from few observations so it is not reliable. \\
The formulas for $\bH(u)$   greatly simplify 
\begin{assumption}\label{asuncor}
\begin{align}\label{uncor1}
\{\bx_i, -\infty<\infty\}\;\;\mbox{and}\;\;\{\eps_i, -\infty<i<\infty\}\;\; \mbox{are independent }
\end{align}
and
\begin{align}\label{uncor2}
E\eps_i\eps_j=0,\;\;i\neq j.
\end{align}
\end{assumption}
hold. Under these conditions
\begin{align*}
\bH(u)=\bA \ch(u),
\end{align*}
where 
$$
\ch(u) =\sum_{i=1}^{k-1} \sigma_i^2 \int_{\tau_{i-1}}^{\tau_i}g^2(u)du+\sigma_k^2\int_{\tau_{k-1}}^tg^2(u)du,\;\;\mbox{if} \;\;\tau_{k-1}<u\leq \tau_k, 1\leq k\leq M+1,
$$
and
$$
\sigma_i^2=\frac{1}{\tau_i-\tau_{i-1}}\sum_{j=\tau_{i-1}+1}^{\tau_i}E\eps_j^2, \;\;1\leq i \leq M+1.
$$
Under Assumption \ref{asuncor}   we suggest using
$$
\hat{\ch}_N(u)=\frac{1}{N}\sum_{j=1}^{\lf Nu\rf}\hat{\eps}_j^2
$$
and
$$
\tilde{\bH}_N(u)=\frac{1}{N}\sum_{j=1}^{\lf Nu\rf}\bx_j\bx_j^\T\hat{\eps}_j^2.
$$

\begin{theorem}\label{thuncor} If $H_0$, Assumptions \ref{as-m}, \ref{asminons}--\ref{asvarnosd}, \ref{totalvar} and \ref{asuncor} hold, then 
\beq\label{nab1}
\sup_{0\leq u \leq 1}|\hat{\ch}_N(u)-\ch(u)|\stackrel{P}{\to} 0
\eeq
and
\beq\label{nab2}
\sup_{0<u<1}\|\tilde{\bH}_N(u)-\bH(u)\|\stackrel{P}{\to}0.
\eeq
\end{theorem}
An other possible estimator for $\bH(u)$ is 
$$
\bH_N^*(u)=\left(\frac{1}{N}\bX_N^\T\bX_N  \right)^{-1}\hat{\ch}_N(u).
$$

The next result is an immediate consequence of Theorem \ref{smo1} and the form of $\bH(u)$, if Assumption \ref{asuncor} holds:

\begin{corollary}\label{newco}  If $H_0$, Assumptions \ref{as-m}, \ref{asminons}--\ref{asvarnosd}, \ref{totalvar} \ref{ww1}  and \ref{asuncor} hold, then 
$$
\sup_{0\leq u\leq 1}\frac{1}{u^{2\alpha_1}}\frac{1}{(1-u)^{2\alpha_2}} \bP_N^\T(u)\left(\frac{1}{N}\bX_N^\T\bX_N\right)^{-1}\bP_N(u)\stackrel{\cD}{\to}\sup_{0\leq u \leq 1}
\frac{\|\bPsi(u)\|^2}{u^{2\alpha_1}(1-u)^{2\alpha_2}},
$$
where $\bPsi(u)=(\Psi_1(u), \Psi_2(u), \ldots, \Psi_d(u))^\T$,
$$
\Psi_k(u)=W_k(\ch(t))-tW_k(\ch(1)), \;\;\;0\leq t \leq 1, \;\;1\leq k \leq d,
$$
and $\{W_1(u), u\geq 1\}, \{W_2(u), u\geq 1\}, \ldots, \{W_d(u), u\geq 0\}$ are independent Wiener processes. 
\end{corollary}

Due to Theorem \ref{thuncor}, it is relatively simple to simulate $\bPsi(u)$, since $\ch(u)$ can be replaced with $\hat{\ch}_N(u)$.

\medskip

\section{Additional Monte Carlo simulation results}\label{sec_sim_supplement}
\subsection{Simulation settings}\hfill\label{sec-app-sim} 

For the purpose of comparing the tests under heteroscedasticity, we also denote the following test statistics
\begin{align*}
& V^{HO}_N(\kappa) =\sup_{0<t<1}\frac{1}{(t(1-t))^\kappa}\left(\bZ_N^\T(t)\hat{\bD}_N^{-1}\bZ_N(t)\right)^{1/2},\;\;\; 0 \leq \kappa<1/2, \mbox{ and }\\
& Q^{HO}_N(\kappa) =\sup_{0<t<1}\frac{1}{(t(1-t))^\kappa}\left\|\hat{\bD}_N^{-1/2}\bZ_N(t)\right\|_\infty, \;\;\; 0 \leq \kappa<1/2.
\end{align*}
We use this notation to remind the reader that these statistics are built under the assumption of homoscedasticity of the model errors and covariates, where their asymptotic results have been discussed in Horv\'ath et al. (2023).

The long run covariance matrix estimator for $\bD$ is defined as
\begin{align*}
\hat{\bD}_N=\sum_{\ell=-(N-1)}^{N-1}K\left( \frac{\ell}{h}  \right)\hat{\bga}_\ell,
\end{align*}
where the autocovariance matrix of weighted residuals at lag $\ell$ is estimated via:
\begin{align*}
\hat{\bga}_\ell=\left\{
\begin{array}{ll}
\displaystyle \frac{1}{N-\ell}\sum_{i=1}^{N-\ell}\bx_i\hat{\eps}_i\bx_{i+\ell}^\T\hat{\eps}_{i+\ell},\;\;\;\mbox{if}\;\;0\leq \ell<N,
\vspace{0.1cm}\\
\displaystyle\frac{1}{N-|\ell|}\sum_{i=-(\ell-1)}^N\bx_i\hat{\eps}_i\bx_{i+\ell}^\T\hat{\eps}_{i+\ell} ,\;\;\;\mbox{if}\;\;-N< \ell<0.
\end{array}
\right.
\end{align*}
As in Horv\'ath et al. (2023), the critical values of the tests $V^{HO}_N(\kappa)$ and $Q^{HO}_N(\kappa)$ may be obtained by simulating their limits when $0\leq \kappa<1/2$, and from the Darling-Erd\H{os} limit result when $\kappa=1/2$.

\subsection{Homoscedastic models}\hfill\label{sec-ho}

We first consider the case that the error terms $\epsilon_i$ in \eqref{sim-mod} are homoscedastic. We consider errors generated as: (i) (\textbf{Normal}) the error terms $\epsilon_i$ are i.i.d normal random variables:
\begin{equation*}
\epsilon_i \sim N(0,1),\quad 1 \leq i\leq N.
\end{equation*}
(ii) (\textbf{AR}) the error terms $\epsilon_i$ follow autoregressive (AR-1) process:
\begin{equation*}
\epsilon_i = 0.5 \epsilon_{i-1} + \varepsilon_i, \;\;  \mbox{with } \varepsilon_i \sim N(0,1), \;\; 1 \leq i\leq N.
\end{equation*}
(iii) (\textbf{GARCH}) the error terms $\epsilon_i$ follow a stationary GARCH(1,1) process defined by
\begin{equation*}
\epsilon_i = h_i^{1/2} \varepsilon_i, \mbox{ with }  h_i = 0.1 + 0.01 \epsilon_i^2 + 0.9 h_{i-1},\quad 1 \leq i\leq N,
\end{equation*}
and the $\varepsilon_i$'s are i.i.d. standard normal random variables.

Figure \ref{figure-homo-th22-early} displays the size and power of the Darling--Erd\H{o}s type tests based on $V^{HO}_N(1/2)$, $Q^{HO}_N(1/2)$, $V_N^{HET}(1/2)$ and $Q_N^{HET}(1/2)$ for these settings when $k^*=\lfloor 0.2N\rfloor$. Figure \ref{figure-homo-th22-mid} complements the results of these statistics for detecting a change $k^* = \lfloor 0.5N  \rfloor$. %

In this setting where the error process was taken to be homoscedastic, we observed that the empirical size was close to nominal for each tests considered, and improved with increasing sample size. In terms of power, we observe the expected increasing power of each of the tests as a function of the size of change parameterized by $\delta$, sample size, and centrality of the change point within the sample. We observed that the weighted CUSUM statistics in Horv\'ath et al. (2023), $V_N^{HO}(1/2)$ and $Q_N^{HO}(1/2)$, exhibited generally higher power to detect changes points than $V^{HET}_N(1/2)$ and $Q^{HET}_N(1/2)$ in small samples, while their power all quickly converges to unity in large samples. The test $Q^{HET}_N(1/2)$ was somewhat over-sized in smaller samples when the model is generated with \textbf{AR} errors, but this effect diminished with larger samples. 

\begin{figure}[h]
	\centering
	\caption{Rejection rates of $V^{HO}_N(1/2)$, $Q^{HO}_N(1/2)$, $V^{HET}_N(1/2)$ and $Q^{HET}_N(1/2)$ at 95\% significance level with sample size $N=125$ (first row) $N=250$ (second row), and $500$ (third row). The error term follows {\bf homoscedastic} \textbf{Normal}, \textbf{AR} and \textbf{GARCH} errors. The change point $k^* = \lfloor 0.2N  \rfloor$.}
	\label{figure-homo-th22-early}
	\includegraphics[width=14.1cm]{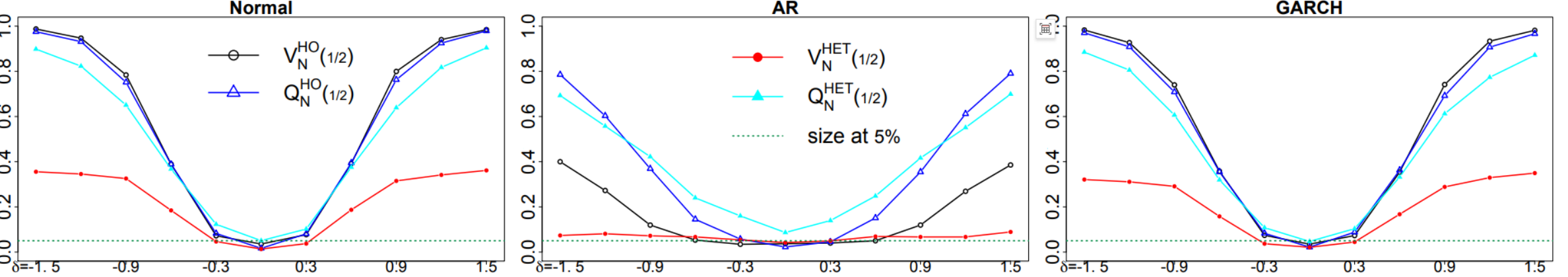}
	\includegraphics[width=14.2cm]{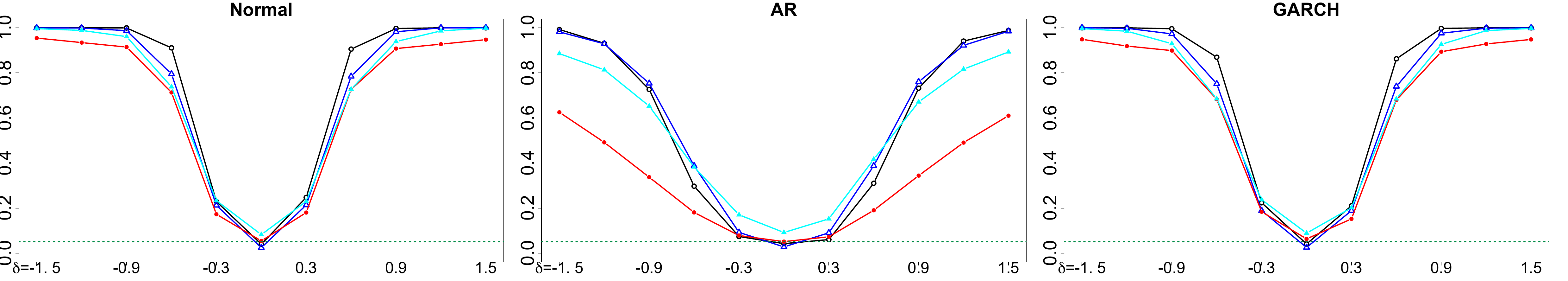}
	\includegraphics[width=14.2cm]{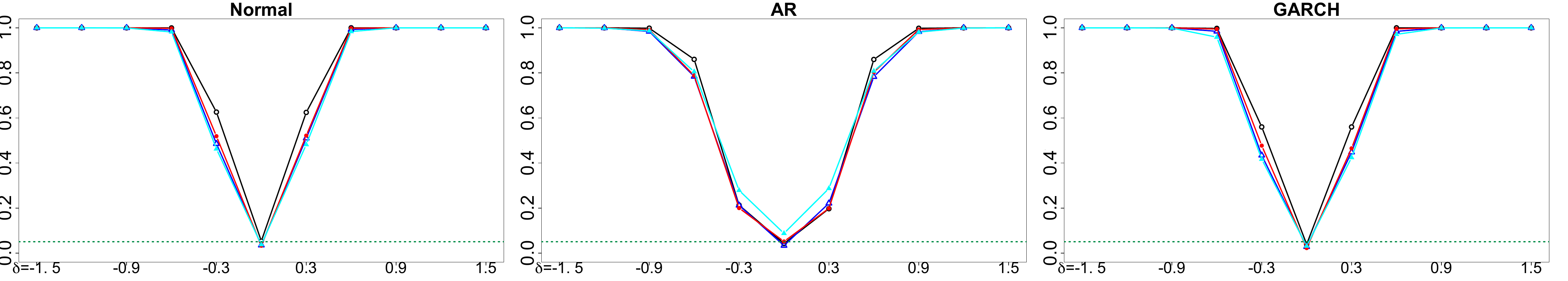}
\end{figure}

\clearpage

\begin{figure}[H]
	\centering
	\caption{Rejection rates of $V^{HO}_N(1/2)$, $Q^{HO}_N(1/2)$, $V_N^{HET}(1/2)$ and $Q_N^{HET}(1/2)$ at 95\% significance level with sample size $N=125$ (first row) $N=250$ (second row), and $500$ (third row). The error term follows {\bf homoscedastic} \textbf{Normal}, \textbf{AR} and \textbf{GARCH} errors. The change point $k^* = \lfloor 0.5N  \rfloor$.}
	\label{figure-homo-th22-mid}
	\includegraphics[width=14.2cm]{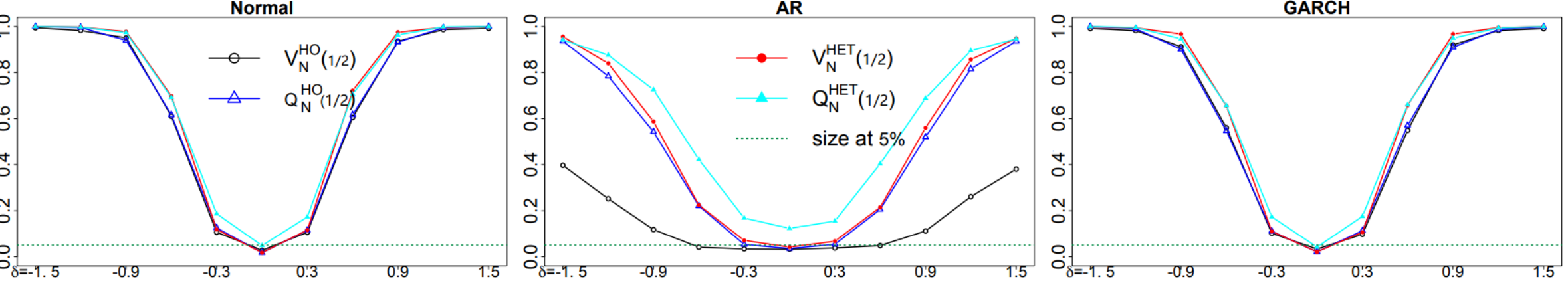}
	\includegraphics[width=14.2cm]{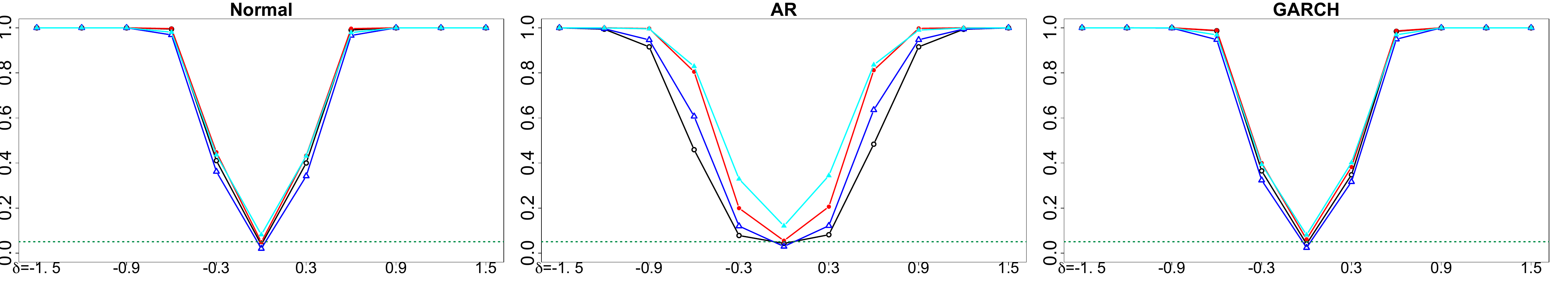}
	\includegraphics[width=14.2cm]{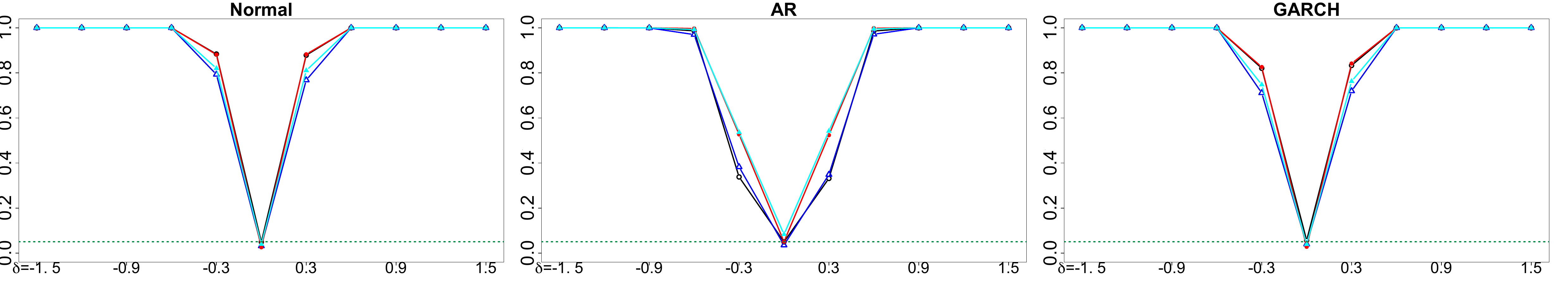}
\end{figure}

\subsection{Heteroscedastic models}\hfill\label{sec-hetesim}
 
In this subsection, we first consider three data generating process for which the errors are heteroscedastic. These are similar to \textbf{Normal}, \textbf{AR} and \textbf{GARCH} above, but include a change in the variance:

(iv) (\textbf{Normal}) the error terms $\epsilon_i$ are i.i.d normal random variables:
\begin{equation*}
\epsilon_i = \left\{\begin{matrix}
N(0,0.5),\quad 1\leq i\leq \lfloor m^* N  \rfloor,
\\
N(0,2), \quad \lfloor m^* N  \rfloor<i\leq N.
\end{matrix}\right.
\end{equation*}

(v) (\textbf{AR}) the error terms $\epsilon_i$ follow autoregressive (AR-1) process:
\begin{equation*}
\epsilon_i = \left\{\begin{matrix}
0.5 \epsilon_{i-1} + \varepsilon^{(1)}_i,\;\;  \mbox{with } \varepsilon^{(1)}_i \sim N(0,0.5), \;\;  1\leq i\leq \lfloor m^* N  \rfloor,
\\
0.5 \epsilon_{i-1} + \varepsilon^{(2)}_i,\;\;  \mbox{with } \varepsilon^{(2)}_i \sim N(0,2), \;\; \lfloor m^* N  \rfloor<i\leq N.
\end{matrix}\right.
\end{equation*}

(vi) (\textbf{GARCH}) the error terms $\epsilon_i$ follow a stationary GARCH(1,1) process defined by
\begin{equation*}
\epsilon_i = h_i^{1/2} \varepsilon_i, \mbox{ with }  h_i =
\left\{\begin{matrix}
0.5 + 0.01 \epsilon_i^2 + 0.9 h_{i-1},\;\;1\leq i\leq \lfloor m^* N  \rfloor,
\\
2 + 0.01 \epsilon_i^2 + 0.9 h_{i-1}, \;\; \lfloor m^* N  \rfloor<i\leq N.
\end{matrix}\right.
\end{equation*}
and the $\varepsilon_i$'s are i.i.d. standard normal random variables. In this section, we always set $m^*=\lf 0.5N \rf$, i.e., there is a variance change in the middle of the sample.

We also consider heteroscedastic covariates that $\mathbf{x}_i = (1, x_{2,i})^\top$ are generated as:
\begin{equation}\label{eq-hete-x}
x_{2,i} = \left\{\begin{matrix}
N(0,0.3),& i\leq N/3,
\\
N(0,2), & (N/3+1)\leq i \leq 2N/3,
\\
N(0,1), & i > 2N/3.
\end{matrix}\right.
\end{equation}

Figure \ref{figure-hetex-th22-early}--\ref{figure-hetex-th22-mid} show the empirical size and power curves of candidate tests. We found the tests $V^{HO}_N(1/2)$ and $Q^{HO}_N(1/2)$ were over-sized, and this distortion became more apparent in larger samples. The proposed tests $V^{HET}_N(1/2)$ and $Q^{HET}_N(1/2)$ exhibited approximately nominal size, higher power even for changes occurring closer to the boundary and smaller sample size.

\begin{figure}[h]
	\centering
	\caption{Rejection rates of $V^{HO}_N(1/2)$, $Q^{HO}_N(1/2)$, $V^{HET}_N(1/2)$ and $Q^{HET}_N(1/2)$ at 95\% significance level with sample size $N=125$ (first row) $N=250$ (second row), and $500$ (third row). The models are generated with heteroscedastic covariates and error terms that follow heteroscedastic \textbf{Normal}, \textbf{AR}, and \textbf{GARCH} distributions. The change point $k^* = \lfloor 0.2N  \rfloor$.}
	\label{figure-hetex-th22-early}
	\includegraphics[width=14.2cm]{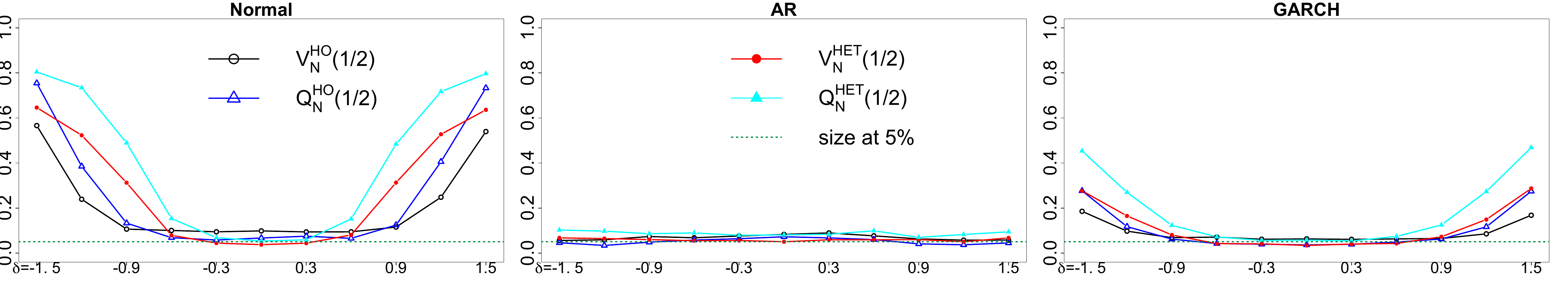}
	\includegraphics[width=14.2cm]{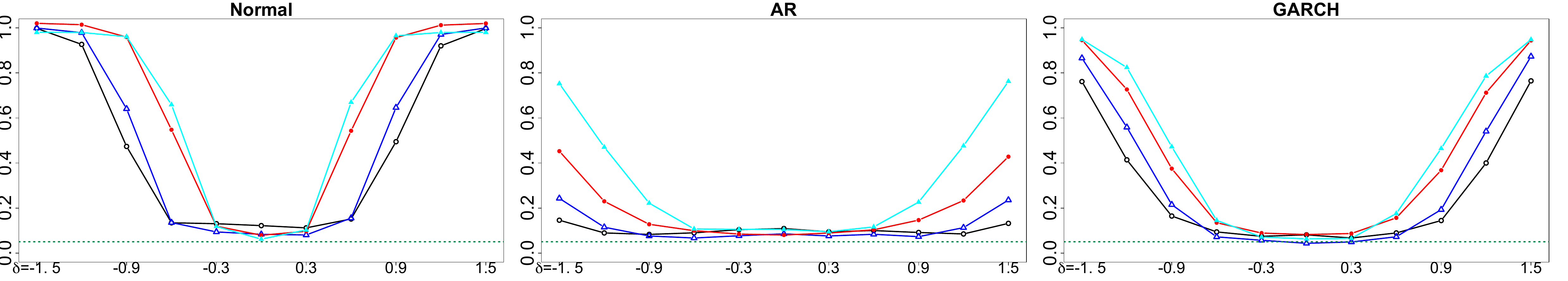}
	\includegraphics[width=14.2cm]{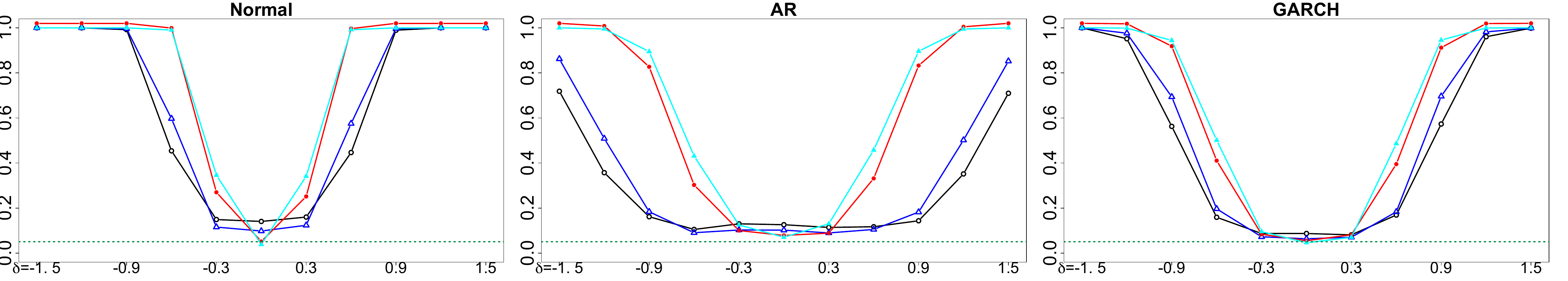}
\end{figure}

\clearpage

\begin{figure}[H]
	\centering
	\caption{Rejection rates of $V^{HO}_N(1/2)$, $Q^{HO}_N(1/2)$, $V_N^{HET}(1/2)$ and $Q_N^{HET}(1/2)$ at 95\% significance level with sample size $N=125$ (first row) $N=250$ (second row), and $500$ (third row). The models are generated with heteroscedastic covariates and error terms that follow heteroscedastic \textbf{Normal}, \textbf{AR}, and \textbf{GARCH} distributions. The change point $k^* = \lfloor 0.5N  \rfloor$.}
	\label{figure-hetex-th22-mid}
	\includegraphics[width=14.2cm]{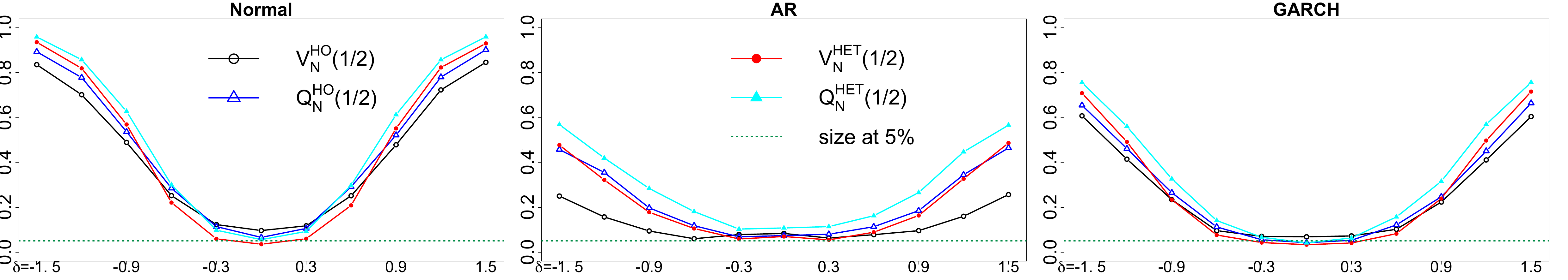}
	\includegraphics[width=14.2cm]{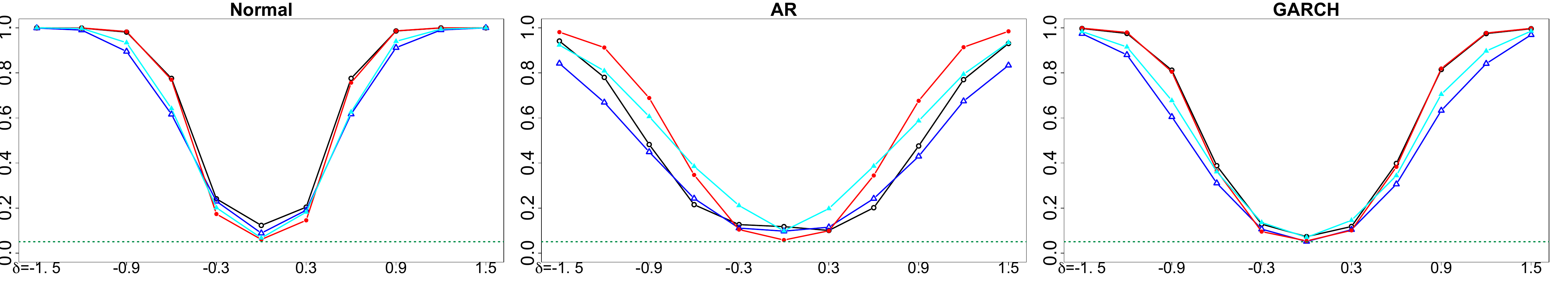}
	\includegraphics[width=14.2cm]{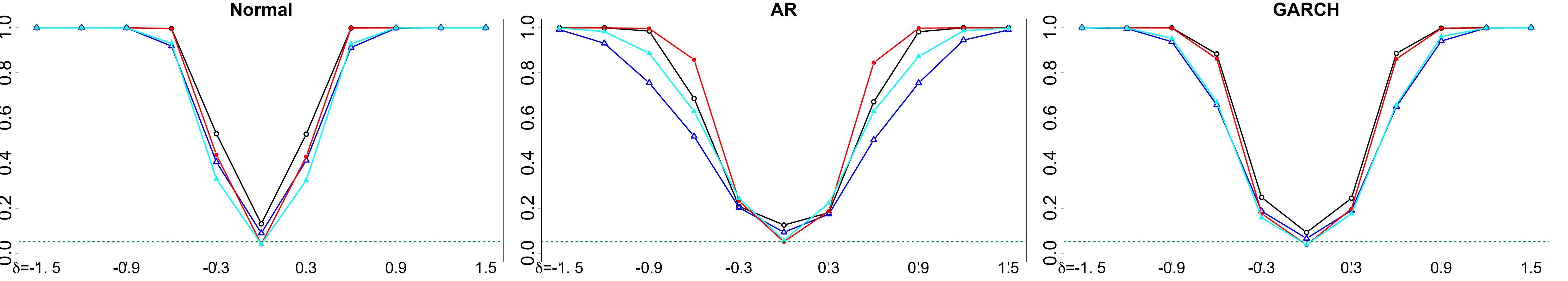}
\end{figure}

\subsection{Effects of $\kappa$, change locations and signal-to-noise ratio}\label{kap}\hfill

In this section, we focus on investigating the effect of several commonly discussed issues on change point tests\footnote{In an unreported simulation study, we study the effect of the choice of truncation parameter $L$ in computing the critical values of the tests $V_N^{HET}(\kappa)$ and $Q_N^{HET}\kappa)$, for $0\leq \kappa<1/2$. By simulation the limits under $H_0$ with $L\in[1, 3, 5, 10, 20]$, we find setting $L=5$ is adequate to approximate the limit distribution. Hence, in the present simulation, we universally set $L=5$. }. We first study the effect of the choice of the weight tunning parameter $\kappa$ for the weighted CUSUM tests $V^{HO}_N(\kappa)$ and $Q^{HO}_N(\kappa)$, $V_N^{HET}(\kappa)$ and $Q_N^{HET}\kappa)$, for $0< \kappa\leq 1/2$. Figures \ref{figure-weight-early} and \ref{figure-weight-mid} display the power curves of the tests when there occurs an early and a middle change, respectively, with the weight parameter $\kappa=\{0.15, 0.3, 0.45\}$. Looking over the over-size problem in $V^{HO}_N(\kappa)$ and $Q^{HO}_N(\kappa)$ as we have discussed in the previous subsections, we find the tuning parameter $\kappa$ with a value close to $0$ or $1/2$ gain higher power in detecting center or boundary change points, respectively. This effect is more manifest in smaller samples, and we thereby omit the results in large samples.

\begin{figure}[H]
	\centering
	\caption{Power functions of the weighted CUSUM tests $V^{HO}_N(\kappa)$, $V_N^{HET}(\kappa)$, $Q^{HO}_N(\kappa)$ and $Q_N^{HET}(\kappa)$ with nominal significance level  5\%, with heteroscedastic {\bf Normal}, {\bf AR} and {\bf GARCH} errors, $\kappa=0.15, 0.30, 0.45$ and sample size $N=125$ when the change occurs at $k^*= \lfloor 0.2N  \rfloor$.}
	\label{figure-weight-early}
	\includegraphics[width=14.5cm]{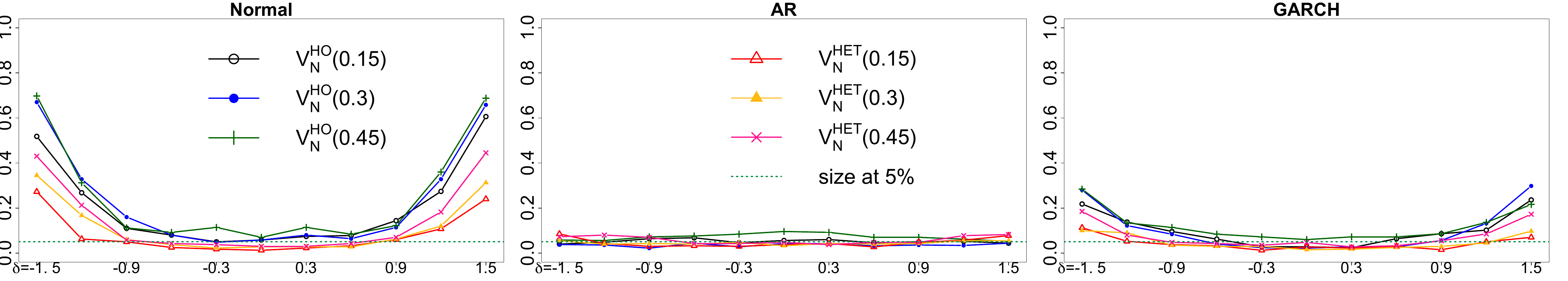}
	\includegraphics[width=14.5cm]{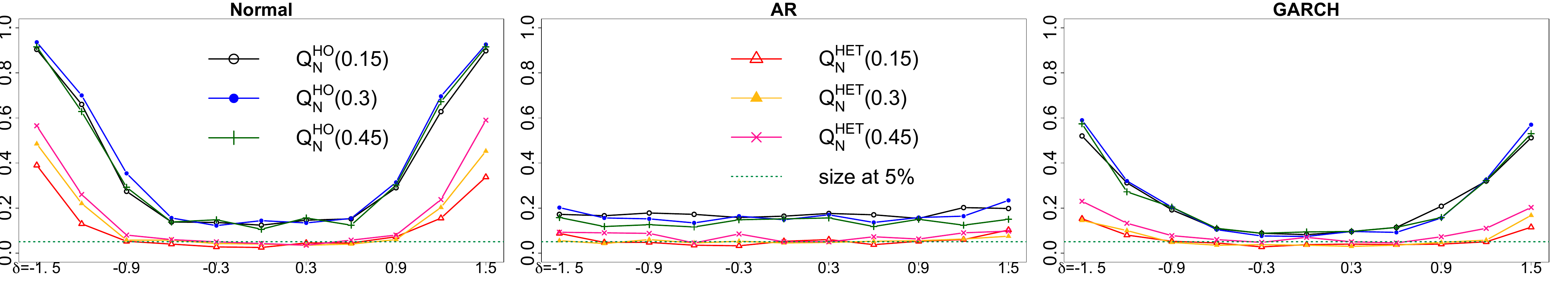}
\end{figure}

\begin{figure}[H]
	\centering
	\caption{Power functions of the weighted CUSUM tests $V^{HO}_N(\kappa)$, $V_N^{HET}(\kappa)$, $Q^{HO}_N(\kappa)$ and $Q_N^{HET}(\kappa)$ with nominal significance level 5\%, with heteroscedastic {\bf Normal}, {\bf AR} and {\bf GARCH} errors, $\kappa=0.15, 0.30, 0.45$ and sample size $N=125$ when the change occurs at $k^*= \lfloor 0.5N  \rfloor$.}
	\label{figure-weight-mid}
	\includegraphics[width=14.5cm]{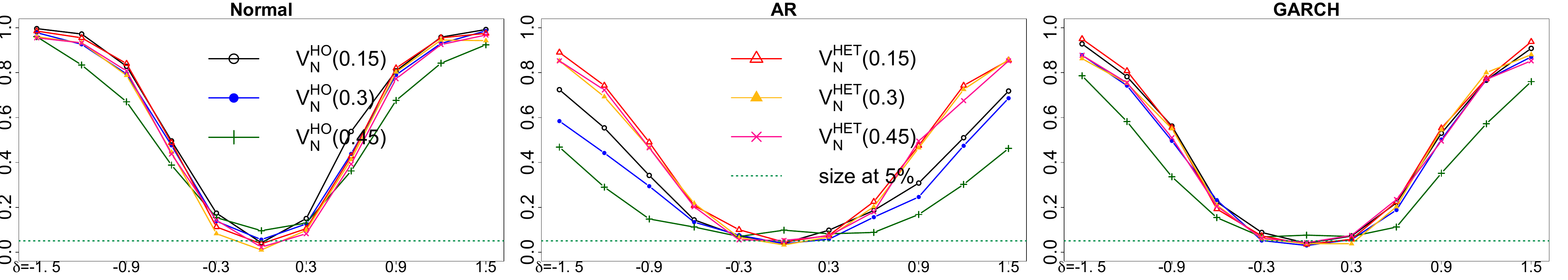}
	\includegraphics[width=14.5cm]{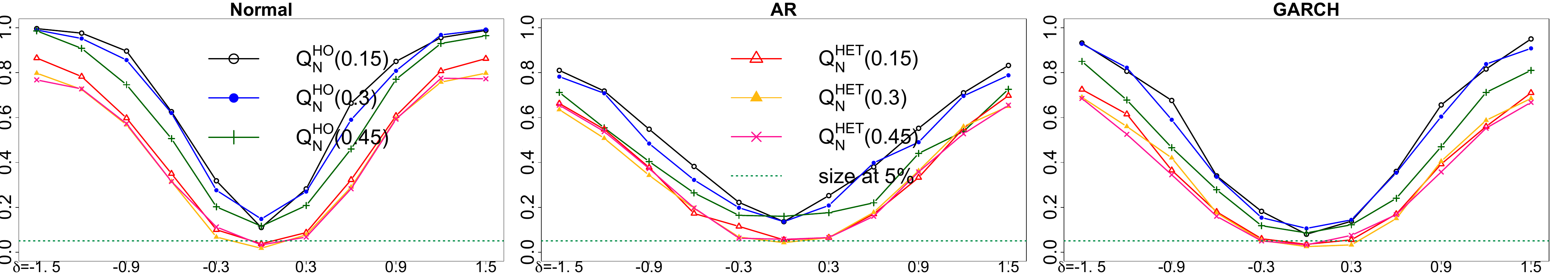}
\end{figure}

Second, we investigate the effect of the location of variance change in the heteroscedascity structure, i.e., $m^*$ on our weighted CUSUM tests $V_N^{HET}(0)$, $Q_N^{HET}(0)$ and Darling--Erd\H{o}s type tests $V_N^{HET}(1/2)$ and $Q_N^{HET}(1/2)$. In the previous heteroscedastic data generating processes, we set $m^*=0.5$. Here we vary the value of $m^*\in[0.1, 0.2, 0.3, 0.4, 0.5, 0.6, 0.7, 0.8, 0.9]$. The size of change point $\delta=0.6$ in this simulation. Figure \ref{figure-var-loca} shows the effect on the power curve. Remarkably, we find that an identical change point in the coefficient and variance, i.e., both $k^*$ and $m^*$ occur in the middle of the sample, will reduce some detecting power, while diverging $m^*$ from $k^*$ will increase the detecting power. Hence, the results displayed in Section \ref{sec-hetesim}--\ref{sec-ho} are based on the lowest detecting power given $k^* = m^* = 0.5$. This effect gets severe when the model is generated with heteroscedastic {\bf AR} and {\bf GARCH} errors. To compare the weighted CUSUM tests $V_N^{HET}(0)$, $Q_N^{HET}(0)$ with $V_N^{HET}(1/2)$ and $Q_N^{HET}(1/2)$, we find the Darling--Erd\H{o}s type tests suffer less power reduction effect than the weighted CUSUM tests.

\begin{figure}[H]
	\centering
	\caption{Rejection rates of $V_N^{HET}(0)$, $Q_N^{HET}(0)$, $V_N^{HET}(1/2)$ and $Q_N^{HET}(1/2)$ at 95\% significance level with sample size $N=250$ (first row) $N=500$ (second row). The error term follows heteroscedastic \textbf{Normal}, \textbf{AR} and \textbf{GARCH} errors. The change point $k^*= \lfloor 0.5N  \rfloor$ and the size of change point $\delta=0.6$. X-axis labels the location of variance change $m^*$.}
	\label{figure-var-loca}
	\includegraphics[width=14.5cm]{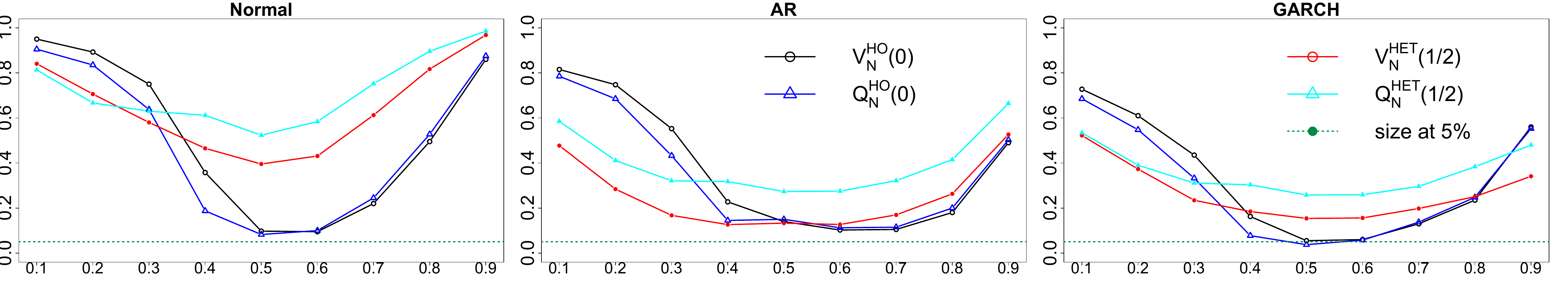}
	\includegraphics[width=14.5cm]{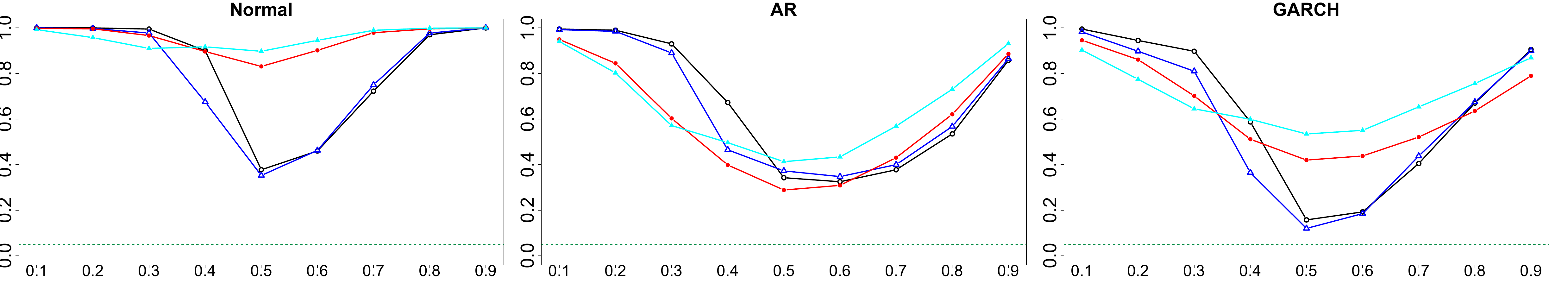}
\end{figure}

The third experiment investigates the effect of size of variance change on detecting power. We consider the models generated with heteroscedastic covariates and errors discussed in Section 4.3. But here, we allow the variance changes with a multiplier ratio $r$. Thus, the data generating process for the errors are:

(vii) (\textbf{Normal}) the error terms $\epsilon_i$ are i.i.d normal random variables:
\begin{equation*}
\epsilon_i = \left\{\begin{matrix}
N(0,0.5r),\quad 1\leq i\leq \lfloor m^* N  \rfloor,
\\
N(0,2r), \quad \lfloor m^* N  \rfloor<i\leq N.
\end{matrix}\right.
\end{equation*}

(viii) (\textbf{AR}) the error terms $\epsilon_i$ follow autoregressive (AR-1) process:
\begin{equation*}
\epsilon_i = \left\{\begin{matrix}
0.5 \epsilon_{i-1} + \varepsilon^1_i,\;\;  \mbox{with } \varepsilon^1_i \sim N(0,0.5r), \;\;  1\leq i\leq \lfloor m^* N  \rfloor,
\\
0.5 \epsilon_{i-1} + \varepsilon^2_i,\;\;  \mbox{with } \varepsilon^2_i \sim N(0,2r), \;\; \lfloor m^* N  \rfloor<i\leq N.
\end{matrix}\right.
\end{equation*}

(ix) (\textbf{GARCH}) the error terms $\epsilon_i$ follow a stationary GARCH(1,1) process defined by
\begin{equation*}
\epsilon_i = h_i^{1/2} \varepsilon_i, \mbox{ with }  h_i =
\left\{\begin{matrix}
0.5r + 0.01 \epsilon_i^2 + 0.9 h_{i-1},\;\;1\leq i\leq \lfloor m^* N  \rfloor,
\\
2r + 0.01 \epsilon_i^2 + 0.9 h_{i-1}, \;\; \lfloor m^* N  \rfloor<i\leq N.
\end{matrix}\right.
\end{equation*}
and the $\varepsilon_i$'s are i.i.d. standard normal random variables.

\begin{figure}[H]
	\centering
	\caption{Rejection rates of $V_N^{HET}(0)$, $Q_N^{HET}(0)$, $V_N^{HET}(1/2)$ and $Q_N^{HET}(1/2)$ at 95\% significance level with sample size $N=250$ (first row) $N=500$ (second row). The error term follows heteroscedastic \textbf{Normal}, \textbf{AR} and \textbf{GARCH} errors. The change point $k^*= \lfloor 0.5N  \rfloor$ and the size of change point $\delta=0.6$. X-axis labels the ratio $r$.}
	\label{figure-noise}
	\includegraphics[width=14.5cm]{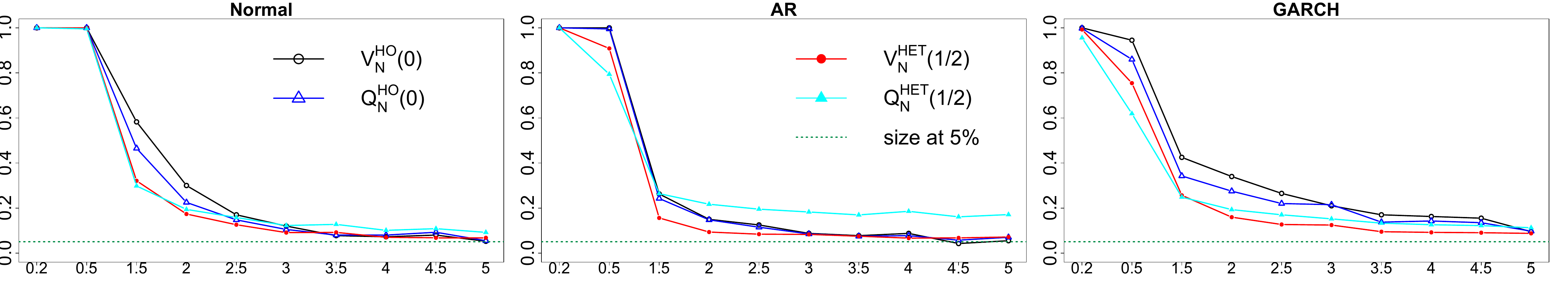}
	\includegraphics[width=14.5cm]{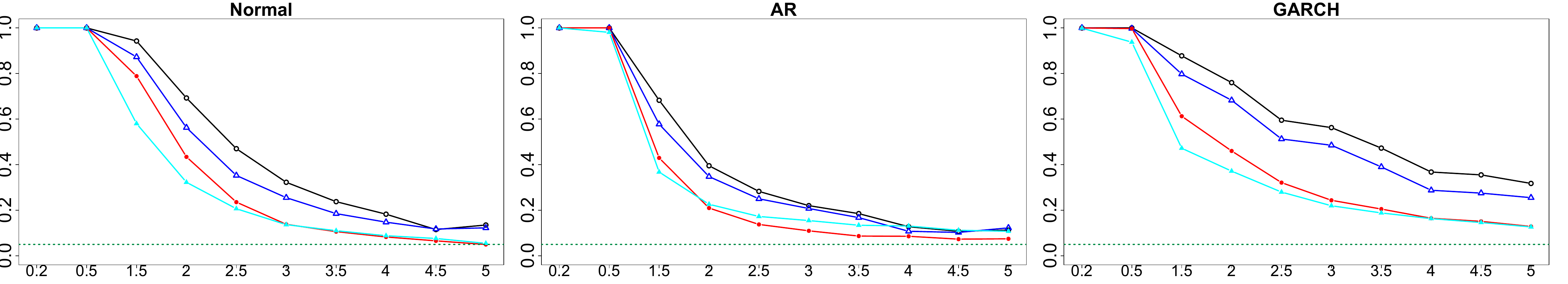}
\end{figure}

Figure \ref{figure-noise} displays the power effect when $r\in[0.2, 0.5, 1.5, 2, 2.5, 3, 3.5, 4, 4.5, 5]$. We set the locations $k^*=m^*=0.5$. Clearly, by fixing $\delta=0.6$, the testing power drop dramatically and nearly lose power for all tests when the ratio $r =2$, i.e., the variance reaches $4$ for the later phase. We attribute this to the effect of the signal-to-nose ratio. Given a large noise, the signal, i.e., the size of change point $\delta$ needs to be enlarged accordingly to gain reasonable power.

\end{document}